\pdfoutput=1
\documentclass[11pt,fleqn]{article}
\usepackage{xcolor}
\usepackage[english]{babel}
\usepackage[utf8]{inputenc}
\usepackage[T1]{fontenc}
\usepackage{lmodern}
\usepackage{amssymb}
\usepackage{mathtools}
\usepackage[lmargin=1in,rmargin=1in,bottom=1in,top=1in,twoside=False]{geometry}
\usepackage{enumitem}
\usepackage{amsfonts}
\usepackage{amsmath}
\usepackage{amsthm}
\usepackage{complexity}
\usepackage{mathrsfs}
\usepackage{graphics}
\usepackage[pagewise]{lineno}
\usepackage{microtype}
\usepackage[defaultlines=4,all]{nowidow}
\usepackage[many]{tcolorbox}
\tcbuselibrary{theorems}
\usepackage[colorinlistoftodos,prependcaption,textsize=scriptsize]{todonotes}
\usepackage{xspace}
\usepackage{makecell}
\usepackage{pifont}
\usepackage{bm}
\usepackage{scalerel}
\usepackage{tabularx}
\usepackage{afterpage}
\usepackage[hidelinks]{hyperref}
\usepackage[nameinlink]{cleveref}
\usepackage{thm-restate}

\definecolor{myPurple}{RGB}{230,230,255}
\definecolor{interesting}{RGB}{140,0,60}
\definecolor{brown}{RGB}{210,180,140}

\tikzset{
  paramRed/.style={
    draw, rectangle, fill=gray!20, font=\small, rounded corners=3pt,
    minimum width=2cm, align=center},
  paramGreen/.style={
    draw, rectangle, fill={rgb,255:red,255; green,230; blue,204}, font=\small, rounded corners=3pt,
    minimum width=2cm, align=center},
  paramHalf/.style={
    draw, rectangle,
    minimum width=2cm, align=center, font=\small, rounded corners=3pt,
    path picture={
      \begin{scope}[sharp corners]
        \fill[gray!20]
          ($ (path picture bounding box.south west)!0.5!(path picture bounding box.north west) $)
          rectangle
          (path picture bounding box.north east);
        \fill[draw=none, fill=white]
          (path picture bounding box.south west)
          rectangle
          ($ (path picture bounding box.south east)!0.5!(path picture bounding box.north east) $);
      \end{scope}
    }
  }
}

\usepackage{subcaption}

\usetikzlibrary{calc,arrows, automata, fit, shapes, arrows.meta, positioning, decorations.pathmorphing, patterns, shapes.geometric, shapes.symbols, shapes.arrows, shapes.misc, decorations.shapes,decorations.pathreplacing}
\usepackage{float}

\tcolorboxenvironment{mytheorem}{
    colframe=orange, sharp corners=west, 
    colback=orange!10,
	rightrule=0pt,
	toprule=0pt,
	bottomrule=0pt,
	top=0pt,
	right=0pt,
	bottom=0pt,
 }
 
\tcolorboxenvironment{mylemma}{
    colframe=yellow, sharp corners=west,
	colback=yellow!10,
	rightrule=0pt,
	toprule=0pt,
	bottomrule=0pt,
	top=0pt,
	right=0pt,
	bottom=0pt,
 }

\renewcommand{\epsilon}{\varepsilon}
\renewcommand{\phi}{\varphi}

\newcommand{\Omc}{\ensuremath{\mathcal{O}}\xspace}

\newcommand{\Cc}{\mathscr{C}}
\newcommand{\Oof}{\Omc}

\newcommand{\N}{\mathbb{N}}

\usepackage{xspace}

\newcommand{\clique}{\textsc{Clique}\xspace}

\newcommand{\sol}{\texttt{Sol}}
\newcommand{\soli}{\sol_i}
\newcommand{\solj}{\sol_j}
\newcommand{\solh}{\sol_h}

\newcommand{\XNLP}{\ensuremath{\mathsf{XNLP}}}

\theoremstyle{remark}
\newtheorem{theorem}{Theorem}[section]

\newtheorem{definition}{Definition}[section]
\newtheorem{lemma}{Lemma}[section]

\newtheorem{claim}{Claim}
\Crefname{corollary}{Corollary}{Corollaries}
\Crefname{lemma}{Lemma}{Lemmas}
\Crefname{section}{Section}{Sections}

\newtheorem*{result*}{}
\newtheorem*{remark*}{Remark}

\newcommand{\MSO}{\textsf{MSO$_1$}\xspace}
\newcommand{\MSOT}{\textsf{MSO$_2$}\xspace}
\newcommand{\MSOD}{\textsf{MSO$_1$}-\textsc{Discovery}}
\newcommand{\MSOTD}{\textsf{MSO$_2$}-\textsc{Discovery}}
\newcommand{\FOD}{\textsf{FO}-\textsc{Discovery}}

\renewcommand{\phi}{\varphi}
\renewcommand{\emptyset}{\varnothing}

\begin{document}
\title{
On Algorithmic Meta-Theorems for Solution Discovery:\\ Tractability and Barriers}
\author{Nicolas Bousquet\thanks{Partly supported by ANR project ENEDISC (ANR-24-CE48-7768-01).}\\Université Claude Bernard Lyon 1, France \and
Amer E. Mouawad\\American University of Beirut, Lebanon \and 
Stephanie Maaz\thanks{Funded by a grant from the Natural Sciences and Engineering Research Council of Canada.}\\University of Waterloo, Canada \and
Naomi Nishimura\footnotemark[1]\\University of Waterloo, Canada \and 
Sebastian Siebertz\\ University of Bremen, Germany}
\date{}
\maketitle

\begin{abstract}
\noindent Solution discovery asks whether a given (infeasible) starting configuration to a problem can be transformed into a feasible solution using a limited number of transformation steps.
This paper investigates meta-theorems for solution discovery for graph problems definable in monadic second-order logic (\MSO and \MSOT) and first-order logic (\FO) where the transformation step is to slide a token to an adjacent vertex, focusing on parameterized complexity and structural graph parameters that do not involve the transformation budget $b$.
We present both positive and negative results. 
On the algorithmic side, we prove that \MSOTD\ is in \XP\ when parameterized by treewidth and that \MSOD\ is fixed-parameter tractable when parameterized by neighborhood diversity. 
On the hardness side, we establish that \FOD\ is \W[1]-hard when parameterized by modulator to stars, modulator to paths, as well as twin cover, numbers. 
Additionally, we prove that \MSOD\ is \W[1]-hard when parameterized by bandwidth.
These results complement the straightforward observation that solution discovery for the studied problems is fixed-parameter tractable when the budget $b$ is included in the parameter (in particular, parameterized by cliquewidth$+b$, where the cliquewidth of a graph is at most any of the studied parameters), and provide a near-complete (fixed-parameter tractability) meta-theorems investigation for solution discovery problems for \textsf{MSO}- and \FO-definable graph problems and structural parameters larger than cliquewidth.    
\end{abstract}

\section{Introduction}
\label{sec:intro}
Real-world systems rarely remain static. 
From automated warehouses adapting to shifting demand patterns~\cite{elementlogic2024} to quantum circuits reallocating qubits for optimal computation~\cite{koch2007charge,DBLP:conf/cgo/SiraichiSCP18,DBLP:journals/pacmpl/SiraichiSCP19}, systems must continuously transform their configurations through sequences of local modifications. 
This dynamic nature of real systems motivates the study of transformation problems, where we seek to change one configuration into another through a sequence of permitted operations.

Over the last few decades, considerable effort has focused on finding efficient transformations between solutions. 
For example, bounded transformation problems under the \emph{combinatorial reconfiguration} framework~\cite{nishimura2018introduction,van2013complexity} take as input two {\em configurations} (for example, solutions to a problem instance, such as independent sets of a graph) and determine whether the first can be transformed into the second via a bounded length sequence of small, local modifications, while maintaining feasibility at every intermediate step.
For vertex subset problems, such as \textsc{Independent Set}, we can view a configuration as a set of tokens placed on the vertices of a solution.
Each reconfiguration step corresponds to a simple change, such as adding, deleting, or moving a token.
The solution space and the complexity of deciding the existence of a transformation depend on the allowed modification rules such as in the \emph{token jumping model} where, at each step, exactly one token can be moved to any unoccupied vertex in the graph and in the \emph{token sliding model} where, at each step, exactly one token moves along an edge of the graph to an unoccupied vertex.

While the reconfiguration framework captures situations in which, starting with a solution, we need each transformation to result in another solution, many practical scenarios present a different challenge.
For instance, urban bike-sharing systems require the distribution of bicycles every few hours to match demand patterns at different stations.
When commuters leave many stations empty and others overflowing, the current distribution becomes invalid (that is, it fails to meet minimum inventory requirements).
Rebalancing trucks must then transform this infeasible configuration into any valid one while minimizing travel distance to reduce operational costs.
The key insight is that during redistribution, stations are allowed to have too few or too many bikes; what matters is reaching as efficiently as possible any final configuration where all minimum requirements are satisfied.
Similarly, in many other applications, if a system becomes corrupted with an invalid configuration, maintaining feasibility during transformation is not required, as it may often be permissible or even preferable to shut down the system during the transformation. 
We then simply wish to obtain a valid solution with minimum transformation cost (that is, time or other resources), preferring efficient restoration over maintaining functionality throughout the transformation.
Inspired by this observation, Fellows et al.~\cite{fellows2023solution} introduced \emph{solution discovery} where instead of seeking a feasibility-preserving transformation between two known solutions, we begin with one known infeasible (corrupted) configuration and ask whether a solution can be reached via a short transformation sequence.
Although similar ideas have been studied in various fields from graphs to geometry with discrete or continuous optimization~\cite{DBLP:conf/soda/DemaineHMSOZ07,DBLP:journals/jco/AnariFGS16}, solution discovery formalizes and further generalizes the approach.
The goal of this paper is to investigate meta-theorems for a class of solution discovery problems on graphs. 

Although the application of solution discovery extends far beyond graphs, in this paper, we confine our focus to the following problem: (fix/given) a vertex subset property $\Pi$ of graphs, given a graph $G$, an initial infeasible configuration $S\subseteq V(G)$, and an integer budget~$b$, does there exist a final feasible configuration $T \subseteq V(G)$ that satisfies $\Pi$ on $G$ such that there exists a transformation sequence of length at most~$b$ from~$S$ to $T$? 
To further specify the transformation sequence, we represent $S$ by placing a token on each vertex of $S$, and (mainly) consider transformation sequences where a transformation step consists of a token slide. 

Unlike problems in reconfiguration (where both endpoints of the reconfiguration sequence are part of the input and intermediate steps are solutions), solution discovery starts with an infeasible configuration, introduces existential quantification over the final configuration, and does not require intermediate configurations to be solutions.
These subtle changes significantly alter the computational landscape. 
Many reconfiguration problems are \PSPACE-complete~\cite{hearn2005pspace,wrochna2018reconfiguration}; moreover, the requirement that all configurations in the sequence be solutions can lead to exponential length transformations~\cite{BousquetDPT23,DBLP:journals/dam/BonamyDO21}. 
In sharp contrast, in solution discovery the number of steps needed to reach a final configuration, if one exists,  
is at most $kn$ steps, where $k$ is the number of tokens, or, in other words, the solution or configuration size, and~$n$ is the number of vertices of the instance graph.
Hence, the transformation sequence can serve as a polynomial-size certificate, demonstrating that the problems lie in \NP. 
Not surprisingly, the solution discovery variants of common \NP-complete problems (where~$\Pi$ is defined through the classical problem and the final configurations must satisfy the standard problem definition) such as \textsc{Vertex Cover, Independent Set}, and \textsc{Dominating Set}, are \NP-complete~\cite{fellows2023solution,DBLP:conf/esa/DemaineHM09}. 
More surprisingly, the solution discovery variants of polynomial-time solvable problems such as \textsc{Vertex Cut, Matching}, and \textsc{Shortest Path} are also \NP-complete~\cite{GroblerMMMRSS24}. 
A fine-grained parameterized classification of the above problems with respect to the parameters $k$ (number of tokens), $b$ (discovery budget), and structural parameters such as treewidth, pathwidth, and feedback vertex set number was established in several papers~\cite{DBLP:conf/esa/DemaineHM09,fellows2023solution, GroblerMMMRSS24, grobler2024kernelization}, to which we refer the reader for more details on solution discovery and related work.

\subsection{Meta-theorems for reconfiguration and solution discovery}
Where possible, instead of individual results on a wide variety of problems and settings, we seek unifying frameworks, so-called \emph{algorithmic meta-theorems}, that can systematically capture the tractability of problems. 
Algorithmic meta-theorems are often phrased in terms of logic (describing the type of problems) and structural parameters (where they can be solved efficiently). 
The logics most commonly studied in this context are first-order logic (\FO), monadic second-order logic \MSO, which extends \FO\ with quantification over vertex subsets, and monadic second-order logic \MSOT, which extends \MSO by quantification over edge subsets. 
For example, Courcelle’s famous meta-theorem states that every problem definable in \MSOT can be solved in linear time on every fixed class of graphs with bounded treewidth~\cite{courcelle1990monadic} (subsequently extended for~\MSO to the more general classes of bounded cliquewidth~\cite{courcelle2000linear}). 
These results led to a new field in optimization, trying to understand the limits of tractability depending on graph classes and logics; we refer the reader to the surveys~\cite{kreutzer2008algorithmic,SiebertzV24}. 
All formal definitions of the logics and the graph parameters relevant for our results are given in the preliminaries. 

Courcelle's theorem traditionally addresses formulas with no \emph{free variables}, that is, formulas that make complete, self-contained statements about whether a graph has a certain property, for example, whether it is $3$-colorable, Hamiltonian, etc.
A variable is `free' if it can appear in a formula without being quantified (bound by $\exists$ or $\forall$), meaning its value isn't specified within the formula itself (just like the graph vertex and edge sets).
In reconfiguration and solution discovery, we need to describe a property of a specific vertex set within a graph (that is, the set of vertices, each containing one token).
To do this, we use a formula~$\phi(X)$ with one free set variable $X$ that represents the vertex set that we are examining against the property described by the formula.
Given a graph $G$ and a specific set $T \subseteq V(G)$, we say that $T$ is a \emph{solution} if substituting $T$ for $X$ makes the formula true in~$G$.

This approach allows us to express many vertex-subset problems, such as \textsc{Independent Set}, \textsc{Vertex Cover}, and \textsc{Dominating Set}, using \FO.
Others, such as \textsc{Vertex Cut}, require the additional expressive power of \MSO.
Importantly, such a formula specifies the property that a set must have, not the constraints on its size.

Courcelle's meta-theorem immediately yields a meta-theorem for combinatorial reconfiguration, as observed by Mouawad et al.~\cite{mouawad2014reconfiguration}. 
If a problem (that is, its feasible solutions) can be defined in~\MSO, then the reconfiguration variants of the problem are fixed-parameter tractable (\FPT) with respect to the parameter $b+tw$, where~b is the length of the transformation sequence and~$tw$ is the treewidth of the input graph. 
Like Courcelle's theorem, this result unifies the tractability of a broad range of problems and modification rules, holding for \MSOT and parameter $b+tw$ as well as for \MSO and parameter $b+cw$, where $cw$ denotes the cliquewidth of the input graph (which is at most its treewidth) and using the results of Courcelle et al.~\cite{courcelle1990monadic,courcelle2000linear}.
Similarly, it is easy to see that if a problem is \FO-definable, then its reconfiguration variants are fixed-parameter tractable with respect to the parameter $b$ on every class of graphs that admits efficient \FO\ model checking~\cite{bousquet2024survey}. 
\emph{Model checking} is the task of determining whether a given formula is true in a given graph.
These results hold for both token sliding and token jumping and follow from the fact that the existence of reconfiguration sequences can be expressed as formulas of the respective logics. 

For solution discovery, we observe that if a problem is \MSOT-definable, then its solution discovery variants (in the token jumping and token sliding models) are fixed-parameter tractable when parameterized by $b+tw$, 
if a problem is \MSO-definable, then its solution discovery variants are fixed-parameter tractable when parameterized by $b+cw$, and if a problem is \FO-definable, then its solution discovery variants are fixed-parameter tractable parameterized by $b$ on every class that admits efficient \FO\ model checking. 
This is possible, for example, on nowhere dense graph classes~\cite{grohe2014deciding}, including all classes with excluded minors or excluded topological minors. 
These results again follow simply from the general results of model checking for the respective logics and the observation that we can write formulas that quantify transformation sequences with lengths depending only on $b$. 
We state these observations as theorems for future reference (\Cref{thm:mso-discovery-b} and~\Cref{thm:fo-discovery-b}). 

Wrochna~\cite{wrochna2018reconfiguration} revealed the critical role that the parameter $b$ plays in the tractability results of reconfiguration problems for general graph classes.
He showed that if the sequence length is not included as a parameter, then the reconfiguration variants of several \FO-definable problems become \PSPACE-complete even on extremely restricted graph classes, namely, on graphs with constant bandwidth (a restrictive subclass of bounded treewidth graphs). 
Undaunted by this result, Gima et al.~\cite{gima2024algorithmic} studied the possibility for meta-theorems for reconfiguration variants of \MSO-definable problems, with respect to parameters that do not involve~$b$, and which are large with respect to cliquewidth (but incomparable to bandwidth). 

\section{Main results}
\label{sec:main}
In the spirit of the implications of the work of Wrochna~\cite{wrochna2018reconfiguration}, Grobler et al.~\cite{grobler2024kernelization} showed that some solution discovery variants of several \FO-definable problems become \XNLP-hard, and hence unlikely to be in \FPT\ when parameterized by treewidth alone.
Our first positive result is an \XP\ algorithm for \MSOTD\ when parameterized by treewidth, which by the above hardness result is the best we can hope for for this parameter.  
Thus, we prove the following.

\begin{restatable*}{theorem}{thmXP}
\MSOTD\ is in \XP\ when parameterized by treewidth.
\end{restatable*}

Formulating solution discovery as a formula (that is, incorporating the movement of the tokens) and directly applying (an optimization variant of) Courcelle's theorem would only give a running time of $2^{O(tw \cdot k)}$. 
We thus have to go into the details of the dynamic programming to carefully keep track of the token movements, which allows us to reduce the exponential dependence to $O(k^{tw})$. 
The details are presented in \Cref{sec:treewidth}.

Note that this result contrasts with results on reconfiguration problems,  which are known to be \PSPACE-complete when restricted to graphs of constant treewidth (and even bandwidth)~\cite{BousquetDMMP25,wrochna2018reconfiguration}.
It is unclear whether \XP algorithms for \MSOD\ exist when the parameter is the cliquewidth of the graph. 
We leave this question open. 

\medskip
Inspired by the approach of Gima et al.~\cite{gima2024algorithmic} we then initiate an investigation of (fixed-parameter tractability) meta-theorems for the solution discovery variants of \MSO- and \FO-definable problems in the token sliding model with respect to structural graph parameters larger than cliquewidth.
We demonstrate that solution discovery is hard for some \FO-definable problems with respect to very large parameters and that the design of polynomial-time algorithms, when they exist for restricted graph classes, is far from simple.
Our main results are depicted in~\Cref{fig:results}. We obtain the following positive result:

\begin{figure}
    \centering
    \resizebox{0.8\textwidth}{!}{
      \begin{tikzpicture}[>=latex, thick, every node/.style={transform shape}]
        \node[paramGreen] (vc)        at (-3,1.5) {vertex cover};
        \node[draw, rectangle, font=\small, rounded corners=3pt, minimum width=2cm, align=center]  (maxleaf)   at (2,2) {max leaf number};
        \node[paramGreen] (nd)        at (-4,3.75) {neighborhood \\ diversity};
        \node[paramRed, draw, rectangle, font=\small, rounded corners=3pt, minimum width=2cm, align=center] (twc)       at (-7,3.5) {twin cover};
        \node[paramRed]   (mstars)    at (-1,3.5) {modulator \\ to stars};
        \node[paramRed]   (mpaths)    at (4.5,3.75) {modulator \\ to paths};
        \node[paramHalf, draw, rectangle, font=\small, rounded corners=3pt, minimum width=2cm, align=center]  (bw)        at (1.5,3.5) {bandwidth};
        \node[paramRed] (mcluster)  at (-7,5) {cluster deletion};
        \node[paramRed]   (mwidth)    at (-4,5) {modulator width};
        \node[paramRed]   (td)        at (-1,4.5) {treedepth};
        \node[paramRed]   (fvs)       at (4,5.25) {feedback \\ vertex set};
        \node[paramRed]   (pwidth)    at (0,5.5) {pathwidth};
        \node[paramRed]   (cw)        at (-4,7) {cliquewidth};
        \node[paramRed]   (tw)        at (1.5,6.5) {treewidth};
        \draw[->] (vc)      -- (nd);
        \draw[->] (vc)      -- (twc);
        \draw[->] (vc)      -- (mpaths);
        \draw[->] (twc)     -- (mcluster);
        \draw[->] (twc)     -- (mwidth);
        \draw[->] (nd)      -- (mwidth);
        \draw[->] (mcluster)-- (cw);
        \draw[->] (mwidth)  -- (cw);
        \draw[->] (tw)      -- (cw);
        \draw[->] (maxleaf) -- (mpaths);
        \draw[->] (mstars)  -- (td);
        \draw[->] (mstars)  -- (fvs);
        \draw[->] (td)      -- (pwidth);
        \draw[->] (maxleaf) -- (bw);
        \draw[->] (mpaths)  -- (fvs);
        \draw[->] (mpaths)  -- (pwidth);
        \draw[->] (bw)      -- (pwidth);
        \draw[->] (fvs)     -- (tw);
        \draw[->] (pwidth)  -- (tw);
        \draw[->] (vc)  -- (mstars);
      \end{tikzpicture}
    }
    \caption{Our \MSO and \FO\ results: beige boxes indicate tractability of \textsc{\MSO-Discovery} with respect to the corresponding parameter, gray boxes indicate hardness of \textsc{\FO-Discovery}, and white boxes indicate open problems.
    For parameter bandwidth, we have a hardness result only for \textsc{\MSO-Discovery}; the complexity of \textsc{\FO-Discovery}  remains an open problem.
    An arrow between two boxes indicates the existence of a function in the parameter above that is a lower bound on the parameter below.} 
    \label{fig:results}
\end{figure}
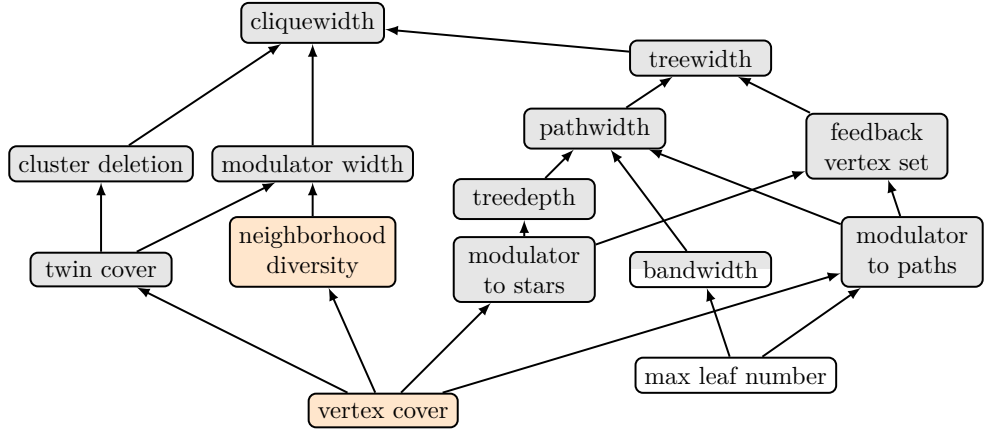

\begin{restatable*}{theorem}{thmND}\label{thm:logic-solution-discovery-neighborhood-diversity-fpt}
    \MSOD\ is in \FPT{} when parameterized by neighborhood diversity.
\end{restatable*}

The proof of our result depends critically on the concept of ``shapes'' of final configurations~\cite{DBLP:journals/lmcs/KnopKMT19}, which limits the number of possible shapes to a function of the formula and the neighborhood diversity number, combined with a result of Gima et al.~\cite{gima2024algorithmic} showing that all configurations sharing the same shape are feasible or infeasible.
We will prove that, using a flow algorithm, for each shape, we can find the shortest transformation to a solution of this shape (if any such transformation exists) in polynomial time. 
The proof of this result can be found in \Cref{sec:nd}.

We then complement this positive result by providing several hardness results. We first prove that the following holds:

\begin{restatable*}{theorem}{thmMTS}\label{thm:logic-solution-discovery-modulator-stars-hardness}
    \FOD\ is \W[1]-hard when parameterized by modulator to stars number.   
\end{restatable*}

\begin{restatable*}{theorem}{thmMTP}\label{thm:logic-solution-discovery-modulator-paths-hardness}
     \FOD\ is \W[1]-hard when parameterized by modulator to paths number.   
\end{restatable*}

These two results follow from two variations of the same reduction from \textsc{Multicolored Clique}. 
To avoid long and repetitive hardness proofs, we combine some of the arguments by abstracting out the commonalities. 
The proofs of these results are deferred to \Cref{sec:logic-solution-discovery-modulator-hardness}.

We then prove that the \FOD\ problem is also hard parameterized by twin cover, using a reduction from the \textsc{Planar Arc Supply} problem. In \textsc{Planar Arc Supply}, each arc is assigned a list of pairs of values (one value in each pair for each endpoint of the arc); the goal is the selection of a pair for each arc resulting in values being sent to endpoints that sum to specified demands. In our reduction, values are represented as cliques; a numerical encoding scheme is used to ensure that values are taken from a single pair.

\begin{restatable*}{theorem}{thmTC}\label{thm:logic-solution-discovery-twincover-hardness}
    \FOD\ is \textsf{W[1]}-hard when parameterized by twin cover.  
\end{restatable*}

Finally, we present our last hardness result in 
\Cref{sec:bandwidth}, which also entails a reduction from \textsc{Planar Arc Supply}, and which is the most technical contribution of the paper. 

\begin{restatable*}{theorem}{thmBW}\label{thm:logic-solution-discovery-bandwidth-hardness}
\MSOD\ is \W[1]-hard with respect to parameter bandwidth. 
\end{restatable*}

In the reduction, values in pairs correspond to numbers of tokens stored in arc-gadgets, and the selection of pairs for arcs follows a delicate choreography, resulting in tokens being moved to demand nodes in each vertex-gadget.

Note that, in contrast to the other hardness results, here we obtain hardness only for an \MSO- and not \FO-definable problem.
We believe that the hardness result can be extended to \FO\ and conjecture that \FOD\ parameterized by bandwidth is not in \FPT.
We also leave open the complexity of both \FOD\ and \MSOD\ parameterized by maximum leaf number. 
Even though this graph parameter seems very simple, the design of an \FPT\ algorithm parameterized by maximum leaf number is surprisingly challenging. 
Nevertheless, we conjecture that such an algorithm exists.

\section{Other Background and Related Work}
\label{sec:backgroundandrelatedwork}
Before the introduction of solution discovery for graph problems, a substantial body of work addressed related movement minimization problems in various settings. These problems, while not framed in the language of the discovery token sliding model, share the fundamental characteristic of transforming an initial configuration into a final configuration while minimizing some measure.

Demaine et al.~\cite{DBLP:conf/soda/DemaineHMSOZ07} established foundational results for movement minimization in discrete settings.
The authors focused mainly on final configurations satisfying the following properties: connectivity, independence, and perfect matchability, while minimizing the number, the sum, or the maximum of movements. 
Their movement minimization problems require that tokens induce a subgraph that satisfies the property, which means allowing multiple tokens to occupy a vertex in the final configuration unless this contradicts the property (for example, independence).
Their paper also introduces the \textsc{Mobile Facility Location Problem (MFLP)} in which existing facilities must be relocated and clients assigned to minimize the (weighted) sum of movement and service costs.
Significant progress in algorithms and heuristics have been made on this
problem and its variants in the uncapacitated~\cite{DBLP:conf/focs/FriggstadS08,DBLP:journals/cor/HalperRS15,DBLP:conf/soda/AhmadianFS13,DBLP:journals/jco/ArmonGS14}, capacitated~\cite{DBLP:journals/eor/RaghavanSS19} with a constraint of not more than one facility at a vertex in the final configuration, and even in the online setting~\cite{DBLP:journals/mst/GhodselahiK23,DBLP:conf/isaac/GhodselahiK15}.
While these and other papers~\cite{DBLP:conf/approx/BermanDZ11,DBLP:journals/tcs/BiloGLP16} focused on approximation and inapproximability results, Demaine et al.~\cite{DBLP:conf/esa/DemaineHM09} proved that with respect to the parameter number of tokens, the complexity of these graph problems depends on the treewidth of the minimal configurations.
The authors also studied the case where at most one token can end up on a vertex and obtained meta-theorems (with respect to the parameter number of tokens) for a subset of such problems.
Other results provide approximation and inapproximability results for these problems (with properties concerned with radius and distance) under explicit per token distance or energy constraints~\cite{DBLP:journals/tcs/BartschiBC0KM21}. 

Extending beyond discrete graph settings, Anari et al.~\cite{DBLP:journals/jco/AnariFGS16} and others~\cite{DBLP:conf/algosensors/BiloDGMPW13,DBLP:journals/dam/DumitrescuJ11,DBLP:conf/adhoc-now/CzyzowiczKKLNOSUY10,DBLP:journals/dam/0001NS20,DBLP:conf/aips/EnginI18} investigated these same movement problems in the geometric domain.
One specific example in the continuous setting is the problem \textsc{Spreading Points}, where points must be separated by a minimum distance while minimizing movement, which has been studied both classically~\cite{DBLP:conf/cccg/LiW15,DBLP:conf/cccg/GhadiriY16} and through parameterized complexity~\cite{DBLP:conf/esa/FominG00Z23}.
In the plane, Tekdas et al.~\cite{10.1007/978-3-642-05434-1_19} considered the problem of finding the final destinations of a subset of tokens that will form an $s$-$t$-path 
that minimizes the number of tokens between $s$ and $t$ and the distance traveled by the tokens. 

Note that in solution discovery, we do not care about the quality of partial solutions, and in reconfiguration, we continuously have a solution. 
A restrictive version of solution discovery coupled with the token jumping model, interesting on its own right, would be related to local search algorithms~\cite{deMeyerKLP25+,DBLP:journals/jcss/FellowsFLRSV12} tasked with moving from one configuration to another in the space of candidate configurations by applying local changes without increasing a potential function, until a solution deemed optimal is found or a time bound has elapsed.

\section{Preliminaries}
\label{sec:prelims}
\subsection{Colored graphs}
We use the symbol~$\N$ for the set of non-negative integers. 
For~$k \in \N$, we define~$[k] = \{1, \ldots, k\}$ with the convention that~$[0] = \varnothing$. 
We denote the set of vertices of a graph $G$ by $V(G)$ and the set of edges by $E(G)$,  using $n$ for $|V(G)|$ and $m$ for $|E(G)|$.
When we speak of a directed graph $D$, we use $V(D)$ to refer to its vertices and $A(D)$ to refer to its arcs.
With a slight abuse of notation, we often write $uv$ for the (undirected) edge
$\{u,v\}$ or for the (directed) arc $(u,v)$; the intended
meaning will always be clear from the surrounding statement.
For $X \subseteq V(G)$, we denote by $G[X]$ and $G - X$ the graphs induced by $X$ (that is, the graph with vertex set $X$ and edges of $E(G)$ with both endpoints in $X$) and $V(G) \setminus X$, respectively. 
We denote by $N(v)$ the set of vertices adjacent to $v$ and $N[v] = N(v) \cup \{v\}$.

Following standard practice in logic, we use \emph{colored graphs} where vertex colors encode unary predicates.
A colored graph is a tuple $(G,\Cc)$, where $\Cc=\{C_1,\ldots, C_c\}$ is a set of \emph{colors}, where each $C_i\subseteq V(G)$ is a subset of vertices. 
Note that a vertex may have multiple colors and for each $v \in V(G)$ we write $\Cc(v)$ to denote the set of colors that~$v$ belongs to. 
We sometimes omit information about colors (i.e., say graph $G$) when it is clear from context.

\subsection{Logic}
A \emph{free variable} is a variable that appears in a formula without being bound by a quantifier. 
For example, in the formula $\phi(X)$, the variable $X$ is free, whereas in $\exists X \subseteq V(G) : \phi(X)$, the variable~$X$ is bound.
The formulas of \emph{first-order logic} (\FO) over the language of colored graphs with one free vertex-set variable $X$ are constructed from the atoms $X(x)$, $C_i(x)$, where $x$ is a vertex variable, $E(x,y)$, where $x$ and $y$ are vertex variables, equality $x = y$ of vertex variables, and the usual connectives $\neg, \land, \lor, \rightarrow, \leftrightarrow$ and quantifiers $\exists, \forall$ operating on vertex variables. 
The atoms $X(x)$, $C_i(x)$, and $E(x,y)$ express, respectively, that vertex $x$ belongs to $X$,  vertex $x$ has color $C_i$, and  there is an edge between the vertices $x$ and $y$.
For \emph{monadic second-order logic} (\MSO), we also have atoms $Y(x)$, where~$Y$ is a vertex-set variable and $x$ is a vertex variable, and allow quantification of vertex-set variables and the atom represents that $x$ belongs to $Y$. 
For \MSOT we also have atoms~$Z(x,y)$, where~$Z$ is an edge-set variable and $x,y$ are vertex variables, representing that edge $xy$ belongs to~$Z$, and allow quantification over sets of edges. 

We often use syntactic sugar (for example, we write ``$\exists x \in Y : \psi$'' to mean ``$\exists x: Y(x) \wedge \psi$''). 
For a formula $\phi(X)$ with a free vertex-set variable $X$, a graph $G$ and a set $S\subseteq V(G)$ we write $G\models\phi(S)$ when $\phi$ is true for $G$ when the free set variable~$X$ is interpreted as $S$. 
In this work, we mainly consider formulas with one free set variable $X$ representing the vertices with tokens in the final configuration, although we occasionally use additional free vertex variables when needed for specific algorithmic purposes.
In that context, for a formula $\phi(X,x_1, \ldots, x_t)$ with a free set variable~$X$ and vertex variables $x_1, \ldots, x_t$, a graph $G$, a set $S \subseteq V(G)$ and vertices $v_1, \ldots, v_t$, we write $G\models\phi(S,v_1,\ldots,v_t)$ when $\phi$ is true for $G$ when the free set variable $X$ is interpreted as $S$ and the free vertex variables $x_1, \ldots, x_t$ are interpreted as $v_1,\ldots,v_t$, respectively. 
We call a formula without free variables a \emph{sentence}.
For more background on logic, we refer to~\cite{hodges1997shorter}. 
Since the number of colors referenced in a formula $\phi$ is bounded by $|\phi|$, which we will treat as a constant (or parameter), we assume that the number of colors $c$ is also a constant. 

\subsection{Solution discovery}
Let $\phi(X)$ be a formula with a free vertex-set variable~$X$, expressing a property of (colored) graphs. 
The input to the \textsc{$\phi$-Discovery} problem consists of a (colored) graph~$(G,\Cc)$, an \emph{initial configuration} $S \subseteq V(G)$ with $|S| = k$, and an integer parameter $b$. 
We view $S$ as a set of vertices each marked by a single \emph{token}; we can \emph{slide} a token from one vertex to an adjacent one in a single \emph{token slide} in the \emph{discovery token sliding model}. 
A token can move to an arbitrary (not necessarily adjacent) vertex in a \emph{token jump} in the \emph{discovery token jumping model}.
The \textsc{$\phi$-Discovery} problem asks whether we can, using at most $b$ token slides (or token jumps), move the tokens to a \emph{final configuration} $T \subseteq V(G)$ such that $G \models \phi(T)$. 
We view~$T$ similarly as a set of vertices each marked by a single token, that is, $|T| = k$. 
The sequence of token slides performed to get from an initial configuration~$S$ to a final configuration $T$ is called \emph{transformation sequence} $\vec{T}$.
We say that in a solution $(T, \vec{T})$ to a solution discovery problem, a token $t$ initially at a vertex $u \in S$ is \emph{destined} for a vertex $v \in T$ when~$\vec{T}$ includes a subsequence of token slides that slide the token from~$u$ to~$v$ and no other token slide in~$\vec{T}$ slides $t$. 

It is easy to see how in the solution discovery setting, the discovery token sliding (jumping) model is equivalent to its corresponding \emph{token sliding (jumping)} variant, where tokens cannot move to already occupied vertices.
However, for the sake of proof clarity, we employ the discovery token sliding model (which allows multiple tokens per vertex during the transformation but not in the initial and final configurations) throughout many of our proofs.
Henceforth, a \emph{configuration} is a placement of tokens on the vertices of the graph.

For a logic $\mathcal{L}$, in the $\mathcal{L}$-\textsc{Discovery} problem we require an $\mathcal{L}$-formula as an additional part of the input. 
Note that tractability of $\mathcal{L}$-\textsc{Discovery} implies tractability of $\mathcal{L}'$-\textsc{Discovery} when $\mathcal{L'}$ is weaker than $\mathcal{L}$, while, vice versa, the intractability of $\mathcal{L'}$-\textsc{Discovery} implies the intractability of $\mathcal{L}$-\textsc{Discovery}. 

\subsection{Parameterized complexity and structural graph parameters} 
We assume that the reader is familiar with parameterized complexity theory and refer to one of the standard textbooks for the basic definitions and more extensive background (for example, \cite{parameterizedalgorithms,downey2013fundamentals}). 

When parameterizing by treewidth plus the number of allowed token moves $b$, then \MSOTD\ both in the token sliding and in the token jumping model is in \FPT. 
When parameterizing by cliquewidth$+b$, then \MSOD\ both in the token sliding and in the token jumping model is in \FPT.
Although this is a straightforward observation, we present a small proof here for completeness.
To obtain these results, we can apply the theorem of Courcelle~\cite{courcelle1990monadic} and of Courcelle, Makowski and Rotics~\cite{courcelle2000linear} in combination with the approximation algorithm for cliquewidth from Korhonen and Soko{\l}owski~\cite{DBLP:conf/stoc/Korhonen024}.

\begin{theorem}[\cite{courcelle1990monadic}]\label{thm:courcelle}
    Let $\phi(X)$ be an \MSOT-formula on (colored) graphs. Then, given a graph~$G$ and $S \subseteq V(G)$ we can check whether $G\models\phi(S)$ in time $f(|\phi|, tw) \cdot n$, where $f$ is some computable function and $tw$ is the treewidth of $G$. 
\end{theorem}

\begin{theorem}[\cite{courcelle2000linear,DBLP:conf/stoc/Korhonen024}]\label{thm:courcelle-cw}
    Let $\phi(X)$ be an \MSO-formula on (colored) graphs. Then, given a graph~$G$ and $S \subseteq V(G)$ we can check whether $G\models\phi(S)$ in time $f(|\phi|, cw) \cdot n^2$, where $f$ is some computable function and $cw$ is the cliquewidth of $G$.
\end{theorem}

We note that the running time in this theorem can be improved to $f(|\phi|,cw)\cdot n^{1+o(1)}+\Oof(m)$, since clique expressions can be approximated in this time~\cite{DBLP:conf/stoc/Korhonen024},
Consequently, the running time in the following theorem can be improved to $g(|\phi|,b,cw) \cdot n^{1+o(1)}+\Oof(m)$, but such optimizations are not the focus of this work. 

\begin{theorem}\label{thm:mso-discovery-b}
    \MSOTD\ on colored graphs in the token sliding and the token jumping models can be solved in time $g(|\phi|,b,tw)\cdot n$, for some computable function $g$ and where $tw$ is the treewidth of the input graph. 
    \MSOD\ on colored graphs in the token sliding and the token jumping models can be solved in time $g(|\phi|,b,cw)\cdot n^{2}$, for some computable function~$g$ and where $cw$ is the cliquewidth of the input graph.  
\end{theorem}
\begin{proof}
    Given a formula $\phi(X)$ of \MSOT or \MSO, we write an \MSOT- or \MSO-formula~$\psi(X)$, respectively, stating the existence of (at most) $b$ sets\footnote{It is possible to either create one formula for each $b' \leq b$ and define $\psi$ as the disjunction of those formulas or, alternatively, we can  allow consecutive sets to be equal.} $X_1,\ldots, X_b$, where $X_1=X$, $X_{i+1}$ differs from~$X_i$ by the slide or jump of exactly one token (depending on the model), and stating that~$\phi(X_b)$ holds. 
    It is immediate by construction that $G\models\psi(S)$ if and only if a final configuration $T$ can be discovered in (at most) $b$ transformation steps starting from~$S$. 
    The length of the formula~$\psi$ depends only on~$\phi$ and $b$.
    Hence, applying \Cref{thm:courcelle} or \Cref{thm:courcelle-cw} we can test whether $G\models \psi(S)$ in time $f(|\psi|,tw)\cdot n=g(|\phi|,b,tw)\cdot n$, or 
    $f(|\psi|,cw)\cdot n^{2}=g(|\phi|,b,cw)\cdot n^2$, for some computable function $g$, respectively.
\end{proof}

We obtain a similar theorem for \FO. 

\begin{theorem}\label{thm:fo-discovery-b}
    Let $\phi(X)$ be an \FO-formula on (colored) graphs. 
    Let $\mathcal{G}$ be a class of (colored) graphs on which \FO\ model checking can be solved in time $f(|\phi|)\cdot n^{\Oof(1)}$, for some computable function~$f$. 
    Then \FOD\ in the token sliding and the token jumping models can be solved in time $g(|\phi|,b)\cdot n^{\Oof(1)}$ on $\mathcal{G}$, for some computable function $g$. 
\end{theorem}

\begin{proof}
    Similarly to the proof of~\Cref{thm:mso-discovery-b}, we quantify over the elements involved in the transformation sequence and verify that the solution reached satisfies the formula $\phi$.
    In particular, in the formula $\psi(X)$, we quantify over the (at most $b$) paths (of total length at most $b$) that tokens will slide on to reach the target solution. 
\end{proof}

\Cref{thm:fo-discovery-b} applies to several graph classes including nowhere dense classes~\cite{grohe2014deciding}, monadically stable classes~\cite{dreier2023first,dreier2024first}, and classes with bounded twinwidth~\cite{bonnet2021twin}, assuming that a contraction sequence certifying the twinwidth is given as part of the input. 

We now define the structural parameters used in our proofs.\smallskip

In a graph $G$, two vertices $u$ and $v$ are \emph{twins} if $N(u) = N(v)$ or $N[u] = N[v]$. 
The \emph{neighborhood diversity} of a graph $G$, denoted by $nd(G)$, is the number of subsets $V_1, \dots, V_d$ in the unique partition of~$V(G)$ into maximal sets of twin vertices. 
Neighborhood diversity and the corresponding partition of a graph can be computed in linear time \cite{DBLP:journals/dm/McConnellS99,tedder2008simpler}.

A \emph{twin cover} of a graph $G$ is a set $X \subseteq V(G)$ whose deletion leaves a disjoint union of cliques such that the vertices of each clique are twins in $G$.
The \emph{twin cover number} of $G$, denoted by $tc(G)$, is the minimum size of a twin cover of $G$ and can be computed in \FPT\ time with respect to~$tc(G)$~\cite{DBLP:conf/iwpec/Ganian11}. 

A \emph{modulator to paths} is a set of vertices $X \subseteq V(G)$ such that $G - X$ is a disjoint union of paths. 
The size of a minimum such modulator, denoted by $mp(G)$, is called \emph{modulator to paths number}. 
Similarly, a \emph{modulator to stars} is a set of vertices $X \subseteq V(G)$ such that $G - X$ is a disjoint union of stars. 
The size of a minimum such modulator, denoted by $ms(G)$, is referred to as \emph{modulator to stars number}.
It has been shown that one can compute a modulator to paths or stars of size~$k$ (if it exists) in \FPT\ time with respect to $k$ \cite{DBLP:journals/dam/NishimuraRT05}.

The \emph{bandwidth} of a graph $G$ is the minimum integer $bw(G)$ such that there exists a bijection $f: V(G) \rightarrow \{1, \ldots, |V(G)|\}$ such that for every edge $\{u,v\} \in E(G)$, $|f(u) - f(v)| \leq bw(G)$. 
Computing bandwidth of a graph $G$ is \W[1]-hard when parameterized by $bw(G)$~\cite{DBLP:conf/stoc/BodlaenderFH94}. 

A \emph{tree decomposition} of a graph $G$ is a pair $(\mathcal{T}, \mathcal{B})$ where $\mathcal{T}$ is a tree and $\mathcal{B} = \{B_i : i \in V(\mathcal{T})\}$ is a collection of subsets of $V(G)$ (called \emph{bags}) such that:
\begin{enumerate}
\item $\bigcup_{i \in V(\mathcal{T})} B_i = V(G)$ (every vertex appears in some bag),
\item for every edge $uv \in E(G)$, there exists $i \in V(\mathcal{T})$ with $\{u,v\} \subseteq B_i$ (every edge appears in some bag), and
\item for every vertex $v \in V(G)$, the set $\{i \in V(\mathcal{T}) : v \in B_i\}$ induces a connected subgraph of $\mathcal{T}$.
\end{enumerate}
The \emph{width} of a tree decomposition is $\max_{i \in V(T)} |B_i| - 1$. 
The treewidth of $G$, denoted~$tw(G)$, is the minimum width over all tree decompositions of $G$.
A \emph{nice tree decomposition} of $G$ is one whose root and leaves correspond to empty bags and such that each other node in the decomposition corresponds to a bag that either \emph{introduces} a vertex $v \in V(G)$ ($B_i = B_{i-1} \cup \{v\}$ for $v \not\in B_i$), \emph{forgets} one ($B_i = B_{i-1} \setminus \{v\}$ for $v \in B_i$), or is a \emph{join node} with two children $j$ and $h$ (where $B_i = B_j = B_h$). 
Computing the treewidth of a graph $G$ and the associated nice tree decomposition can be done in \FPT\ time when the parameter is $tw(G)$~\cite{DBLP:journals/siamcomp/Bodlaender96,parameterizedalgorithms}.

Finally, some section-specific preliminaries appear instead in their relevant sections.
For clarity, in reductions, we refer to \emph{vertices} in the original problem and \emph{nodes} in our constructions.
From here onward, all results pertain to the (discovery) token sliding model.

\section{\MSOT-Discovery Parameterized by Treewidth is in \XP}\label{sec:treewidth}
In this section, we show that \MSOTD\ is in \XP\ when parameterized by treewidth, or equivalently, that the discovery problem is polynomial-time solvable for every fixed $\MSOT$-definable property and fixed class of bounded treewidth.

Our proof is an extension of the proof of Courcelle's theorem via dynamic programming. 
Before we present our algorithm, let us explain why we cannot simply apply this theorem, and then give an intuition for our approach. 
We aim to compute by dynamic programming a final configuration $T\subseteq V(G)$ that satisfies $\phi(T)$ and such that $T$ can be reached from our initial configuration $S$ (which we assume is marked by a color in the input graph) with at most~$b$ token slides. 
One possible approach is to express the token movement in an extended \MSOT-formula (with multiple free set variables), where $k$ free set variables each represents the path for the movement of a token, and apply the optimization version of Courcelle's theorem~\cite{arnborg1991easy}, which allows us to find minimum weight sets satisfying a formula in \FPT\ time with respect to treewidth. 
Our extended \MSOT-formula would on a very high level be: 
$\psi(M_1,\ldots, M_k, T) := \text{``$T$ is the result of applying movements $M_1,\ldots, M_k$ to $S$''}\wedge \phi(T)$. 
The optimization version of Courcelle's theorem now permits us to find sets $M_1,\ldots, M_k$ and $T$ such that $\sum_{1\leq i\leq k}|M_i|$ is minimized, and we simply check if this sum is at most $b$. 
However, with this approach, the dependence on $k$ in the running time is too large; each free set variable adds one membership bit per bag vertex, and thus the dynamic programming state space is multiplied by about $2^{tw}$ for each such set variable. 
With the above formula, we get a running time of $2^{\Oof(tw\cdot k)}\cdot f(|\phi|)\cdot n^{\Oof(1)}$, which has an exponential dependence on~$k$.

In our proof, we follow the proof of Courcelle's theorem, but track the token movement shape and cost explicitly to avoid a blow-up in the time complexity.  
The intuitive reason why this is possible is that when multiple tokens pass the same vertex, it is not important which token continues its way in which direction; we only care about the number of tokens that pass in a certain direction. 
We are hence able to compress the information about the movement to a size of $\Oof(k)^{tw}$ in each of our dynamic programming tables. 
We additionally track the (achievable) and satisfiable logical properties in the subgraph induced by all bags in the subtree rooted at a bag (this follows standard techniques). 

Fix a formula $\phi(X)$ for which we want to solve $\phi$-\textsc{Discovery}, say of quantifier rank~$q$.  
Recall that the \emph{quantifier rank} of a formula is the nesting depth of the quantifiers in the formula. 
Also fix a nice tree decomposition $(\mathcal{T}, \{B_i\}_{i \in V(\mathcal{T})})$ of width~$tw$ of our input graph~$G$.  
Let $i$ be a node of $\mathcal{T}$. 
We write $G_i$ for the subgraph induced by the vertices of the bags $B_j$ for $j$ below~$i$ (including $i$) in $\mathcal{T}$. 

We intuitively view our dynamic program as a nondeterministic procedure that guesses the solution $T$. 
It goes from the bottom up in $\mathcal{T}$, that is, from the leaves to the root, and for each node~$i$, it will have guessed the partial final configuration $T_i=T\cap V(G_i)$. 
Formally, we will dynamically maintain all sets $T_i'\subseteq V(B_i)$ that can be extended to sets $T_i$ satisfying certain logical properties (the logical properties will be captured by logical types, which we will explain shortly) that can be reached with a certain token sliding cost. 
We will at the same time maintain, among other things (such as the number of tokens in $G_i \cap T$, the number of token sliding steps in the subproblem on~$G_i$), the number of superfluous tokens that must slide out of $G_i$ and the number of tokens that must slide into $G_i$.
To update the logical type of the sets~$T_i'$, we introduce additional variables to our logical properties that are initiated with the vertices of the bag $B_i$, so that we can capture all information about the newly introduced vertex at an introduce node. 

Let us formally introduce logical types. 
For a quantifier-rank $q$, a \emph{$q$-type} $\Phi(X,x_1, \ldots, x_t)$ is a set of formulas $\psi(X,x_1,\ldots, x_t)$ of quantifier rank at most $q$ such that there exists a graph $H$, $A \subseteq V(H)$ and $v_1,\ldots, v_t\in V(H)$ such that $H \models \psi(A,v_1,\ldots, v_t)$ for all $\psi\in \Phi$. 
We will often abbreviate tuples $x_1,\ldots, x_t$ and $v_1,\ldots, v_t$ by~$\bar x$ and~$\bar v$, respectively, and leave it to the context to determine the size of the tuples. 
We sometimes abuse notation and treat~$\bar v_i$ as a set; for example, we write $v\in \bar v_i$ to denote $v \in \{v_1,\ldots,v_t\}$.
When we write $\psi(X,\bar x)$ we mean that the free variables of $\psi$ are among~$X$ and~$\bar x$; not all variables of $\bar x$ actually have to appear in $\psi$. 
The \emph{$q$-type of $A$ and~$\bar v$ in a graph $H$} is the set $tp_q(H,A,\bar v)$ of all formulas of quantifier rank at most $q$ such that $H\models\psi(A,\bar v)$ for all $\psi\in \Phi$. 
We say that a $q$-type $\Phi$ is \emph{realized in $H$} if there are $A$ and~$\bar v$ with $tp_q(H,A,\bar v)=\Phi$. 

In the dynamic programming algorithm, when we are dealing with node $i$,
we will instantiate~$H$ with $G_i$ and $A$ with possible $T_i\subseteq V(G_i)$, and hence compute the realizable types $tp_q(G_i,T_i,\bar v_i)$, where $\bar v_i$ is the tuple of the vertices of $B_i$.
More precisely, we will consider $T_i'\subseteq B_i$ and compute all types $tp_q(G_i,T_i,\bar v_i)$ that are realized for some extension $T_i\supseteq T_i'$ of $T_i'$ with $T_i\subseteq V(G_i)$ (and that satisfy the additional token sliding cost constraints). 
Since up to equivalence there are only finitely many \MSOT-formulas of quantifier rank at most $q$ (and we can choose a normalized representative of each formula), $q$-types are finite, and we can represent a type by the set of normalized formulas it contains, see, for example, Proposition 7.5 of~\cite{libkin2004elements}. 
Thus, let us write $\Gamma(q,t)$ for the set of normalized $q$-types with~$t$ free variables and let $\gamma(q,t)=|\Gamma(q,t)|$, which is a computable function. 
In the following, when we speak of types, we will always mean $q$-types with at most $tw+1$ free variables (with size dependent on $q$ and $tw$), and by slightly abusing notation we now write $\gamma$ to refer to $\gamma(q,tw+1)$.

Recall that we assume that the root node $r$ of $\mathcal{T}$ corresponds to an empty bag.
Also, recall our treatment of free vertex variables, which in the root bag leads to types of the form $\Phi(X,\bar x)=\Phi(X)$. 
All sets $T_i$ that we consider throughout the dynamic process may satisfy very different logical properties, and in particular possibly do not satisfy $\phi$, but all that matters is that in the root bag in one of our states we have a set $T_r=T$ with $G_r=G\models\phi(T)$. 
This is the case if and only if $\phi(X)$ is contained in one of the types in the root bag. 

We collect one more technical lemma that we will require to update types. 

For tuples $\bar v=(v_1,\ldots, v_t)$ and $\bar w=(w_1,\ldots, w_\ell)$ we write $\bar v\bar w$ for their concatenation $(v_1,\ldots, v_t, \\w_1,\ldots, w_\ell)$. 
We use the following decomposition theorem due to Feferman and Vaught~\cite{Feferman1959} adapted from the form presented in~\cite{DBLP:conf/birthday/Grohe08}.

\begin{lemma}[\cite{DBLP:conf/birthday/Grohe08}, adapted]\label{thm:feferman-vaught}
Let $G_i,G_j$ be colored graphs, $T_i\subseteq V(G_i), T_j\subseteq V(G_j)$, $\bar u\in V(G_i)^h$, \\$\bar v\in V(G_i)^t$, $\bar w\in V(G_j)^\ell$ s.\ th.\ $V(G_i)\cap V(G_j)=\bar v$ and $T_i\cap V(G_i)\cap V(G_j)={T_j\cap V(G_i)\cap V(G_j)}$. 
Let $T=T_i\cup T_j$. 
Then for all $q\geq 0$, $tp_q(G_i\cup G_j,T,\bar u \bar v \bar w)$ is determined by $tp_q(G_i,T_i,\bar u \bar v)$ and $tp_q(G_j,T_j,\bar v \bar w)$.
Furthermore, there is an algorithm that computes $tp_q({G_i\cup G_j},T,\bar u \bar v \bar w)$ from $tp_q(G_i,T_i,\bar u \bar v)$ and $tp_q(G_j,T_j,\bar v \bar w)$ in time $f(q, t,h,\ell)$ for some computable function $f$.
\end{lemma}

We are now ready to present our \XP\ algorithm. 

\thmXP

\begin{proof}
Fix a formula $\phi(X)$ with quantifier rank $q$, and let $b$ be the number of token slides allowed.
We first compute a nice tree decomposition $(\mathcal{T},(B_i)_{i\in V(\mathcal{T})})$ of $G$ of width $tw$ and with $\Oof(n)$ nodes in \FPT\ time. 

We compute the following tables for the dynamic programming: for each node $i \in V(\mathcal{T})$, we compute a set of potential subsolutions that we denote by $\soli$. 
Every entry of $\soli$ is a six-tuple that we refer to as a \emph{state}. 
A state $s\in \soli$ is hence a tuple $(\kappa_s, \ell_s, T'_s, A_s, \Phi_s, f_s)$, where $0 \leq \kappa_s \leq k$ is an integer, $0 \leq \ell_s \leq b$ is an integer, $T'_s, A_s\subseteq B_i$, $\Phi_s$ is a $q$-type, $f_s: B_i \to [-k,k]$ is a function. 
The meaning of these entries is the following. 

\begin{itemize}
    \item $\kappa_s$ specifies the size of a partial configuration (placement) of tokens $T_i\subseteq V(G_i)$, 
    \item $\ell_s$ specifies the number of token sliding steps  used in the sub-problem on $G_i$ (we consider a more general sub-problem where multiple tokens may move to each vertex of~$B_i$), 
    \item $T'_s$ specifies the intersection of $T_i$ with $B_i$, 
    \item $A_s \subseteq T'_s$ denotes the set of vertices to which tokens initially located on vertices of $G_i$ are destined, 
    \item $\Phi_s$ specifies the type $tp_q(G_i, T_i, \bar v_i)$, and 
    \item $f_s: B_i \to [-k,k]$ is a function that specifies the number of tokens that slide \emph{through} each vertex in $B_i$, where a positive value indicates that more tokens slide from~$G_i$ to the rest of the graph, while a negative value indicates the converse. 
    We let $f^+_s = \{v \in B_i \mid f_s(v) > 0\}$ and $f^-_s = \{v \in B_i \mid f_s(v) < 0\}$. 

    \smallskip
    Intuitively, for each vertex $v \in f^+_s$ we have collected $f_s(v)$ tokens from $G_i$ that we imagine to lie on $v$ and now need to be sent outside of $G_i$. 
    A negative value indicates that $v$ has already given away tokens for the configuration $T_i$ that were not yet there, and now it still needs to collect $f_s(v)$ tokens from outside. 
\end{itemize}

\pagebreak
We will hence call a state $s=(\kappa_s, \ell_s, T'_s, A_s, \Phi_s, f_s)$ \emph{valid} if and only if 

\begin{itemize}
    \item there exists a set $T_i \subseteq V(G_i)$ of size~$\kappa_s$ such that
    \item $T'_s=T_i\cap B_i$, \item $A_s\subseteq T'_s$, 
    \item $\Phi_s=tp(G_i, T_i, \bar v_i)$, and 
    \item in $G_i$, the configuration placing a token on each vertex of $T_i\setminus (T'_s\setminus A_s)$ and $f_s(v)$ additional tokens on each vertex of $f^+_s$ can be reached from the configuration placing a token on each vertex of $S$ and additional~$|f_s(v)|$ tokens on each vertex of $f^-_s$ using at most $\ell_s$ token sliding steps.
\end{itemize}

We denote by $\soli$ the set of all valid states of the node $i$, which we aim to compute by dynamic programming. 
We accept the \textsc{$\phi$-Discovery} instance $(G,\phi,b)$ if we reach a valid state $s$ in the root node with $\kappa_s=k$, $\ell_s\leq b$ and such that $\Phi_s$ contains the formula $\phi$.

\subparagraph*{Leaf node.}
Let $i$ be a leaf node with bag $B_i = \emptyset$.
There is only one valid state $s = (0, 0, \emptyset, \emptyset, \Phi_{G_\emptyset}, \emptyset \to [-k,k])$, where $\Phi_{G_\emptyset}$ is the $q$-type of the empty graph $G_\emptyset$, which is efficiently computable (by brute force) in time depending only on $q$.

\subparagraph*{Introduce node.} Let $i$ be an introduce node with child $i+1$ such that $B_i = B_{i+1} \cup \{u\}$, for some $u \notin B_{i+1}$. 
Let $\bar v_{i+1}$ be the tuple of the vertices $B_{i+1}$ and $\bar v_i=\bar v_{i+1}u$. 

We generate all possible states of $G_i$ and for each potential state 
${s'=(\kappa_{s'}, \ell_{s'}, T'_{s'}, A_{s'}, \Phi_{s'}, f_{s'})}$ we verify whether it is valid as follows. 
For each state $s \in \sol_{i+1}$ we test wether $s$ can give rise to $s'$. 

Fix some $s=(\kappa_s, \ell_s, T'_s, A_s,\Phi_s, f_s)$. 
We now distinguish the following cases, depending on whether $T'_{s'}=T'_s$ (and in this case, we verify that $\kappa_{s'} = \kappa_{s}$) or $T'_{s'}=T'_s\cup \{u\}$ (and in this case, we verify that $\kappa_{s'} = \kappa_{s} + 1$) and on whether $u\in S$ or $u\not\in S$. 
We present the first case in detail; the other cases are similar.

\begin{enumerate}
    \item Assume $T'_{s'}=T'_s$ and $u \notin S$.
    In this case, no token is placed on $u$ in the initial, nor in the final configuration, we use $u$ only to slide tokens through it.
    We check whether $A_{s'}=A_s$.
    Observe that $N[u] \cap V(G_i) \subseteq B_i$ and let $N_u = N(u) \cap B_i$. 
    We check whether $f_s(v) = f_{s'}(v)$ for every $v \not\in N[u]$ and let $g(v) = f_s(v) - f_{s'}(v)$ for every $v \in N[u]$ and check whether $\sum_{v \in N_u}g(v) = f_{s'}(u)$ and $0 \leq g(v) \leq f_s(v)$ if $v\in f_s^+$, $0 \geq g(v) \geq f_s(v)$ if $v\in f_s^-$, and $g(v)=0$ otherwise.
    We also let $co(g) = \sum_{v \in N_u} |g(v)|$, the number of token sliding steps required to slide tokens through or to $u$.  
    We check whether $\ell_{s'}=\ell_s+co(g)$ and $\ell_{s'}\leq b$. 
    Finally, we verify that~$\Phi_{s'}$ results from $\Phi_s$ by \Cref{thm:feferman-vaught} where $G_j = G_{i+1}$, $G_i = G[B_i]$, $T_i = T'_{s'}$, $\bar u=u$, $\bar v=\bar v_{i+1}, \bar w=\emptyset$.
    \item If $T'_{s'}=T'_s$ and $u \in S$, then we proceed as in Case 1.
    However, the only difference is that we require $f_{s'}(u)=\big(\sum_{v\in N_u}g(v)\big)+1$, since we have one token on $u$. 
    \item When $T'_{s'}=T'_s \cup \{u\}$ and $u \in S$, we check whether $A_{s'}=A_s\cup \{u\}$ and proceed exactly the same as in Case 1.
    \item When $T'_{s'}=T'_s \cup \{u\}$ and $u \notin S$, we proceed as in Case 1 if $A_{s'}=A_s$ but require $f_{s'}(u)=\big(\sum_{v\in N_u}g(v)\big)-1$, and if $A_{s'}=A_s\cup \{u\}$ we require $f_{s'}(u)=\big(\sum_{v\in N_u}g(v)\big)-1$ and for at least one $v \in N_u$, $g(v) > 0$, as we need one additional token on $u$ from $G_i$. 
\end{enumerate} 

Finally, after verifying all potential states, in a cleaning step, for each obtained $\kappa_{s'}, T'_{s'}, A_{s'},$ $\Phi_{s'},$ and $f_{s'}$, we keep only one state with these values and the minimum value for $\ell$.
Let us estimate the running time. 
We need to generate all possible states and test whether they are valid. 
The number of possible states in $\soli$ or $\sol_{i+1}$ is bounded by $k\cdot b\cdot 2^{tw+2}\cdot \gamma\cdot (2k+1)^{tw+1}$ and each generated state can be validated in time that depends polynomially on $k^{tw}, b$, and $f'(q,tw)$ for some computable function~$f'$ from \Cref{thm:feferman-vaught} (we subsume here also dependence on the size of a type, which is bounded by a function depending only $q$ and $tw$). 
Thus, the running time is bounded by $h'(q,tw)\cdot k^{tw} n^{\Oof(1)}\subseteq h'(q,tw)\cdot n^{\Oof(tw)}$ for some computable function $h'$. 

\subparagraph*{Forget node.} Let $i$ be a forget node with child $i+1$ such that $B_i = B_{i+1} \setminus \{u\}$, for some $u \in B_{i+1}$. 
For each state $s\in \sol_{i+1}$, $s=(\kappa_s, \ell_s, T'_s, A_s, \Phi_s, f_s)$ and $f_s(u) = 0$, we compute a list of states to include in $\soli$. 
Observe that $N[u] \subseteq V(G_{i+1})$. 
Therefore, tokens can no longer enter $u$ from $G_i$ or leave $u$ to outside of $G_i$, so all states $s$ in $\sol_{i+1}$ that do not satisfy $f_s(u) = 0$ and $\{u\} \cap T'_{s} = \{u\} \cap A_{s}$ are no longer valid and are excluded from further consideration.
For each other state $s \in \sol_{i+1}$, we compute the follow-up state $s'=(\kappa_{s'}, \ell_{s'}, T'_{s'}, A_{s'}, \Phi_{s'}, f_{s'})$, where $\kappa_{s'}=\kappa_{s}, \ell_{s'}=\ell_{s}, T'_{s'}= T'_{s} \setminus \{u\}$, $A_{s'}= A_{s} \setminus \{u\}$, $f_{s'}(v) = f_s(v)$ for each $v \in B_i$, and~$\Phi_{s'}$ is verified by \Cref{thm:feferman-vaught} (in a time bounded by a function of $q$ and $tw$) by setting $G_j = G_{i+1}$, $G_i = G[B_i]$, $\bar u=\emptyset, \bar v=\bar v_i$, and $\bar w=\emptyset$. 

Finally, in a cleaning step, for each obtained $\kappa_{s'}, T'_{s'}, A_{s'}, \Phi_{s'},$ and $f_{s'}$, we keep only one state with these values and the minimum value for $\ell$.
The cleaning step can be performed in time polynomial in the number of such states.
We again obtain a running time of $h''(q,tw)\cdot n^{\Oof(tw)}$ for some computable function $h''$. 

\subparagraph*{Join node.}
Let $i$ be a join node with children $j$ and $h$ such that $B_i = B_j = B_h$. 
For each two states $s_j=(\kappa_{s_j}, \ell_{s_j}, T'_{s_j}, A_{s_j}, \Phi_{s_j}, f_{s_j})\in \sol_j$ and $s_h=(\kappa_{s_h}, \ell_{s_h}, T'_{s_h}, A_{s_h}, \Phi_{s_h}, f_{s_h})\in \sol_h$, we compute a potential state for $s'=(\kappa_{s'}, \ell_{s'}, T'_{s'}, A_{s'}, \Phi_{s'}, f_{s'})$ of $\soli$ as follows.

We check whether $T'_{s'}=T'_{s_j}= T'_{s_h}$, $A_{s_j}\cap A_{s_h} =\emptyset$, $A_{s'}=A_{s_j}\cup A_{s_h}$ and $\kappa_{s'}=\kappa_{s_j}+\kappa_{s_h}-|T'_{s'}|$ (we do not want to double count the shared part of the configuration).
We update the cost for the crossing of tokens via $B_i$ as follows. 
Note that a crossing of tokens occurs if $f_{s_j}(v)$ is positive and $f_{s_h}(v)$ is negative, or vice versa (when one vertex requests tokens from outside in one subgraph and wants to provide them in the other subgraph, it can satisfy its own needs, and the function~$f_{s'}$ is adapted accordingly). 
This is the case whenever $f_{s_j}(v) \cdot f_{s_h}(v) < 0$ for a vertex $v \in B_i$.
Thus, we check whether $f_{s'}=f_{s_j}+f_{s_h}$ and $\ell_{s'}=\ell_{s_j}+\ell_{s_h}$. 
Finally, we verify that $\Phi_{s'}$ is obtained from $\Phi_{s_j}$ and $\Phi_{s_h}$ using \Cref{thm:feferman-vaught} where $G_i = G_j$ and $G_j = G_h$, $\bar u=\bar w=\emptyset$ and $\bar v=\bar v_i$. 

When the loop ends, in a cleaning step, for each obtained $\kappa_{s'}, T'_{s'}, A_{s'}, \Phi_{s'},$ and $f_{s'}$, we keep only one state with these values and the minimum value for $\ell$.
The running time is dominated by the number of states in each of $\solj$ and $\solh$ (which is bounded by $k \cdot 2^{tw+2}\cdot \gamma\cdot (2k+1)^{tw+1}$) and the time to verify the types by \Cref{thm:feferman-vaught}. 
We again obtain a running time of $h'''(q,tw)\cdot n^{\Oof(tw)}$ for some computable function $h'''$. 

Let~$r$ be the root node. 
We obtain the final answer from the table entry at the root node. 
That is, we have a yes-instance if and only if $\phi(X) \in \Phi_s$ for some state $s\in \sol_r[(k,\ell, \emptyset, \emptyset, \Phi_s, \emptyset\to[-k,k])]$ for some $\ell \leq b$. 

It remains to add up the running time for all nodes. 
As we have $\Oof(n)$ nodes, this gives an additional factor of $\Oof(n)$, which in total sums up to $h(q,tw)\cdot n^{\Oof(tw)}$ for a final computable function~$h$. 
\end{proof}

\section{\MSO-Discovery Parameterized by Neighborhood Diversity is in \FPT}\label{sec:nd}
In this section, we show that \MSOD\ is in \FPT\ when parameterized by the parameter neighborhood diversity.

Let $V_1,\ldots,V_\ell$ be the partition of $V$ into sets of vertices such that any two vertices~$u$ and $v$ in~$V_i$ for $i \in [\ell]$ are twins and have the same colors (that is, $\mathcal{C}(u) = \mathcal{C}(v)$).
We refer to each $V_i$ as a \emph{vertex-type} and say that two vertices $u, v \in V_i$ are of the \emph{same vertex-type}. 
This partition can be computed in polynomial time and the number $\ell$ depends only on $nd(G)$ and on the number of colors in $\mathcal{C}$.

For this \FPT\ algorithm, we will require the concept of \emph{shapes} of vertex subsets~\cite{DBLP:journals/lmcs/KnopKMT19}.
The key insight is that, when dealing with logical formulas, many different vertex subsets behave identically from the formula's perspective.
Intuitively, the shape records how many vertices a vertex subset contains from each vertex-type, but with a crucial optimization.
When the subset $R$ contains a ``medium'' number of vertices from a vertex-type $V_i$ (not too few, not too many relative to some number), the exact count becomes irrelevant for the truth value of the formula.
The shape abstracts this by marking such counts as $\bot$.

\begin{definition}[Shape]
Let $\phi(X)$ be an \MSO-formula with a free set variable~$X$, and let $q_s$ and $q_v$ be the numbers of set and vertex quantifiers in $\phi$, respectively. 
Let~$G$ be a graph with $\ell$ vertex-types $V_1,\ldots,V_\ell$, and define $q(\phi) = 2^{q_s} \cdot q_v$. 
Let $R \subseteq V(G)$ and let $\sigma_R(i) = |R \cap V_i|$. 
The \emph{shape} of $R$ is defined as follows:
\[
\overline{\sigma}_R(i) =
\begin{cases}
\bot & \text{if } q(\phi) \le \sigma_R(i) \le |V_i| - q(\phi), \\
\sigma_R(i) & \text{otherwise},
\end{cases}
\]
\end{definition}

The following result of Gima et al.~\cite{gima2024algorithmic} is central to our algorithm.

\begin{lemma}\label{lem:vertex-shapes-feasible}
If two size-$k$ sets $R, R' \subseteq V(G)$ have the same shape, then $G\models\phi(R)$ if and only if~$G\models\phi(R')$.   
\end{lemma}

This abstraction dramatically reduces the search space when searching for a set that satisfies the formula.
Instead of considering all $\binom{n}{k}$ possible vertex subsets of size~$k$, we only need to examine one subset of each distinct shape, and the number of shapes depends on the formula size and the graph structure (particularly its neighborhood diversity), rather than the graph size.

Additionally, we will require an \FPT\ algorithm for the \textsc{Minimum‐Cost Flow with Interval Demand and Supply} problem (\textsc{MinMCF}) and an \FPT\ algorithm for model checking.
The latter exists for parameter neighborhood diversity and $|\phi|$ by a corollary from Lampis~\cite{DBLP:journals/algorithmica/Lampis12}.

\textsc{MinMCF} is solvable in polynomial time~\cite{DBLP:journals/networks/Hoppmann-Baum22} and is defined on a digraph \(D\) as follows.  
Each arc~$a$ in $D$ has a maximum capacity \(ca(a)\) and a unit of flow cost~\(co(a)\).  
Each vertex \(v\) in $D$ is assigned an interval \([I_v^{\min},\,I_v^{\max}]\).  
A flow is the assignment of a value \(fl(a)\) to each arc such that:
\begin{itemize}
  \item \(0 \le fl(a) \le ca(a)\) for all $a \in A(D)$;
  \item the \emph{net balance} for each $v \in V(D)$, denoted $ba(v) = \sum_{u:(u,v) \in A} fl(u,v)$ $- \sum_{w:(v,w) \in A} fl(v,w)$, is in the interval $[I_v^{\min},I_v^{\max}]$ (a positive net balance creates a sink, a negative net balance creates a source), and 
  \item \(\sum_{v\in V(D)} ba(v) = 0\).
\end{itemize}
The goal of the problem is to answer whether there exists a flow with a total cost \(\sum_{a\in A(D)} co(a) fl(a)\) at most $\ell$.

\pagebreak
We are now ready to present the \FPT\ algorithm.

\thmND

\begin{proof}
Let $((G,\mathcal{C}),S,b,\phi)$ be an instance of \MSOD\ where $C$ is the set of colors, $S$ is the initial set, $b$ is the budget, and $\phi$ is the formula. 
Let $\ell$ be the number of vertex-types in $(G,\mathcal{C})$; then $\ell \le nd(G) \cdot 2^{|\mathcal{C}|} \le nd(G) \cdot 2^{|\phi|}$.
We enumerate every admissible shape $\overline{\sigma}_{T}: [\ell] \rightarrow \{0, \ldots, q(\phi) - 1, \bot, |V_i| - q(\phi) + 1, |V_i| - q(\phi) + 2, \ldots,|V_i|\}$ of the final solution $T$ (of size $k$).
There are only $\ell$ vertex-types and $q(\phi)$ is a parameter, so this yields a bounded number ($(2q(\phi)+1)^\ell$) of shapes, which we call \emph{guesses}. 
Using a flow argument, we will now prove that, for each guess, we can compute in \FPT\ (in fact polynomial) time the minimum budget needed for a configuration of that shape (if any exists).

Let $\overline{\sigma}_{T}$ be a guess.
Let us first give an algorithm simulating the token sliding process using a flow in a digraph.
For each vertex-type $V_i$ of $(G,\mathcal{C})$, we create a vertex~$v_i$ in the digraph $D$ and connect the vertices $v_i,v_j$ (in both directions) if and only if all vertices in $V_i$ are adjacent to all vertices in $V_j$ (recall that vertices of the same type are twins in the graph).
Arcs of $D$ have infinite capacity and a unit of flow cost of $1$. 
Each vertex $v_i$ in $D$ has:
\begin{itemize}
    \item an interval equal to $[-\big(|V_i \cap S| - {\sigma}_{T}(i)\big)], -\big(|V_i \cap S| - {\sigma}_{T}(i)\big)]$ if $\overline{\sigma}_{T}(i) \neq \bot$ and $|V_i \cap S| > {\sigma}_{T}(i)$;
    \item an interval equal to $[{\sigma}_{T}(i) - |V_i \cap S|, {\sigma}_{T}(i) - |V_i \cap S|]$ if $\overline{\sigma}_{T}(i) \neq \bot$ and $|V_i \cap S| < {\sigma}_{T}(i)$;
    \item an interval equal to $[q(\phi) - |V_i \cap S|, |V_i| - q(\phi) - |V_i \cap S|]$ if $\overline{\sigma}_{T}(i) = \bot$ and $|V_i \cap S| < q(\phi)$;
    \item an interval equal to $[-\big(|V_i \cap S| - |V_i| + q(\phi)\big), \big(|V_i \cap S| - q(\phi)\big)]$ if $\overline{\sigma}_{T}(i) = \bot$ and $|V_i \cap S| > |V_i| - q(\phi)$; and
    \item an interval equal to $[0,0]$, otherwise.
\end{itemize}

We now ask whether there exists a flow in $D$ with a total cost of at most $b$. 
This can be done in \FPT-time (using the algorithm for \textsc{MinMCF}~\cite{DBLP:journals/networks/Hoppmann-Baum22}) since the number of vertex-types, thus vertices in~$D$, is bounded.
Then, for each guess $\overline{\sigma}_{T}$ with a flow of total cost at most $b$, we construct a $k$-size set of vertices $T \subseteq V(G)$ of shape $\overline{\sigma}_{T}$ in linear time.
We check whether $G \models \phi(T)$, which can be done in \FPT\ time with respect to $nd(G)$ and $|\phi|$, and return YES only when $G \models \phi(T)$. 
If we have exhausted all enumerated shapes and did not return YES, we return NO.

Since all steps involve a bounded number of guesses (equal to $(2q(\phi) + 1)^\ell$ with $\ell \le nd(G) \cdot 2^{|\mathcal{C}|} \le nd(G) \cdot 2^{|\phi|}$ since we are forming every possible mapping of $[\ell]$ to $(2q(\phi) + 1)$ numbers), and for each guess we incur \FPT\ running time with respect to $(nd(G), |\phi|)$, the overall algorithm is \FPT\ with respect to $(nd(G), |\phi|)$.

The algorithm is correct since, if for a fixed \MSO-formula $\phi$, a set $S$, and a budget $b$, a solution~$T^\star$ exists along with a transformation sequence $\vec{T}^\star$ of at most $b$ token slides that realizes $T^\star$ from $S$, then for at least one guess $T$ we will have  $\overline{\sigma}_{T} = \overline{\sigma}_{T^\star}$.
Given that the sequence $\vec{T}^\star$ encounters at most $b$ token slides between vertex-types, and achieves the set $T^\star$ of shape $\overline{\sigma}_{T^\star}$, routing one unit of flow for every inter‑vertex-type slide in $\vec{T}^\star$ yields a flow of cost at most $b$. 
In the final checking phase, the algorithm will be able to construct a set of shape $\overline{\sigma}_{T} = \overline{\sigma}_{T^\star}$, and since $T^\star$ is a solution, the set constructed by the algorithm will also be a solution by \Cref{lem:vertex-shapes-feasible} and therefore, the algorithm will return YES.

Conversely, if the algorithm accepts, then there exists a flow that satisfies the interval and cost constraints imposed by some guessed fixed shape $\overline{\sigma}_{T}$ and the budget~$b$.
A unit of flow sent along an arc in $D$ represents sliding one token across an edge of~$G$ that joins the two corresponding vertex-types; since each arc in $D$ between vertices representing vertex-types has a unit of flow cost of~$1$, a flow of total cost at most~$b$ encodes $b$ token slides between vertex-types. 
The intervals guarantee that, after performing those slides imposed by the flow, every type $V_i$ contains exactly ${\sigma}_{T}(i)$ tokens if $\overline{\sigma}_{T}(i) \neq \bot$ and between $q(\phi)$ and $|V_i| - q(\phi)$ tokens otherwise. 
Thus, the set~$T^\star$, resulting from performing those slides imposed by the flow on $G$, has the guessed shape $\overline{\sigma}_{T}$ and is a solution as is any other set of the shape $\overline{\sigma}_{T}$ which has lead to acceptance (\Cref{lem:vertex-shapes-feasible}).
\end{proof}
\section{Hardness for Modulator to Stars and Paths Numbers}\label{sec:logic-solution-discovery-modulator-hardness}
In this section, we show that \FOD\ is \W[1]-hard when parameterized by the parameters modulator to stars and modulator to paths numbers.
Before giving the proofs, we explain the different gadgets we use in each of the reductions.

For both results, we reduce from (a simplification of) the \textsc{Multicolored Clique} problem.
In this problem, we are given a graph $G$ where $V(G)$ is partitioned into $\kappa$ color classes, which we denote by $\{V_1, \ldots, V_\kappa\}$ (to avoid confusion with the color classes in our construction). 
The goal is to decide whether there exists a clique that contains a vertex from each color class.
The problem is \textsf{W[1]}-hard when parameterized by the solution size~$\kappa$ \cite{parameterizedalgorithms}. 
Without loss of generality, we will assume, throughout our reductions, that for each $i \in [\kappa]$, $|V_i| = n$ (we can add vertices of degree $0$ as needed).
For each $i \in [\kappa]$, we assign a distinct index in $[n]$ to each vertex in $V_i$ and we refer to this index for each $v \in V$ as \emph{$\iota(v)$}.
In an instance~$(G, \kappa)$ of the \textsc{Multicolored Clique} problem, the edge set~$E(G)$ can be partitioned into $\binom{\kappa}{2}$ sets $E_{i,j} = \{uv \mid u \in V_i, v \in V_j, uv\in E(G)\}$ for $1 \le i < j \le \kappa$. 

Let $(G, \kappa)$ be an instance of \textsc{Multicolored Clique}.
At a high level, token movements in the instance $((H, \mathcal{C}), S, b, \phi)$ of \FOD\ will encode both the selection of one vertex from each color class and one edge between each unordered pair of color classes, and the budget will force edge verification (that is, that the selected vertices form a clique).

The construction creates a \emph{vertex-block} for each color class $V_i$ and an \emph{edge-block} for each unordered pair of color classes.
A vertex-block for color class $V_i$, denoted \emph{vertex-block for $i$}, contains $n$ \emph{vertex-gadgets} (one for each vertex in $V_i$), plus a special \emph{block node} $b_i$ that connects all vertex-gadgets in the block. 
Similarly, an edge-block for an unordered pair of color classes $V_i$ and $V_j$, denoted \emph{edge-block for $\{i,j\}$}, contains \emph{edge-gadgets} (one for each edge between $V_i$ and $V_j$), plus a special block node $b_{i,j}$.
The gadgets themselves have the appropriate structure.
That is, they are stars in the modulator to stars reduction, and paths in the modulator to paths reduction, which ensures that they are not included in the modulator.
All vertex-blocks are identical, all edge-blocks are identical, all vertex-gadgets are identical, and all edge-gadgets are identical; information is encoded entirely in how gadgets are connected to \emph{connector nodes}. 
Initially, the leaves of the stars and the paths constituting vertex-gadgets are filled with tokens, and nodes constituting edge-gadgets are empty.

The \FO-formula specifies that for each vertex-block, exactly one vertex-gadget releases its tokens (encoding vertex selection) and for each edge-block exactly one edge-gadget, or all its leaves in the case of the modulator to stars reduction, is filled with tokens (encoding edge presence between selected vertices).

The reduction employs connector nodes that act as controlled gateways between vertex-blocks and edge-blocks.
For each unordered pair of color classes $V_i$ and $V_j$, we introduce four \emph{connector nodes for $\{i,j\}$}.
Equivalently, we introduce two pairs of connector nodes, \emph{connector nodes from color $i$ to $\{i,j\}$} and \emph{connector nodes from color~$j$ to $\{i,j\}$};
\begin{enumerate}
    \item the \emph{index connector node from color $i$ to $\{i,j\}$} which we call for brevity \emph{index$(i$,$\{i,j\})$},
    \item the \emph{remainder connector node from color $i$ to $\{i,j\}$} which we call for brevity \emph{remainder\\$(i$,$\{i,j\})$},
    \item the \emph{index connector node from color $j$ to $\{i,j\}$} which we call for brevity \emph{index$(j$,$\{i,j\})$}, and
    \item the \emph{remainder connector node from color $j$ to $(i,j)$} which we call for brevity \emph{remainder\\$(j$,$\{i,j\})$},
\end{enumerate}
with carefully designed adjacencies.
A connector node from color $i$ to $\{i,j\}$ indicates that the vertex-block for $i$ will serve as the origin of the tokens and the edge-block for $\{i,j\}$ will serve as the destination of the tokens that will traverse this connector.

The connections from a vertex-gadget for $v$ in $V_i$ to the connector nodes for $\{i,j\}$ encode the index of $v$ by having $\iota(v)$ (non-block) nodes of the gadget adjacent only to index$(i$,$\{i,j\})$ among all connector nodes and $n - \iota(v)$ other (non-block) nodes of the gadget adjacent only to remainder$(i$,$\{i,j\})$ among all connector nodes.
Connections from an edge-gadget representing edge $uv$ between $V_i$ and $V_j$ to the connector nodes for $\{i,j\}$ encode both indices $\iota(u)$ and $\iota(v)$, by having $\iota(u)$ (non-block) nodes of the gadget adjacent only to index$(i$,$\{i,j\})$ among all connector nodes, $n - \iota(u)$ other (non-block) nodes of the gadget adjacent only to remainder$(i$,$\{i,j\})$ among all connector nodes, $\iota(v)$ other (non-block) nodes of the gadget adjacent only to index$(j$,$\{i,j\})$ among all connector nodes, and $n - \iota(v)$ other (non-block) nodes of the gadget adjacent only to remainder$(j$,$\{i,j\})$ among all connector nodes.
This creates a ``lock-and-key'' mechanism: only the vertex-gadget representing $u$ can send tokens from its nodes connected to the connector nodes from color $i$ to $\{i,j\}$ to the edge-gadget for edge~$uv$ using only two token slides per token, and only the vertex-gadget for $v$ can send tokens from its nodes connected to the connector nodes from color $j$ to $\{i,j\}$ to the edge-gadget for edge $uv$ using only two token slides per token.
In other words, tokens can only flow efficiently (using two token slides per token, which is imposed by the budget) from any two vertex-gadgets to an edge-gadget (through connectors) if the indices match (that is, an edge exists between these two vertices in $G$).
Thus, a valid token movement pattern exists if and only if the selected vertices form a multicolored clique.

The modulator consists of all the block nodes and connector nodes, which total at most $\kappa + \binom{\kappa}{2} + 4\binom{\kappa}{2}$, and removing it leaves a disjoint union of stars (or paths), ensuring that the parameters are bounded by a function of $\kappa$.

We now present in greater detail the gadgets, blocks, and budget. 
In the subsequent subsections, we give the formula and proofs specific to each parameter.

As mentioned earlier, the construction consists of identical vertex-blocks (made of identical vertex-gadgets) and identical edge-blocks (made of identical edge-gadgets) which are joined together by connector vertices. 

To form the colored graph $(H, \mathcal{C})$, we first create $\kappa$ vertex-blocks, one for each color class $V_i$ in the instance $(G, \kappa)$ of \textsc{Multicolored Clique}.
\subparagraph*{Star vertex-block.} A \emph{star vertex-block for $i$} consists of a \emph{vertex block node $b_i$} that has as children the roots of $n$ stars (that is, star vertex-gadgets), one for each vertex $v$ in $V_i$.
The \emph{star vertex-gadget for $v$} is a star with a \emph{vertex root node for $v$}, and $\kappa-1$ \emph{groups} of $n$ leaves, where each group corresponds to one of the $\kappa-1$ color classes that differ from the color class $V_i$.
See \Cref{fig:logic-solution-discovery-modulator-hardness-star-vertex-gadget}.

We assign the following colors to the nodes in a star vertex-block:
\begin{itemize}          
\item{} purple ($C_1$) to all vertex root nodes, and 
\item{} teal ($C_2$) to all vertex block nodes.
\end{itemize} 

\noindent We also create $\binom{\kappa}{2}$ edge-blocks, one for each unordered pair of color classes $V_i$ and $V_j$.

\begin{figure}[ht]
    \centering
    \begin{tikzpicture}[every node/.style={inner sep=1.5pt}]
      \node[circle,draw,fill=violet,label={[font=\scriptsize]below right:\shortstack{vertex root node\\for $v$}}] (pix) at (0,-1.5) {};
      \node[circle,draw,fill=violet,label={[font=\scriptsize]below left:vertex root node for $u$}] (pux) at (-3,-1.5) {};
      \node[circle,draw,fill=violet,label={[font=\scriptsize]below right:vertex root node for $w$}] (pwx) at (3,-1.5) {};
      \node[circle,draw,fill=teal,label=below:{$b_3$}] (bi) at (0,-2) {};
      \draw (pix)--(bi);
      \draw (pux)--(bi);
      \draw (pwx)--(bi);

      \foreach \X [count=\n from 1] in {1,x,x+1,n}{
        \node[circle,draw] (q1\n) at (\n-7.0,0) {};
      }
      \node at ($(q11)!0.5!(q12)$) {$\cdots$};
      \node at ($(q13)!0.5!(q14)$) {$\cdots$};

      \foreach \X [count=\n from 1] in {1,x,x+1,n}{
        \node[circle,draw] (q2\n) at (\n-2.5,0) {};
      }
      \node at ($(q21)!0.5!(q22)$) {$\cdots$};
      \node at ($(q23)!0.5!(q24)$) {$\cdots$};

      \foreach \X [count=\n from 1] in {1,x,x+1,n}{
        \node[circle,draw] (q3\n) at (\n+2.0,0) {};
      }
      \node at ($(q31)!0.5!(q32)$) {$\cdots$};
      \node at ($(q33)!0.5!(q34)$) {$\cdots$};
      \foreach \g in {1,2,3}{
        \foreach \n in {1,2,3,4}{
          \draw (pix)--(q\g\n);
        }
      }
      \draw[decorate,decoration={brace,amplitude=5pt}] 
        ($(q11)+(0,0.4)$) -- ($(q14)+(0,0.4)$) 
        node[midway,above=5pt,font=\scriptsize,align=center] {group of $n$ nodes \\ corresponding to color class $V_{1}$};
      \draw[decorate,decoration={brace,amplitude=5pt}] 
        ($(q21)+(0,0.4)$) -- ($(q24)+(0,0.4)$) 
        node[midway,above=5pt,font=\scriptsize,align=center] {group of $n$ nodes \\ corresponding to color class $V_{2}$};
      \draw[decorate,decoration={brace,amplitude=5pt}] 
        ($(q31)+(0,0.4)$) -- ($(q34)+(0,0.4)$) 
        node[midway,above=5pt,font=\scriptsize,align=center] {group of $n$ nodes \\ corresponding to color class $V_{4}$};
    \end{tikzpicture}
    \caption{A star vertex-gadget for vertex $v \in V_3$ with $\iota(v)=x$, for an instance with $\kappa=4$, the vertex root nodes for $u, w \in V_3$, and the vertex block node $b_3$.}
    \label{fig:logic-solution-discovery-modulator-hardness-star-vertex-gadget}
\end{figure}
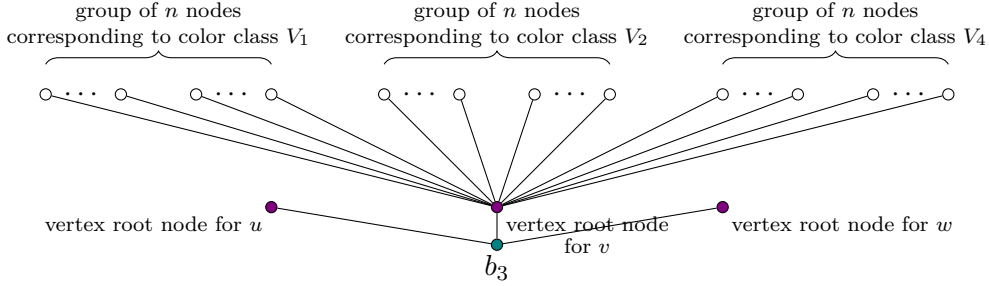
\subparagraph*{Star edge-block.} A \emph{star edge-block for $\{i,j\}$} consists of an \emph{edge block node $b_{i,j}$} that has as children the roots of multiple stars (that is, star edge-gadgets), one for each edge between $V_i$ and $V_j$.
The \emph{star edge-gadget for $e$} is a star with a \emph{edge root node for~$e$}, and two \emph{groups} of $n$ leaves, each corresponding to one endpoint of $e$.
See \Cref{fig:logic-solution-discovery-modulator-hardness-star-edge-gadget}.

We assign the following colors to the nodes in a star edge-block:
\begin{itemize}        
\item{} black ($C_3$) to all edge root nodes, and 
\item{} gray ($C_4$) to all edge block nodes.
\end{itemize} 

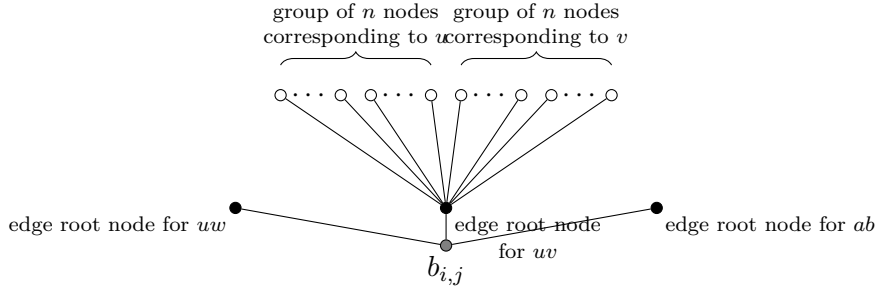
\begin{figure}[h]
  \centering
\begin{tikzpicture}[every node/.style={inner sep=1.5pt}]
   \node[circle,draw,fill=black, label={[font=\scriptsize]below right:\shortstack{edge root node\\for $uv$}}] (pe) at (0,-1.5) {};
   \node[circle,draw,fill=black,label={[font=\scriptsize]below left:edge root node for $uw$}] (uw) at (-2.8,-1.5) {};
   \node[circle,draw,fill=black,label={[font=\scriptsize]below right:edge root node for $ab$}] (ab) at (2.8,-1.5) {};
   \node[circle,draw,fill=gray,label=below:{$b_{i,j}$}] (bi) at (0,-2) {};
   \draw (pe)--(bi);
   \draw (uw)--(bi);
   \draw (ab)--(bi);
  
  \node[circle,draw] (g1n1) at (-2.2,0) {};
  \node at (-1.8,0) {$\cdots$};
  \node[circle,draw] (g1n3) at (-1.4,0) {};
  \node[circle,draw] (g1n4) at (-1.0,0) {};
  \node at (-0.6,0) {$\cdots$};
  \node[circle,draw] (g1n6) at (-0.2,0) {};
  
  \node[circle,draw] (g2n1) at (0.2,0) {};
  \node at (0.6,0) {$\cdots$};
  \node[circle,draw] (g2n3) at (1.0,0) {};
  \node[circle,draw] (g2n4) at (1.4,0) {};
  \node at (1.8,0) {$\cdots$};
  \node[circle,draw] (g2n6) at (2.2,0) {};
  
  \draw (pe)--(g1n1) (pe)--(g1n3) (pe)--(g1n4) (pe)--(g1n6);
  \draw (pe)--(g2n1) (pe)--(g2n3) (pe)--(g2n4) (pe)--(g2n6);
  
  \draw[decorate,decoration={brace,amplitude=5pt}] 
    ($(g1n1)+(0,0.4)$) -- ($(g1n6)+(0,0.4)$) 
    node[midway,above=5pt,font=\scriptsize,align=center] {group of $n$ nodes \\ corresponding to $u$};
  \draw[decorate,decoration={brace,amplitude=5pt}] 
    ($(g2n1)+(0,0.4)$) -- ($(g2n6)+(0,0.4)$) 
    node[midway,above=5pt,font=\scriptsize,align=center] {group of $n$ nodes \\ corresponding to $v$};
\end{tikzpicture}
  \caption{A star edge-gadget representing $e = uv$ where $u \in V_i$ and $v \in V_j$ with $\iota(u) = x$ and $\iota(v) = y$, the edge root nodes for $uw$ and $ab$ (edges between $V_i$ and $V_j$), and the edge block node $b_{i,j}$. The edge-gadget for $uv$ is formed of $2n$ leaf nodes, $n$ corresponding to $u$ and $n$ corresponding to $v$.}
  \label{fig:logic-solution-discovery-modulator-hardness-star-edge-gadget}
\end{figure}

In the case of the modulator to paths reduction, the vertex-blocks and edge-blocks are constructed as follows.
\subparagraph*{Path vertex-block.} A \emph{path vertex-block for $i$} consists of a \emph{vertex block node $b_i$} adjacent to the nodes of $n$ paths (that is, path vertex-gadgets), one for each vertex $v$ in $V_i$.
The \emph{path vertex-gadget for $v$} is a path with $n(\kappa-1)$ nodes, divided into $\kappa-1$ \emph{groups} of $n$ consecutive nodes, where each group corresponds to one of the $\kappa-1$ colors that differ from the color class $V_i$.
See \Cref{fig:logic-solution-discovery-modulator-hardness-path-vertex-gadget}.

We assign teal ($C_2$) to all vertex block nodes in a path vertex-block.

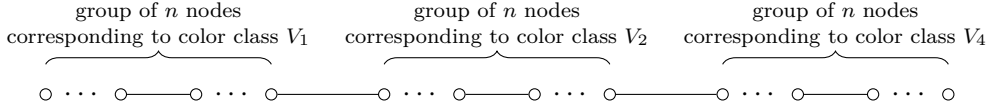
\begin{figure}[h]
    \centering
\begin{tikzpicture}[every node/.style={inner sep=1.5pt}]
  \node[circle,draw] (g1n1) at (-8,0) {};
  \node at (-7.5,0) {$\cdots$};
  \node[circle,draw] (g1n3) at (-7,0) {};
  \node[circle,draw] (g1n4) at (-6,0) {};
  \node at (-5.5,0) {$\cdots$};
  \node[circle,draw] (g1n6) at (-5,0) {};
  
  \node[circle,draw] (g2n1) at (-3.5,0) {};
  \node at (-3,0) {$\cdots$};
  \node[circle,draw] (g2n3) at (-2.5,0) {};
  \node[circle,draw] (g2n4) at (-1.5,0) {};
  \node at (-1,0) {$\cdots$};
  \node[circle,draw] (g2n6) at (-0.5,0) {};
  
  \node[circle,draw] (g3n1) at (1,0) {};
  \node at (1.5,0) {$\cdots$};
  \node[circle,draw] (g3n3) at (2,0) {};
  \node[circle,draw] (g3n4) at (3,0) {};
  \node at (3.5,0) {$\cdots$};
  \node[circle,draw] (g3n6) at (4,0) {};
  
  \draw (g1n3)--(g1n4);
  \draw (g1n6)--(g2n1);
  \draw (g2n3)--(g2n4);
  \draw (g2n6)--(g3n1);
  \draw (g3n3)--(g3n4);
  
  \draw[decorate,decoration={brace,amplitude=5pt}] 
    ($(g1n1)+(0,0.4)$) -- ($(g1n6)+(0,0.4)$) 
    node[midway,above=5pt,font=\scriptsize,align=center] {group of $n$ nodes \\ corresponding to color class $V_{1}$};
  \draw[decorate,decoration={brace,amplitude=5pt}] 
    ($(g2n1)+(0,0.4)$) -- ($(g2n6)+(0,0.4)$) 
    node[midway,above=5pt,font=\scriptsize,align=center] {group of $n$ nodes \\ corresponding to color class $V_{2}$};
  \draw[decorate,decoration={brace,amplitude=5pt}] 
    ($(g3n1)+(0,0.4)$) -- ($(g3n6)+(0,0.4)$) 
    node[midway,above=5pt,font=\scriptsize,align=center] {group of $n$ nodes \\ corresponding to color class $V_{4}$};
\end{tikzpicture}
  \caption{A path vertex-gadget for vertex $v \in V_3$ with $\iota(v)=x$, for an instance with $\kappa=4$.
  The vertex block node $b_3$ is adjacent to all nodes of path vertex-gadgets for all vertices in $V_3$, but it is omitted from the figure for clarity.}
  \label{fig:logic-solution-discovery-modulator-hardness-path-vertex-gadget}
\end{figure}

\subparagraph*{Path edge-block.} A \emph{path edge-block for $\{i,j\}$} consists of an \emph{edge block node $b_{i,j}$} adjacent to the nodes of multiple paths (that is, path edge-gadgets), one for each edge between $V_i$ and $V_j$.
The \emph{path edge-gadget for $e$} is a path with $2n$ nodes, with the \emph{first group} of $n$ nodes corresponding to one endpoint of $e$ and the \emph{last group} of $n$ nodes corresponding to the other endpoint.
See \Cref{fig:logic-solution-discovery-modulator-hardness-path-edge-gadget}.

We assign gray ($C_4$) to all edge block nodes in a path edge-block.

\begin{figure}[h]
    \centering
    \begin{tikzpicture}[every node/.style={inner sep=1.5pt}]
      \node[circle,draw] (g1n1) at (-5,0) {};
      \node at (-4.5,0) {$\cdots$};
      \node[circle,draw] (g1n3) at (-4,0) {};
      \node[circle,draw] (g1n4) at (-3,0) {};
      \node at (-2.5,0) {$\cdots$};
      \node[circle,draw] (g1n6) at (-2,0) {};
      
      \node[circle,draw] (g2n1) at (-0.5,0) {};
      \node at (0,0) {$\cdots$};
      \node[circle,draw] (g2n3) at (0.5,0) {};
      \node[circle,draw] (g2n4) at (1.5,0) {};
      \node at (2,0) {$\cdots$};
      \node[circle,draw] (g2n6) at (2.5,0) {};
      
      \draw (g1n3)--(g1n4);
      \draw (g1n6)--(g2n1);
      \draw (g2n3)--(g2n4);
      
     \draw[decorate,decoration={brace,amplitude=5pt}] 
     ($(g1n1)+(0,0.4)$) -- ($(g1n6)+(0,0.4)$) 
     node[midway,above=5pt,font=\scriptsize,align=center] {first group of $n$ nodes \\ corresponding to $u$};
     \draw[decorate,decoration={brace,amplitude=5pt}] 
     ($(g2n1)+(0,0.4)$) -- ($(g2n6)+(0,0.4)$) 
     node[midway,above=5pt,font=\scriptsize,align=center] {last group of $n$ nodes \\ corresponding to $v$};
    \end{tikzpicture}
    \caption{A path edge-gadget for edge $e=uv$ where $u \in V_i$ and $v \in V_j$ with $\iota(u) =x$ and $\iota(v)=y$.
    It is formed of $2n$ nodes, $n$ corresponding to $u$ and $n$ corresponding to $v$.
    The edge block node $b_{i,j}$ is adjacent to all nodes of path vertex-gadgets for all edges between $V_i$ and $V_j$, but it is omitted from the figure for clarity.}
    \label{fig:logic-solution-discovery-modulator-hardness-path-edge-gadget}
\end{figure}
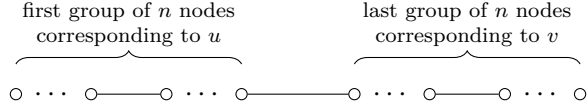

Then, for each unordered pair of color classes $V_i$ and $V_j$, we create four connector nodes; index connector node from color $i$ to $\{i,j\}$ which we call for brevity index$(i$,$\{i,j\})$, remainder connector node from color $i$ to $\{i,j\}$ which we call for brevity remainder$(i$,$\{i,j\})$, index connector node from color $j$ to $\{i,j\}$ which we call for brevity index$(j$,$\{i,j\})$, remainder connector node from color $j$ to $\{i,j\}$ which we call for brevity remainder$(j$,$\{i,j\})$. 
See \Cref{fig:logic-solution-discovery-modulator-hardness}.
We assign green ($C_5$) to all connector nodes.

We connect vertex-blocks and edge-blocks to green connector nodes as follows.

\begin{itemize}
\item{} For the vertex-gadget for the vertex $v \in V_i$ with $\iota(v) = x$, the $n$ nodes corresponding to the color class $V_j$ are split as follows: the first $x$ are adjacent to index$(i$,$\{i,j\})$, and the last $n-x$ are adjacent to remainder$(i$,$\{i,j\})$.
\item{} For the edge-gadget for the edge $e$ with endpoints $u \in V_i$ and $v \in V_j$ with $\iota(u) = x$ and $\iota(v) = y$, the $n$ nodes corresponding to $x$ are split as follows: the first $x$ are adjacent to index$(i$,$\{i,j\})$, and the last $n-x$ are adjacent to remainder$(i$,$\{i,j\})$. 
The $n$ nodes corresponding to $y$ are split as follows: the first $y$ are adjacent to index$(j$,$\{i,j\})$, and the last $n-y$ are adjacent to remainder$(j$,$\{i,j\})$. 
\end{itemize}

As described before, initially, the leaves of the stars and the paths constituting vertex-gadgets are filled with tokens. 

So that tokens can travel from vertex-gadgets to fill the nodes of one edge-gadget for each edge-block using exactly two token slides per token (that is, while checking that the indices match between the vertices and the edges), we set $b = 2n(\kappa-1)\kappa$.
We are now ready to give the formula and proofs specific to each parameter.

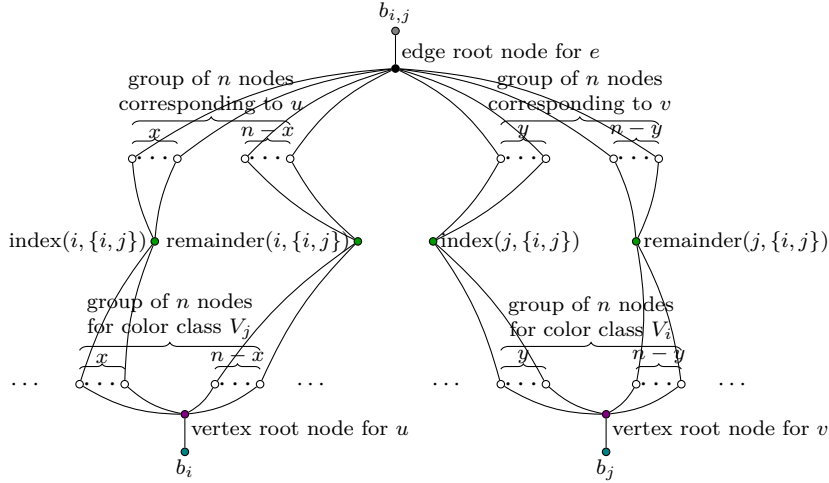
\begin{figure}[h]
    \centering
\begin{tikzpicture}[every node/.style={inner sep=1pt}, scale=1]
  \begin{scope}[yshift=1.5cm]
    \node[circle,draw,fill=black,label={[font=\scriptsize]above right: edge root node for $e$}] (oa) at (0,1.2) {};
    \node[circle,draw,fill=gray,label={[font=\scriptsize]above:$b_{i,j}$}] (bij) at (0,1.7) {};
    \draw (oa)--(bij);
    
    \node[circle,draw] (qa1) at (-3.5,0) {};
    \node at (-3.2,0) {$\cdots$};
    \node[circle,draw] (qa3) at (-2.9,0) {};
    \node[circle,draw] (qa4) at (-2,0) {};
    \node at (-1.7,0) {$\cdots$};
    \node[circle,draw] (qa6) at (-1.4,0) {};
    
    \node[circle,draw] (qa7) at (1.4,0) {};
    \node at (1.7,0) {$\cdots$};
    \node[circle,draw] (qa9) at (2,0) {};
    \node[circle,draw] (qa10) at (2.9,0) {};
    \node at (3.2,0) {$\cdots$};
    \node[circle,draw] (qa12) at (3.5,0) {};
    
    \draw[decorate,decoration={brace,amplitude=2pt}] 
      ($(qa1)+(0,0.2)$) -- ($(qa3)+(0,0.2)$) 
      node[midway,above=1pt,font=\scriptsize] {$x$};
    \draw[decorate,decoration={brace,amplitude=2pt}] 
      ($(qa4)+(0,0.2)$) -- ($(qa6)+(0,0.2)$) 
      node[midway,above=1pt,font=\scriptsize] {$n-x$};
    \draw[decorate,decoration={brace,amplitude=3pt}] 
      ($(qa1)+(0,0.45)$) -- ($(qa6)+(0,0.45)$) 
      node[midway,above=3pt,font=\scriptsize,align=center] {group of $n$ nodes\\corresponding to $u$};
      
    \draw[decorate,decoration={brace,amplitude=2pt}] 
      ($(qa7)+(0,0.2)$) -- ($(qa9)+(0,0.2)$) 
      node[midway,above=1pt,font=\scriptsize] {$y$};
    \draw[decorate,decoration={brace,amplitude=2pt}] 
      ($(qa10)+(0,0.2)$) -- ($(qa12)+(0,0.2)$) 
      node[midway,above=1pt,font=\scriptsize] {$n-y$};
    \draw[decorate,decoration={brace,amplitude=3pt}] 
      ($(qa7)+(0,0.45)$) -- ($(qa12)+(0,0.45)$) 
      node[midway,above=3pt,font=\scriptsize,align=center] {group of $n$ nodes\\corresponding to $v$};
  \end{scope} 
  
  \begin{scope}[yshift=-1.5cm, xshift=-2.8cm]
    \node[circle,draw,fill=violet,label={[font=\scriptsize]below right:vertex root node for $u$}] (pix) at (0,-0.4) {};
    \node[circle,draw,fill=teal,label={[font=\scriptsize]below:$b_i$}] (bi) at (0,-0.9) {};
    \draw (pix)--(bi);
    
    \node[circle,draw] (qb1) at (-1.4,0) {};
    \node at (-1.1,0) {$\cdots$};
    \node[circle,draw] (qb3) at (-0.8,0) {};
    \node[circle,draw] (qb4) at (0.4,0) {};
    \node at (0.7,0) {$\cdots$};
    \node[circle,draw] (qb6) at (1,0) {};
    
    \node[font=\footnotesize] at (-2.1,0) {$\cdots$};
    \node[font=\footnotesize] at (1.7,0) {$\cdots$};
    
    \draw[decorate,decoration={brace,amplitude=2pt}] 
      ($(qb1)+(0,0.2)$) -- ($(qb3)+(0,0.2)$) 
      node[midway,above=1pt,font=\scriptsize] {$x$};
    \draw[decorate,decoration={brace,amplitude=2pt}] 
      ($(qb4)+(0,0.2)$) -- ($(qb6)+(0,0.2)$) 
      node[midway,above=1pt,font=\scriptsize] {$n-x$};
    \draw[decorate,decoration={brace,amplitude=3pt}] 
      ($(qb1)+(0,0.45)$) -- ($(qb6)+(0,0.45)$) 
      node[midway,above=3pt,font=\scriptsize,align=center] {group of $n$ nodes\\for color class $V_j$};
  \end{scope}
  
  \begin{scope}[yshift=-1.5cm, xshift=2.8cm]
    \node[circle,draw,fill=violet,label={[font=\scriptsize]below right:vertex root node for $v$}] (pjy) at (0,-0.4) {};
    \node[circle,draw,fill=teal,label={[font=\scriptsize]below:$b_j$}] (bj) at (0,-0.9) {};
    \draw (pjy)--(bj);
    
    \node[circle,draw] (qc1) at (-1.4,0) {};
    \node at (-1.1,0) {$\cdots$};
    \node[circle,draw] (qc3) at (-0.8,0) {};
    \node[circle,draw] (qc4) at (0.4,0) {};
    \node at (0.7,0) {$\cdots$};
    \node[circle,draw] (qc6) at (1,0) {};
    
    \node[font=\footnotesize] at (-2.1,0) {$\cdots$};
    \node[font=\footnotesize] at (1.7,0) {$\cdots$};
    
    \draw[decorate,decoration={brace,amplitude=2pt}] 
      ($(qc1)+(0,0.2)$) -- ($(qc3)+(0,0.2)$) 
      node[midway,above=1pt,font=\scriptsize] {$y$};
    \draw[decorate,decoration={brace,amplitude=2pt}] 
      ($(qc4)+(0,0.2)$) -- ($(qc6)+(0,0.2)$) 
      node[midway,above=1pt,font=\scriptsize] {$n-y$};
    \draw[decorate,decoration={brace,amplitude=3pt}] 
      ($(qc1)+(0,0.45)$) -- ($(qc6)+(0,0.45)$) 
      node[midway,above=3pt,font=\scriptsize,align=center] {group of $n$ nodes\\for color class $V_i$};
  \end{scope}
  
  \begin{scope}
    \node[circle,draw,label={[font=\scriptsize]left:index$(i,\{i,j\})$}, fill=green!60!black] (sij) at (-3.2,0.4) {};
    \node[circle,draw,label={[font=\scriptsize]left:remainder$(i,\{i,j\})$}, fill=green!60!black] (rij) at (-0.5,0.4) {};
    \node[circle,draw,label={[font=\scriptsize]right:index$(j,\{i,j\})$}, fill=green!60!black] (sji) at (0.5,0.4) {};
    \node[circle,draw,label={[font=\scriptsize]right:remainder$(j,\{i,j\})$}, fill=green!60!black] (rji) at (3.2,0.4) {}; 
    
    \draw[bend right=15] (oa) to (qa1);
    \draw[bend right=15] (oa) to (qa3);
    \draw[bend right=15] (oa) to (qa4);
    \draw[bend right=15] (oa) to (qa6);
    
    \draw[bend left=15] (oa) to (qa7);
    \draw[bend left=15] (oa) to (qa9);
    \draw[bend left=15] (oa) to (qa10);
    \draw[bend left=15] (oa) to (qa12);
    
    \draw[bend left=10] (sij) to (qa1);
    \draw[bend left=10] (sij) to (qa3);
    \draw[bend left=10] (rij) to (qa4);
    \draw[bend left=10] (rij) to (qa6);
    
    \draw[bend right=10] (sji) to (qa7);
    \draw[bend right=10] (sji) to (qa9);
    \draw[bend right=10] (rji) to (qa10);
    \draw[bend right=10] (rji) to (qa12);
    
    \draw[bend left=15] (pix) to (qb1);
    \draw[bend left=15] (pix) to (qb3);
    \draw[bend right=15] (pix) to (qb4);
    \draw[bend right=15] (pix) to (qb6);
    
    \draw[bend left=15] (pjy) to (qc1);
    \draw[bend left=15] (pjy) to (qc3);
    \draw[bend right=15] (pjy) to (qc4);
    \draw[bend right=15] (pjy) to (qc6);
    
    \draw[bend right=10] (sij) to (qb1);
    \draw[bend right=10] (sij) to (qb3);
    \draw[bend right=10] (rij) to (qb4);
    \draw[bend right=10] (rij) to (qb6);
    
    \draw[bend left=10] (sji) to (qc1);
    \draw[bend left=10] (sji) to (qc3);
    \draw[bend left=10] (rji) to (qc4);
    \draw[bend left=10] (rji) to (qc6);
  \end{scope}
\end{tikzpicture}
    \caption{Part of the colored graph $(H,\mathcal{C})$ formed by the reduction corresponding to the modulator to stars number. 
    It shows the star edge-gadget of an edge $e=uv$ where $u \in V_i$ and $v \in V_j$ with $\iota(u)=x$ and $\iota(v)=y$, and relevant parts of the two star vertex-gadgets for $u$ and $v$.
    Additionally, the block nodes $b_i$, $b_j$, and $b_{i,j}$ are shown connecting to their respective root nodes.
    For $u$, the first $x$ nodes of the group corresponding to $V_j$ are adjacent to index$(i$,$\{i,j\})$ and the remaining $n-x$ nodes are adjacent to remainder$(i$,$\{i,j\})$.
    Similarly, for $v$, the first $y$ nodes of the group corresponding to~$V_i$ are adjacent to index$(j$,$\{i,j\})$ and the remaining $n-y$ nodes are adjacent to remainder$(j$,$\{i,j\})$.}
    \label{fig:logic-solution-discovery-modulator-hardness}
\end{figure}

\subsection{Hardness of \FOD\ for Modulator to Stars Number}\label{subsec:logic-solution-discovery-modulator-stars-hardness}
We start with the parameter modulator to stars number. 
More formally, the goal of this section is to prove the following: 

\thmMTS 

The rest of this subsection is devoted to the proof of \Cref{thm:logic-solution-discovery-modulator-stars-hardness}. 
Let $(G, \kappa)$ be an instance of the \textsc{Multicolored Clique} problem.
We consider the graph $(H, \mathcal{C})$ described at the start of \Cref{sec:logic-solution-discovery-modulator-hardness} where the edge- and vertex-blocks of $H$ are star edge-blocks and star vertex-blocks, the set $S$ consisting of all leaves of all star vertex-gadgets, and the budget $b=2n(\kappa-1)\kappa$. 
We complete $((H, \mathcal{C}), S, b)$ with an \FO-formula $\phi$ to create an instance of \FOD. 
We form the \FO-formula $\phi(X)$ to express that the set $X$ satisfies the following conditions:

S1. No colored node contains a token. (That is, the block nodes, the root nodes, and the connector nodes do not contain tokens.) 

S2. For each teal (vertex block) node, there exists exactly one adjacent purple (vertex root) node such that all the uncolored neighbors (leaves) of this purple node do not contain tokens. 
Moreover, all the other purple (vertex root) nodes adjacent to this teal node have tokens on all the uncolored nodes (leaves) in their neighborhoods.
(That is, in a vertex-block, exactly one vertex-gadget has released all the tokens on its leaves, while all other vertex-gadgets retain tokens on all their leaves.)
    
S3. For each gray (edge block) node, there exists exactly one adjacent black (edge root) node such that all the uncolored neighbors (leaves) of this black node contain tokens. 
Moreover, all the other black (edge root) nodes adjacent to this gray node are such that all the uncolored nodes (leaves) in their neighborhoods are token-free.
(That is, in an edge-block, all but one edge-gadget have no tokens on their leaves and one edge-gadget has tokens on all its leaves.)
    
The construction of these expressions in \FO\ is presented in Appendix~\ref{appsubsec:modulators}. 
It is easy to see that this reduction can be performed in polynomial time.
Before diving into the proof, let us reiterate some of the intuition. 
Each vertex-gadget initially contains $n \cdot (\kappa-1)$ tokens.
This corresponds to $n$ tokens for each of the other color classes that the vertex does not belong to. 
To transform the current set of tokens into a solution, we proceed as follows. 
For every $i \in [\kappa]$, we must empty exactly one vertex-gadget in the vertex-block for $i$ ((S1)+(S2)).
We must also fill the leaves of exactly one edge-gadget in each edge-block ((S1)+(S3)).
The transformation works by sliding tokens through green connector nodes. 
To consume at most $2n \cdot \kappa(\kappa-1)$ token slides (equivalently, two token slides per token), we must for each pair $\{i,j\}$, slide $n$ tokens from the selected vertex-gadget in the vertex-block for $i$ to the selected edge-gadget in edge-block for $\{i,j\}$ through adjacent green connector nodes. 
Similarly, we must slide $n$ tokens from the selected vertex-gadget in the vertex-block for $j$ to the same edge-gadget through adjacent green connector nodes. 
This completely fills the $2n$ leaves of that edge-gadget and empties the selected vertex- and edge-gadgets. 

We now present the lemmas necessary to complete the proof and start by showing that the modulator to stars number is a function of $\kappa$.

\begin{lemma}\label{lem:logic-solution-discovery-modulator-hardness-modulator-bounded}
$ms(H) \le 5\binom{\kappa}{2} + \kappa$.    
\end{lemma} 

\begin{proof}
Let $M$ be the set containing:
\begin{itemize}
    \item for each unordered pair of color classes $V_i$ and $V_j$, the green connector nodes $\text{index}(i,\{i,j\})$, $\text{remainder}(i,\{i,j\})$, $\text{index}(j,\{i,j\})$, $\text{remainder}(j,\{i,j\})$, and the gray edge block node $b_{i,j}$;
    \item for each color class $V_i$, the teal vertex block node $b_i$.
\end{itemize}
Then $G - M$ is a disjoint union of stars, since removing $M$ from $G$ leaves only the vertex-gadgets and edge-gadgets, each of which forms a star.
\end{proof}

We then prove correctness via the following lemmas.

\begin{lemma}\label{lem:logic-solution-discovery-modulator-stars-hardness-forward}
If $(G, \kappa)$ is a yes-instance of \textsc{Multicolored Clique}, then $((H,\mathcal{C}), S, b, \phi)$ is a yes-instance of \FOD.
\end{lemma}

\begin{proof}
Let $K \subseteq V(G)$ be a multicolored clique in $G$ of size $\kappa$. 
Let us now provide a transformation from the initial set of tokens $S$ to a valid solution in $b$ token sliding steps.
We repeat the following procedure for each unordered pair of color classes $V_i$ and $V_j$.

For $V_i$ and $V_j$, let $\iota(u) = x \in [n]$ (index of $u$ is $x$) be such that $u$ is the vertex in $K \cap V_i$ and $\iota(v) = y \in [n]$ (index of $v$ is $y$) be such that $v$ is the vertex in $K \cap V_j$. 
We know that the edge $uv$ exists since $K$ is a clique. 
We slide tokens as follows.
\begin{itemize}
    \item From the vertex-gadget for $u$ in the vertex-block for $i$: 
    \begin{itemize}
        \item we slide $x$ tokens from the first $x$ leaves in the group corresponding to color class $V_j$ to $\text{index}(i,\{i,j\})$ and then to the first $x$ leaves of the edge-gadget for $uv$ in the edge-block for $\{i,j\}$,
        \item we slide $n - x$ tokens from the remaining $n-x$ leaves in the same group to $\text{remainder}(i,\{i,j\})$ and then to the next $n-x$ leaves of the edge-gadget for $uv$.
    \end{itemize} 
    \item From the vertex-gadget for $v$ in vertex-block for $j$:
    \begin{itemize}
        \item we slide $y$ tokens from the first $y$ leaves in the group corresponding to color class $V_i$ to $\text{index}(j,\{i,j\})$ and then to the next $y$ leaves of the edge-gadget for $uv$ in the edge-block for $\{i,j\}$;
        \item we slide $n - y$ tokens from the remaining $n-y$ leaves in the same group to $\text{remainder}(j,\{i,j\})$ and then to the final $n-y$ leaves of the edge-gadget for $uv$.
    \end{itemize}
\end{itemize}

Let $T$ be the final configuration. 
Then it is easy to see that $T$ satisfies conditions (S1), (S2), and (S3), since in each vertex-block, exactly one vertex-gadget loses all its tokens ($n(\kappa-1)$ tokens in total), in each edge-block exactly one edge-gadget gains tokens on all its leaves ($2n$ tokens), all tokens are on star leaves, and the total number of token slides used during the transformation is $2n(\kappa-1)\kappa$ (for each $i < j \in [\kappa]$, $2n$ tokens slide exactly twice).  
Hence, $((H,\mathcal{C}), S, b, \phi)$ is a yes-instance of \FOD.
\end{proof}

To prove the backward direction, we will need the following lemma.

\begin{lemma}\label{lem:logic-solution-discovery-modulator-stars-hardness-implications-of-conditions}
If $((H,\mathcal{C}), S, b, \phi)$ is a yes-instance of \FOD, then 
\begin{enumerate}
    \item[i.] each token that moves is moved exactly two token sliding steps and passes through a green connector node,
    \item[ii.] for each edge-block, the solution $T$ contains $2n$ leaves of exactly one edge-gadget in the edge-block, and
    \item[iii.] exactly $n(\kappa-1)$ tokens slide from exactly one vertex-gadget in each vertex-block.
\end{enumerate}
\end{lemma}

\begin{proof}
Since $T$ is a solution and satisfies Condition S3, each edge-block must contain $2n$ tokens on the leaves of exactly one edge-gadget, which proves (ii).
Since the leaves of each edge-gadget are not adjacent to any node in $S$, each token that ends up on a leaf of this edge-gadget requires at least two token slides. 
So, at least $4n$ token slides are needed for each edge-block.
Since $b=2n \kappa (\kappa-1)= 4n\binom{\kappa}{2}$ and the above argument holds for all $\binom{\kappa}{2}$ edge-blocks, if there is a solution, every token that moves slides exactly twice. 
Since green connector nodes are the only nodes adjacent to both leaves of vertex-gadgets and leaves of edge-gadgets, the intermediate positions of all moved tokens must be green connector nodes, which proves (i).
From Condition~S2, for each vertex-block, exactly one vertex-gadget releases all its tokens (which are $n(\kappa-1)$ tokens), which proves (iii).
\end{proof}

\begin{lemma}\label{lem:logic-solution-discovery-modulator-stars-hardness-backward}
If $((H,\mathcal{C}), S, b, \phi)$ is a yes-instance of \FOD, then $(G, \kappa)$ is a yes-instance of \textsc{Multicolored Clique}. 
\end{lemma}

\begin{proof}
Let $T$ be a solution to $((H,\mathcal{C}), S, b, \phi)$. 
From (S1) no block, root, or connector node contains a token.
From \Cref{lem:logic-solution-discovery-modulator-stars-hardness-implications-of-conditions}, for each edge-block, exactly $2n$ leaves of one edge-gadget must contain tokens in the final configuration $T$. 
Each of these tokens slides through one of the green connector nodes and uses exactly two token slides.
So each of the $2n$ tokens that end up in the edge-block for $\{i,j\}$ must come from one of the green connector nodes $\text{index}(i,\{i,j\})$, $\text{remainder}(i,\{i,j\})$, $\text{index}(j,\{i,j\})$, or $\text{remainder}(j,\{i,j\})$ after leaving one of the vertex-blocks for $i$ or $j$ (since these green connectors are not adjacent to leaves from other vertex-blocks). 

By \Cref{lem:logic-solution-discovery-modulator-stars-hardness-implications-of-conditions}, there exist vertices $u \in V_i$ and $v \in V_j$ such that the vertex-gadget for $u$ and the vertex-gadget for $v$ are completely emptied of tokens (and all other vertex-gadgets in the vertex-blocks for $i$ or $j$ retain all their tokens).
Similarly, there exists an edge $u'v'$ with $u' \in V_i$ and $v' \in V_j$ such that all leaves of the edge-gadget for~$e$ are in $T$ (and no other edge-gadget in the edge-block for $\{i,j\}$ gets a token). 

We show that $u = u'$ and $v = v'$.
Let $\iota(u) = x$ and $\iota(v) = y$. 
Recall from the construction that the vertex-gadget for $u$ has exactly $n$ leaves in the group corresponding to color class $V_j$, where the first $x$ leaves are adjacent only to $\text{index}(i,\{i,j\})$ among all green connector nodes, and the remaining $n-x$ leaves are adjacent only to $\text{remainder}(i,\{i,j\})$ among all green connector nodes.
Similarly, the vertex-gadget for $v$ has exactly $n$ leaves in the group corresponding to color class~$V_i$, where the first $y$ leaves are adjacent only to $\text{index}(j,\{i,j\})$ among all green connector nodes, and the remaining $n-y$ leaves are adjacent only to $\text{remainder}(j,\{i,j\})$ among all green connector nodes.

Since $2n$ tokens must reach the edge-gadget for $u'v'$ through the green connector nodes from the vertex-gadgets for $u$ and $v$ using exactly two token slides each, each token from these vertex-gadgets moves to its unique adjacent green connector node during the first token slide. 
Thus, during the transformation, $x$ tokens from the vertex-gadget for $u$ slide to $\text{index}(i,\{i,j\})$ and $n - x$ tokens slide to $\text{remainder}(i,\{i,j\})$. 
Similarly, $y$ tokens from the vertex-gadget for $v$ slide to $\text{index}(j,\{i,j\})$ and $n - y$ tokens slide to $\text{remainder}(j,\{i,j\})$.

All these tokens must next slide to the same edge-gadget for $u'v'$. 
Let $\iota(u') = x'$ and $\iota(v') = y'$. 
The edge-gadget for $u'v'$ has $x'$ leaves adjacent to $\text{index}(i,\{i,j\})$, $n-x'$ leaves adjacent to $\text{remainder}(i,\{i,j\})$, $y'$ leaves adjacent to $\text{index}(j,\{i,j\})$, and $n-y'$ leaves adjacent to $\text{remainder}(j,\{i,j\})$.
For all $2n$ tokens to successfully slide from the green connector nodes to fill all leaves of this edge-gadget using only one additional token slide per token, we must have $x = x'$ and $y = y'$.
Thus $u = u'$ and $v = v'$, which means the edge $uv$ exists in $G$.

This equality holds for every unordered pair of color classes. 
Thus, for each $i < j \in [\kappa]$, the vertices whose vertex-gadgets are emptied in the vertex-block for $i$ and the vertex-block for $j$ are adjacent in $G$, which establishes a multicolored clique in $(G, \kappa)$.
\end{proof}

\subsection{Hardness of \FOD\ for Modulator to Paths Number}\label{subsec:logic-solution-discovery-modulator-paths-hardness}
We now present the hardness for the parameter modulator to paths number. 
More formally, the goal of this section is to prove the following: 

\thmMTP

The rest of this subsection is devoted to the proof of \Cref{thm:logic-solution-discovery-modulator-paths-hardness}.
Let $(G, \kappa)$ be an instance of the \textsc{Multicolored Clique} problem.
We consider the graph $(H, \mathcal{C})$ described at the beginning of \Cref{sec:logic-solution-discovery-modulator-hardness} where the vertex-blocks and edge-blocks of $H$ are path vertex-blocks and path edge-blocks, the set $S$ consisting of all nodes of all path vertex-gadgets, and the budget $b=2n(\kappa-1)\kappa$. 
We complete $((H, \mathcal{C}), S, b)$ with an \FO-formula $\phi$ to create an instance of \FOD. 
We form the \FO-formula~$\phi(X)$ to express that the set $X$ satisfies the following conditions:

P1. No colored node contains a token. (That is, the block nodes and the connector nodes do not contain tokens.)

P2. For each teal (vertex block) node, there exists at least one adjacent uncolored node (in a path vertex-gadget) that is token-free.

P3. For each gray (edge block) node, there exists at least one adjacent uncolored node (in a path edge-gadget) that contains a token.

P4. Every uncolored node $v$ (in a gadget) contains a token if and only if all uncolored nodes (in a gadget) at distance at most two from $v$ in $H[V \setminus (C_2 \cup C_4 \cup C_5)]$ also contain tokens. (That is, a node of a gadget contains a token if and only if the gadget nodes at distance two from it in the subgraph excluding all block and connector nodes contain tokens.)

The construction of these expressions in \FO\ is presented in Appendix~\ref{appsubsec:modulators}.
It is easy to see that this reduction can be performed in polynomial time.

Before proving the theorem, let us explain some of the intuition behind the conditions.
Condition P4 ensures that in each path gadget, either all nodes contain tokens or all nodes are token-free. Combined with conditions (P2) and (P3), this means that in each vertex-block, exactly one path vertex-gadget becomes token-free, while in each edge-block, exactly one path edge-gadget becomes filled with tokens.

We first show that the modulator to paths number is bounded.

\pagebreak
\begin{lemma}\label{lem:logic-solution-discovery-modulator-paths-hardness-modulator-bounded}
$mp(H) \le 5\binom{\kappa}{2} + \kappa$.
\end{lemma}

\begin{proof}
Let $M$ be the set containing:
\begin{itemize}
    \item for each unordered pair of color classes $V_i$ and $V_j$, the green connector nodes $\text{index}(i,\{i,j\})$, $\text{remainder}(i,\{i,j\})$, $\text{index}(j,\{i,j\})$, $\text{remainder}(j,\{i,j\})$, and the gray edge block node $b_{i,j}$;
    \item for each color class $V_i$, the teal vertex block node $b_i$.
\end{itemize}
Then $G - M$ is a disjoint union of paths, since removing $M$ from $G$ leaves only the vertex-gadgets and edge-gadgets, each of which forms a path.
\end{proof}

\begin{lemma}\label{lem:logic-solution-discovery-modulator-paths-hardness-forward}
If $(G, \kappa)$ is a yes-instance of \textsc{Multicolored Clique}, then $((H,\mathcal{C}), S, b, \phi)$ is a yes-instance of \FOD.
\end{lemma}

\begin{proof}
The proof follows the same transformation as in \Cref{lem:logic-solution-discovery-modulator-stars-hardness-forward}. 
The token movements use only edges between the gadgets nodes and the green connector nodes, which are identical in both constructions. 
The final configuration $T$ satisfies conditions (P1)-(P4): no colored nodes contain tokens, each vertex-block has one empty path vertex-gadget, each edge-block has one full path edge-gadget, and Condition P4 is satisfied since tokens fill or empty entire paths.
\end{proof}

To prove the backward direction, we need the following lemma.

\begin{lemma}\label{lem:logic-solution-discovery-modulator-paths-hardness-implications-of-conditions}
If $((H,\mathcal{C}), S, b, \phi)$ is a yes-instance of \FOD, then 
\begin{itemize}
    \item[i.] each token that moves is moved exactly two token sliding steps and passes through a green connector node,
    \item[ii.] for each edge-block, the solution $T$ contains tokens on all $2n$ nodes of exactly one edge-gadget in the edge-block, and
    \item[iii.] exactly $n(\kappa-1)$ tokens slide from exactly one vertex-gadget in each vertex-block.
\end{itemize}
\end{lemma}

\begin{proof}
By Condition P4, one node in a path contains a token if and only if all nodes in that path contain tokens.
From conditions (P3) and (P4), for each edge-block, at least one path edge-gadget must be completely filled with tokens.
Thus, exactly $2n$ tokens must be placed on one path edge-gadget.

Since the nodes of edge-gadgets are not adjacent to any nodes in $S$, each token that ends up on an edge-gadget requires at least two token slides. 
So, at least $4n$ token slides are needed for each edge-block.
Since $b = 2n(\kappa-1)\kappa = 4n\binom{\kappa}{2}$ and $4n$ token slides are needed for all $\binom{\kappa}{2}$ edge-blocks, every token that moves slides exactly twice and every edge-block gets to fill exactly one of its path edge-gadgets, which proves (ii). 
Since green connector nodes are the only nodes adjacent to both nodes of vertex-gadgets and nodes of edge-gadgets, all moved tokens must pass through green connector nodes, which proves (i).

From conditions (P2) and (P4), for each vertex-block, at least one path vertex-gadget must be completely token-free. 
Since each path vertex-gadget initially contains $n(\kappa-1)$ tokens, exactly this many tokens must slide from one vertex-gadget, which proves (iii).
\end{proof}

\begin{lemma}\label{lem:logic-solution-discovery-modulator-paths-hardness-backward}
If $((H,\mathcal{C}), S, b, \phi)$ is a yes-instance of \FOD, then $(G, \kappa)$ is a yes-instance of \textsc{Multicolored Clique}. 
\end{lemma}

\begin{proof}
The proof follows the same argument as in \Cref{lem:logic-solution-discovery-modulator-stars-hardness-backward}, using \Cref{lem:logic-solution-discovery-modulator-paths-hardness-implications-of-conditions} instead of \Cref{lem:logic-solution-discovery-modulator-stars-hardness-implications-of-conditions}. The key observation is that Condition P4 ensures that entire paths are either filled or emptied, maintaining the same token movement pattern as in the star case.
\end{proof}

This completes the proof of \Cref{thm:logic-solution-discovery-modulator-paths-hardness}.

\section{Hardness for Twin Cover Number}\label{sec:twincover}
In this section, we prove the following theorem.

\thmTC

To prove \Cref{thm:logic-solution-discovery-twincover-hardness}, we reduce from the \textsc{Planar Arc Supply} problem.
In this problem, we are given a planar digraph $D$ with $|V(D)| + |A(D)| = \kappa$, no loops, and no double arcs. 
Each node $u \in V(D)$ has a demand $\delta(u)$ and each arc $uv\in A(D)$ has a list $L_{uv}= \{ (x^1_{uv},y^1_{uv}), (x^2_{uv},y^2_{uv}), \ldots, (x^\ell_{uv},y^\ell_{uv})\}$ of integer pairs, called the \emph{supply pairs} of $uv$. 
We let $\mathcal{L}$ represent the supply lists of all arcs and $\lambda : A(D) \rightarrow \mathcal{L}$ be a function that maps each arc in $D$ to a list in $\mathcal{L}$.
Informally, we are asked to pick from the list of each arc $uv \in A(D)$ a supply pair $(x^i_{uv},y^i_{uv})$ and send a supply of $x^i_{uv}$ to $u$ and a supply of $y^i_{uv}$ to $v$ such that for every vertex $v$, the sum of the supplies from the arcs incident to $v$ satisfies exactly its demand.  
More formally, the task is to decide whether there is a function $\rho_x: A(D)\to\mathbb{N}$ and a function $\rho_{y}: A(D)\to\mathbb{N}$ such that for every arc $uv\in A(D)$ we have $(\rho_{x}(uv),\rho_{y}(uv))\in L_{uv}$ and for every vertex $u\in V(D)$ we have:
\[
   \delta(u)=\sum_{v\in N^{+}(u)} \rho_{x}(uv)\;+\;
                     \sum_{v\in N^{-}(u)} \rho_{y}(vu).
\]

The problem is \W[1]-hard when parameterized by $\kappa$~\cite[Section 3.2]{DBLP:conf/iwpec/BodlaenderLP09} even when the integers are given in unary, and vertices have either degree two or degree four.

Let $(D,\delta,\lambda)$ be an instance of \textsc{Planar Arc Supply}. 
At a high level, the instance $((H, \mathcal{C}), S, b, \phi)$ of \FOD\ will have a \emph{demand clique} for each vertex (initially with no tokens), a \emph{reservoir clique} for each arc (initially full of tokens), and two \emph{supply cliques} for each tuple in each arc's supply list in $D$ (see \Cref{fig:logic-solution-discovery-twincover-hardness-a}).
The key insight of the reduction is that token movements in $H$ will encode the selection of supply pairs from the lists in $\mathcal{L}$.

The twin cover structure acts as a ``control skeleton'' that connects the different cliques: we create a twin cover node for each original vertex and two twin cover nodes for each arc in $D$, ensuring that the resulting graph has a bounded twin cover number.
Each vertex $u \in V(D)$ gets a corresponding twin cover \emph{vertex node} connected to its demand clique of size $\delta(u)$, while each arc $uv \in A(D)$ gets two twin cover \emph{arc nodes}, one $u$\emph{-arc node for} $uv$ and one $v$\emph{-arc node for} $uv$.
The $u$-arc node for $uv$ is connected to supply cliques, each encoding the first component of a distinct supply pair (the values of $x$) in $L_{uv}$.
These supply cliques are also connected to the vertex node for $u$ and we refer to them as the \emph{$u$-supply cliques for $uv$}.
The $v$-arc node for $uv$ is connected to supply cliques, each encoding the second component of a distinct supply pair (the values of $y$) in~$L_{uv}$.
These supply cliques are also connected to the vertex node for~$v$ and we refer to them as the \emph{$v$-supply cliques for $uv$}.
Both arc nodes for $uv$ are connected to the reservoir clique for $uv$.

Each supply clique contains two types of nodes: \emph{supply nodes} that initially contain tokens (whose number equals the corresponding component value in the supply pair) and empty \emph{index nodes} (whose number encodes the index of the tuple within the supply list).
The \FO-formula ensures that for each arc, we must ``activate'' exactly one supply pair by expressing that for each arc node, exactly one of the adjacent supply cliques empties its supply nodes and fills its index nodes with tokens. 
The \FO-formula additionally specifies that the reservoir cliques are emptied of tokens and the demand cliques are full of tokens.

The crux of our reduction lies in a numerical encoding scheme (see \Cref{fig:logic-solution-discovery-twincover-hardness-b}). 
Each reservoir clique for an arc $uv \in A(D)$ contains $\Sigma^2$ tokens (where $\Sigma$ is larger than the sum of all demand and supply values), while the number of index nodes in the $u$-supply cliques for $uv$ and the number of index nodes in the $v$-supply cliques for $uv$ uses a base-$\Sigma$ representation: for the $i$th supply pair of the arc~$uv$, the corresponding $u$-supply clique for $uv$ has $\Sigma^2 - \Sigma \cdot (i-1)$ index nodes and the corresponding $v$-supply clique for $uv$ has $\Sigma \cdot (i-1)$ index nodes.
This creates a useful mathematical property: if forced to split between only two supply cliques, one connected to the arc node for $u$ and one connected to the arc node for $v$, the $\Sigma^2$ tokens from a reservoir clique can only fill the index nodes of a matching pair of supply cliques (corresponding to a tuple in $L_{uv}$) whose index node counts sum up to exactly $\Sigma^2$.

This numerical constraint, combined with our tight budget of only two token slides per displaced token, creates a forced choreography of token movements.
Tokens must flow into a demand clique for a vertex $u \in V(D)$ from at most one $u$-supply clique for each incident arc to $u$.
These token slides will amount to two times the sum of all demands in the instance. 
Additionally, with these budget and numerical constraints, the tokens in a reservoir clique must leave and first fill the index nodes of exactly one pair of supply cliques that correspond to a tuple in $L_{uv}$.
Once the index nodes are filled, the tokens on the supply nodes of these activated cliques can be displaced and must move to their corresponding demand cliques.
Any deviation from this pattern either violates the formula or exceeds the budget.
Thus, valid token movements exist if and only if we can select supply pairs whose values exactly satisfy all vertex demands and $(D, \delta, \lambda)$ is a yes-instance of \textsc{Planar Arc Supply}.

In the following description of an instance $((H, C), S, b, \phi)$ of \FOD, we use specific colors both for greater clarity and for use in the \FO-formula.

As described above, we set $\Sigma$ to be twice the sum of all integers in the input (that is, the sum of all demands and all integers appearing inside the tuples for all tuples). We form the instance $((H, C), S, b, \phi)$ of \FOD\ as follows. See \Cref{fig:logic-solution-discovery-twincover-hardness-b} for a sketch of the construction.

\subparagraph*{Twin cover.} For each $u \in V(D)$ we create a \emph{vertex node} and for each arc $uv \in A(D)$ we create two \emph{arc nodes}, one \emph{$u$-arc node for $uv$} and one \emph{$v$-arc node for $uv$}.

We assign the following colors to the nodes in the twin cover:
\begin{itemize}
\item purple ($C_1$) to all vertex nodes, and
\item black ($C_2$) to all arc nodes.
\end{itemize}

\subparagraph*{Demand cliques.} For each $u \in V(D)$, we form a \emph{demand clique for $u$} of size $\delta(u)$, whose nodes are adjacent to the vertex node for $u$.

We assign yellow ($C_3$) to all nodes in the demand cliques.

\subparagraph*{Supply cliques.} For each arc $uv$ such that $L_{uv}$ has $\ell$ pairs, we create, for each $i \in [\ell]$, a \emph{$u$-supply clique for $uv$} with $x^i_{uv}$ \emph{supply nodes} and $\Sigma^2 - \Sigma \cdot (i - 1)$ \emph{index nodes}.
We also create, for each $i \in [\ell]$, a \emph{$v$-supply clique for $uv$} with $y^i_{uv}$ supply nodes and $\Sigma \cdot(i - 1)$ index nodes.

We assign the following colors to the nodes in the supply cliques:
\begin{itemize}
\item orange ($C_4$) to all supply nodes, and 
\item teal ($C_5$) to all index nodes.
\end{itemize}

\subparagraph*{Reservoir cliques.} For each arc $uv \in A(D)$, we form a \emph{reservoir clique} for $uv$ that contains~$\Sigma^2$ nodes and make its nodes adjacent to the arc node for $u$ and the arc node for $v$.

We assign brown ($C_6$) to all nodes in the reservoir cliques.
\\

Orange supply nodes and brown nodes of reservoir cliques contain tokens.
Let~$\Delta$ be the sum of the demands of all the vertices in $D$ (equivalently, the total size of all the demand cliques in $H$). 
We set the budget to $2\Delta + |A(D)| \cdot 2\Sigma^2$. 
The first component ($2\Delta$) accounts for moving tokens from supply nodes to demand cliques (each token making two token slides: from its initial supply clique to a twin cover arc node, then to a demand clique). 
The second component ($|A(D)| \cdot 2\Sigma^2$) accounts for the $\Sigma^2$ tokens in each reservoir clique, where each token also makes two token slides to go from the reservoir clique to fill the index nodes in the selected supply cliques pair.

Our twin cover is made up of one vertex node for each vertex $u \in V(D)$ and two arc nodes for each arc $uv \in A(D)$. 
Deleting the nodes in the twin cover leaves a graph with only disjoint cliques.
Additionally, each of the nodes in the twin cover is adjacent to all or none of the nodes in each clique.
Thus, the twin cover number is at most $\kappa + 2\binom{\kappa}{2}$.

\tikzset{
  halfandhalf/.style n args={3}{
    draw,
    minimum size=#1,
    circle,
    path picture={
      \pgfmathsetlengthmacro{\rad}{0.5*#1}
      \begin{scope}
        \clip (0,0) circle (\rad);
        \path[fill=#2,draw=none] (-\rad,0) rectangle (\rad,\rad);
      \end{scope}
      \begin{scope}
        \clip (0,0) circle (\rad);
        \path[fill=#3,draw=none] (-\rad,-\rad) rectangle (\rad,0);
      \end{scope}
    }
  }
}

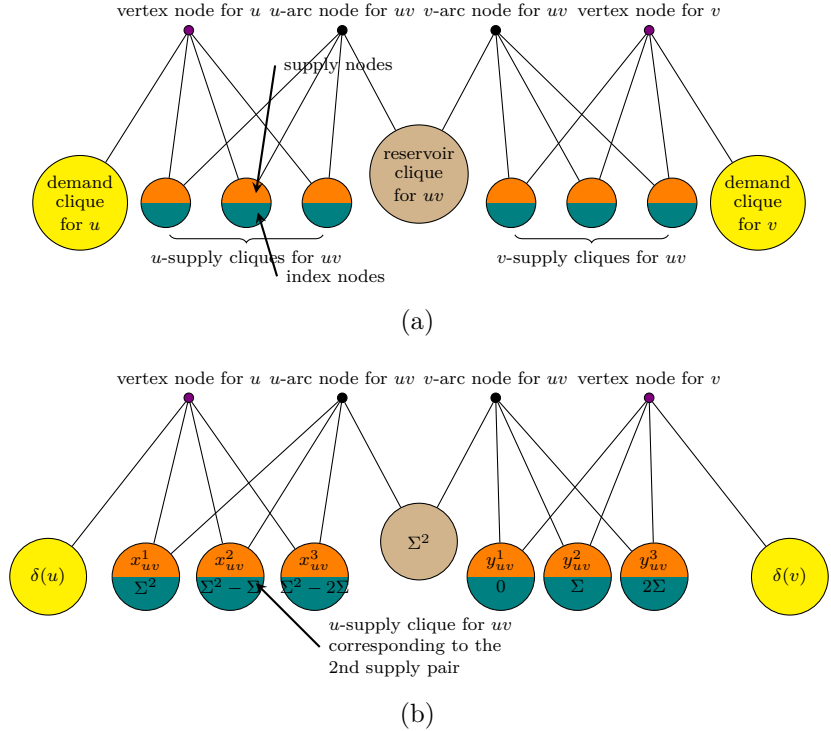
\begin{figure}[h]
    \centering
    \begin{subfigure}{\textwidth}
    \centering
    \begin{tikzpicture}[>=stealth,
                    every node/.style={font=\scriptsize},
                    vertex/.style={circle,draw=black,inner sep=1.5pt,font=\scriptsize},
                    scale=0.85,transform shape]
  \node[vertex,fill=violet,label=above: vertex node for $u$] (wu) at (0,1.5) {};
  \node[vertex,fill=black,label=above: $u$-arc node for $uv$] (wuv) at (2.4,1.5) {};
  \node[vertex,fill=black,label=above: $v$-arc node for $uv$] (wuvstar) at (4.8,1.5) {};
  \node[vertex,fill=violet,label=above: vertex node for $v$] (wv) at (7.2,1.5) {};
  \node[vertex,circle,draw,fill=yellow,minimum size=1.2cm] (Cu) at (-1.7,-1.2) {\shortstack{demand\\clique\\for $u$}};
  \node[vertex,circle,draw,fill=yellow,minimum size=1.2cm] (Cv) at (8.9,-1.2) {\shortstack{demand\\clique\\for $v$}};
  \node[vertex,circle,draw,fill=brown,minimum size=1.2cm] (Cuv) at (3.6,-0.75) {\shortstack{reservoir\\clique\\for $uv$}};
  \draw (wu) -- (Cu);
  \draw (wv) -- (Cv);
  \draw (wuv) -- (Cuv);
  \draw (wuvstar) -- (Cuv);
  \def\t{3}       
  \def\diam{22pt}
  \def\shift{0.25*\diam}
  \def\shiftBelow{0.8*\diam}
  \foreach \i in {1,...,\t} {
      \pgfmathsetmacro{\pos}{0.15 + (\i-1)*0.35}
      \node (DiscU\i)
            [halfandhalf={\diam}{orange}{teal}]
            at ($(-0.9,-1.2)!{\pos}!(2.7,-1.2)$) {};
      \draw (wu) -- (DiscU\i);
      \draw (wuv) -- (DiscU\i);
      
      \node (DiscV\i)
            [halfandhalf={\diam}{orange}{teal}]
            at ($(4.5,-1.2)!{\pos}!(8.1,-1.2)$) {};
      \draw (wv) -- (DiscV\i);
      \draw (wuvstar) -- (DiscV\i);
    }
    \draw[<-,thick] ($(DiscU2)+(0.15,0.15)$) -- ++(0.3,2) node[right,font=\scriptsize] {supply nodes};
    \draw[<-,thick] ($(DiscU2)+(0.18,-0.15)$) -- ++(0.3,-1) node[right,font=\scriptsize] {index nodes};
    
    \draw[decorate,decoration={brace,amplitude=3pt,mirror}] (-0.3,-1.7) -- (2.1,-1.7)
        node[midway,below=3pt,font=\scriptsize] {$u$-supply cliques for $uv$};
    \draw[decorate,decoration={brace,amplitude=3pt,mirror}] (5.1,-1.7) -- (7.5,-1.7)
        node[midway,below=3pt,font=\scriptsize] {$v$-supply cliques for $uv$};
\end{tikzpicture}
    \caption{}
    \label{fig:logic-solution-discovery-twincover-hardness-a}
    \end{subfigure}
    
    \vspace{0.3cm}
    
    \begin{subfigure}{\textwidth}
    \centering
    \begin{tikzpicture}[>=stealth,
                    every node/.style={font=\scriptsize},
                    vertex/.style={circle,draw=black,inner sep=1.5pt,font=\scriptsize},
                    scale=0.85,transform shape]
  \node[vertex,fill=violet,label=above: vertex node for $u$] (wu) at (0,1.5) {};
  \node[vertex,fill=black,label=above: $u$-arc node for $uv$] (wuv) at (2.4,1.5) {};
  \node[vertex,fill=black,label=above: $v$-arc node for $uv$] (wuvstar) at (4.8,1.5) {};
  \node[vertex,fill=violet,label=above: vertex node for $v$] (wv) at (7.2,1.5) {};
  \node[circle,draw,fill=yellow,minimum size=1.2cm] (Cu) at (-2.2,-1.3) {$\delta(u)$};
  \node[circle,draw,fill=yellow,minimum size=1.2cm] (Cv) at (9.4,-1.3) {$\delta(v)$};
  \node[circle,draw,fill=brown,minimum size=1.2cm] (Cuv) at (3.6,-0.75) {$\Sigma^2$};
  \draw (wu) -- (Cu);
  \draw (wv) -- (Cv);
  \draw (wuv) -- (Cuv);
  \draw (wuvstar) -- (Cuv);
  \def\t{3}       
  \def\diam{30pt}
  \def\shift{0.25*\diam}
  \def\shiftBelow{0.8*\diam}
  \foreach \i/\j/\k in {1/{0}/{\Sigma^2},2/{\Sigma}/{\Sigma^2-\Sigma},3/{2\Sigma}/{\Sigma^2-2\Sigma}} {
      \pgfmathsetmacro{\pos}{0.12 + (\i-1)*0.3}
      \node (DiscU\i)
            [halfandhalf={\diam}{orange}{teal}]
            at ($(-1.2,-1.3)!{\pos}!(3.2,-1.3)$) {};
      \node[font=\scriptsize, inner sep=0pt]
            at ($(DiscU\i)+(0,\shift)$) {$x^{\i}_{uv}$};
      \node[font=\scriptsize, inner sep=0pt]
            at ($(DiscU\i)+(0,-0.15)$) {$\k$}; 
      \draw (wu) -- (DiscU\i);
      \draw (wuv) -- (DiscU\i);
      
      \node (DiscV\i)
            [halfandhalf={\diam}{orange}{teal}]
            at ($(4.4,-1.3)!{\pos}!(8.4,-1.3)$) {};
      \node[font=\scriptsize, inner sep=0pt]
            at ($(DiscV\i)+(0,\shift)$) {$y^{\i}_{uv}$};
      \node[font=\scriptsize, inner sep=0pt]
            at ($(DiscV\i)+(0,-0.15)$) {$\j$};
      \draw (wv) -- (DiscV\i);
      \draw (wuvstar) -- (DiscV\i);
    }
    \draw[<-,thick] ($(DiscU2)+(0.4,-0.1)$) -- ++(1,-1) node[right,font=\scriptsize,align=left] {$u$-supply clique for $uv$\\corresponding to the\\$2$nd supply pair};
\end{tikzpicture}
    \caption{}
    \label{fig:logic-solution-discovery-twincover-hardness-b}
    \end{subfigure}
    \caption{Part of the colored graph $(H,\mathcal{C})$ formed by the reduction corresponding to the twin cover number. 
    It shows the part of the twin cover and the cliques that correspond to two vertices $u, v \in V(D)$ and the edge $uv \in A(D)$ with $|L_{uv}| = 3$.
    A line from a vertex or arc node to a clique means that the node is adjacent to all nodes of the clique.
    The nodes of the demand clique are yellow and those of the reservoir clique are brown.  
    Supply nodes are orange, index nodes are teal, vertex nodes are purple, and arc nodes are black.
    (a) Since $|L_{uv}| = 3$, we have three $u$-supply cliques for $uv$ and three $v$-supply cliques for $uv$.
    (b) Each $u$-supply clique for $uv$ corresponding to the $i$th supply pair contains $\Sigma^2 - \Sigma \cdot (i - 1)$ teal ($C_5$) index nodes and $x^i_{uv}$ orange ($C_4$) supply nodes.
    Each $v$-supply clique for $uv$ corresponding to the $i$th supply pair contains $\Sigma \cdot (i - 1)$ teal index nodes and $y^i_{uv}$ orange supply nodes.}
    \label{fig:logic-solution-discovery-twincover-hardness}
\end{figure}

To complete the construction of the \FOD\ instance $((H, C), S, b, \phi)$, we form the \FO-formula $\phi(X)$ to express that $X$ satisfies the following conditions:

T1. Each yellow (demand clique) node contains a token. 

T2. The brown (reservoir clique) nodes do not contain tokens. 

T3. For each pair of purple (vertex) and black (arc) nodes, exactly one adjacent (supply) clique containing at least one orange (supply) node must be such that all its teal (index) nodes contain tokens and all its orange (supply) nodes do not. For all other adjacent (supply) cliques of this type, their orange (supply) nodes contain tokens, and their teal (index) nodes do not. 
(That is, for every arc $uv \in A(D)$ there exists a unique $u$-supply clique for $uv$ and a unique $v$-supply clique for $uv$ that are modified: all tokens are removed from the supply nodes, and every index node is filled with a token.)

T4. The purple (vertex) and black (arc) nodes do not contain tokens.

The construction of these expressions in \FO\ is given in Appendix~\ref{appsubsec:twincover} to not disrupt the flow of the paper.
It is easy to see that this reduction can be performed in polynomial time.
We start by presenting the forward direction of the proof of~\Cref{thm:logic-solution-discovery-twincover-hardness}.

\begin{lemma}
If $(D, \delta, \lambda)$ is a yes-instance of \textsc{Planar Arc Supply}, then $((H, C),  S, b, \phi)$ is a yes-instance of \FOD.    
\end{lemma}

\begin{proof}
Let $\rho_x$ and $\rho_y$ be a solution of $(D, \delta, \lambda)$. For each arc $uv \in A(D)$, let $(x^i_{uv}, y^i_{uv})$ be such that $\rho_x(uv) = x^i_{uv}$ and $\rho_y(uv) = y^i_{uv}$.

Let us now describe a solution of $((H, C), S, b, \phi)$. 
We slide the $x^i_{uv}$ tokens from the orange supply nodes of the $u$-supply clique for $uv$ corresponding to the $i$th supply pair for $uv$ toward the yellow nodes of the demand clique for $u$ and the $y^i_{uv}$ tokens from the orange supply nodes of the $v$-supply clique for $uv$ corresponding to the $i$th supply pair for $uv$ toward the yellow nodes of the demand clique for $v$. In total, for all arcs $uv$, this consumes $2\Delta$ token slides.
We then slide $\Sigma^2 - \Sigma \cdot(i - 1)$ tokens from the reservoir clique for $uv$ to the teal index nodes of the $u$-supply clique for $uv$ corresponding to the $i$th supply pair for $uv$ and $\Sigma \cdot(i-1)$ tokens to the teal index nodes of the $v$-supply clique for $uv$ corresponding to the $i$th supply pair for $uv$.
This consumes $|A(D)| \cdot 2\Sigma^2$ token slides.

Given that $\rho_x$ and $\rho_y$ are a solution, Condition T1 is satisfied.
Since $\Sigma^2-\Sigma \cdot(i-1)+\Sigma \cdot(i-1) = \Sigma^2$, Condition T2 is also satisfied. 
Condition T3 is satisfied since we chose the supply cliques for $uv$ to correspond to one supply tuple in the supply list for $uv$ and moved the tokens at the nodes of the reservoir clique for $uv$ to fill the teal index nodes of the chosen supply cliques and moved the tokens at the orange supply nodes out of those chosen supply cliques.
Finally, no purple vertex or black arc node is in $T$ and Condition T4 is satisfied.
Thus, we find that $((H, C), S, b, \phi)$ is a yes-instance, as needed.
\end{proof}

The following lemma shows the backward direction of the proof of~\Cref{thm:logic-solution-discovery-twincover-hardness}.

\begin{lemma}
If $((H, C), S, b, \phi)$ is a yes-instance of \FOD, then $(D, \delta, \lambda)$ is a yes-instance of \textsc{Planar Arc Supply}.    
\end{lemma}

\begin{proof}
Assume $((H, C), S, b, \phi)$ is a yes-instance of \FOD, and let $\mathcal{T}$ be a transformation from the initial set of tokens to a solution.
Then, at least $|A(D)| \cdot 2\Sigma^2$ token slides are needed to empty the brown nodes of the reservoir cliques and satisfy conditions (T2) and (T4).
Since the order of the moves of different tokens does not impact the validity of a solution, we can assume that these moves are performed first in the transformation and let us call $S'$ the resulting configuration of tokens.
Note that all nodes in $S'$ are in supply cliques.

Since the total budget is $2\Delta + |A(D)| \cdot 2\Sigma^2$, there exists a transformation from~$S'$ to a solution in at most $2\Delta$ token slides.
Since no token of $S'$ is incident to a demand clique, at least $2\Delta$ token slides are needed in order to fill the $\Delta$ yellow nodes of demand cliques and satisfy Condition T1.
Thus, from $S'$, exactly $\Delta$ tokens have to move and require exactly two token slides each to fill the demand cliques and these are the only token slides that remain.

The final solution satisfies conditions (T3) and (T4). 
Let us consider the supply cliques to which the tokens initially in the reservoir cliques went.
By Condition~T3, the final configuration is modified compared to the initial configuration only on one $u$-supply clique for~$uv$ and one $v$-supply clique for~$uv$ for every arc $uv \in A(D)$. 
Thus, for every arc $uv \in A(D)$, there exist one chosen $u$-supply clique for~$uv$ (say, corresponding to the $i$th supply pair for $uv$) and one chosen $v$-supply clique for~$uv$ (say, corresponding to the $j$th supply pair for $uv$) such that more than $\Delta$ tokens of the reservoir clique for $uv$ slide and stay on (as at most $\Delta$ of $\Sigma^2$ tokens can slide toward the demand cliques for~$u$ and $v$).

Moreover, since tokens of a solution lie on different nodes, there are at most $\Delta$ tokens that can move from the chosen $u$-supply clique for $uv$ and the chosen $v$-supply clique for $uv$ to demand cliques.
Since all teal index nodes in the chosen $u$-supply clique for~$uv$ and the chosen $v$-supply clique for~$uv$ have to contain tokens in the final configuration and since all those tokens must come from the reservoir clique for $uv$, $\Sigma^2 - \Sigma \cdot(i-1)$ of those tokens move to the chosen $u$-supply clique for $uv$ and $\Sigma \cdot (j - 1)$ move to the chosen $v$-supply clique for $uv$. 
So we should have $\Sigma^2 \geq \Sigma^2 - \Sigma \cdot (i - 1) + \Sigma \cdot (j - 1)$ and thus $i \geq j$. 
Moreover, the total number of tokens that can move after $S'$ is bounded by $\Delta$, and the number of tokens that move from the reservoir clique for $uv$ to the chosen supply cliques is $\Sigma^2 - (\Sigma^2 - \Sigma \cdot(i - 1) + \Sigma \cdot(j - 1)) = \Sigma \cdot(i - 1) - \Sigma \cdot(j - 1) \leq \Delta$. 
For $\Sigma$ to be smaller than or equal to $\Delta$, $i \le j$.
Thus, $i = j$ and all the tokens that move from the reservoir clique for $uv$ go to teal index nodes of the chosen $u$-supply clique for $uv$ and the chosen $v$-supply clique for $uv$.

Condition T3 ensures that the tokens initially on the orange supply nodes of the chosen $u$-supply clique for $uv$ and the chosen $v$-supply clique for $uv$ have to move.
Since every token that moves after $S'$ can move exactly twice to go to the unique demand clique at distance two from it, all the tokens on the orange supply nodes of the chosen $u$-supply clique for $uv$ have to go to the nodes of the demand clique for $u$ and all the tokens on the orange supply nodes of the chosen $v$-supply clique for $uv$ have to go to the nodes of the demand clique for $v$.
This argument holds for every arc in the instance.

Since tokens lie on different nodes of the final solution, for every $u \in A(D)$, the sum of the tokens going to the demand clique of $u$ from the $u$-supply cliques for the arcs incident to $u$ cannot be larger than $\delta(u)$. 
To conclude, let us prove that the sum has to be equal to $\delta(u)$.
Assume by contradiction that some demand is not satisfied, i.e., (T1) is not satisfied for a vertex $u$ after moving the tokens on the orange supply nodes.
To transform the current set into a solution, additional tokens from a $u$-supply clique for some arc incident to $u$ should go to the demand clique for $u$.
By Condition~T3, such tokens cannot be initially on an orange supply node of another $u$-supply clique. 
Additionally, as argued before such tokens cannot be tokens initially on a brown node in a reservoir clique as these tokens fill the teal index nodes and cannot slide further.
So the demand of each vertex $u \in V(D)$ has to be exactly satisfied from the orange supply nodes of the chosen $u$-supply cliques for incident arcs, which completes the proof.
\end{proof}

\section{Hardness of \MSOD\ for Bandwidth}\label{sec:bandwidth}
This section is dedicated for the proof of the following theorem.

\thmBW

We use a reduction from the \textsc{Planar Arc Supply} problem, starting with a slight modification to form the instance $(D,\delta,\lambda)$.
To ensure that all lists are of equal length, while maintaining the validity of the supply constraints, we add as many copies of the first tuple as needed (zero or more) to increase the length to $t$.
We then create for each vertex $u \in V(D)$ a vertex-gadget $H_u$ and for each arc $uv \in A(D)$ an arc-gadget $H_{uv}$. 

At a high level, we can view each gadget as a ``ladder'', with two long paths forming the \emph{side rails} and shorter paths forming the \emph{rungs} that connect the side rails.
We refer to the ``top'' and ``bottom'' of a side rail, consistent with how one might draw an image of a ladder.
The two side rails of an arc-gadget $H_{uv}$ correspond to the two endpoints $u$ and $v$ of the arc $uv$. 
In a vertex-gadget, the \emph{isolated side rail} has no connections outside the vertex-gadget, whereas the \emph{connected side rail} has connections to side rails of arc-gadgets associated with (at most four) incident arcs. 
When two nodes in a vertex-gadget (one on the isolated side rail and one on the connected side rail) are joined by a rung, we call them \emph{rail neighbors}.

To represent supply pairs, an arc-gadget $H_{uv}$ contains $t$ \emph{sub-ladders}, one for each supply tuple in $L_{uv}$.
At the bottom of the sub-ladder for supply tuple $(x,y)$ is a rung, above which the side rail corresponding to $u$ contains $x$ tokens, and the side rail corresponding to $v$  contains $y$ tokens; all tokens are located on \emph{token-bearing nodes}.
The \MSO-formula ensures that each arc-gadget empties exactly one sub-ladder.
It also specifies that for each vertex $u$, a token be moved to each of the $\delta(u)$ \emph{demand nodes} located at the top of the connected side rail of the vertex-gadget.  In a yes-instance of \textsc{Planar Arc Supply},  tokens can be moved from the side rails of arc-gadget sub-ladders associated with the chosen supply pairs into the corresponding vertex-gadgets, and then up to the demand nodes.

Also specified in the \MSO-formula, the \emph{rail neighbor condition} requires that for any node $g$ on the connected side rail that is the neighbor of a token-bearing node $h$, there must be a token on either $h$ or the \emph{mirror} of $h$ (the rail neighbor of $g$, located on the isolated side rail).
Each vertex-gadget contains a total of $t+1$ sub-ladders, one for each supply pair, as well as one at the bottom of the ladder, containing $\delta(u)$ \emph{reservoir nodes} that each hold a token.
Thus, when the tokens on a side rail of a sub-ladder of an arc-gadget are emptied into the vertex-gadget and moved up to the demand nodes, the rail neighbor condition is satisfied by moving tokens to the mirrors of those token-bearing nodes from the reservoir nodes.

Crucially, no matter which sub-ladder is selected from an arc-gadget, we incur the costs of moving its tokens upward to the demand nodes in the vertex-gadgets and moving tokens from the reservoirs in the vertex-gadgets to the mirrors of now-empty token-bearing nodes.
Since, per shifted token, the two costs sum up to the length of a vertex-gadget side rail, the total cost will be the same no matter which sub-ladder of the arc-gadget we choose.
The budget restricts tokens to slide in the described manner, so that each vertex-gadget for a vertex $u \in V(D)$ receives exactly $\delta(u)$ tokens and does so by taking in the tokens (i.e. supplies) of the selected sub-ladders (i.e. selected supply pairs) of connected arc-gadgets.

To restrict the ways tokens can slide, we introduce \emph{spacer nodes}, designed to increase the lengths of certain subpaths within gadgets. 
There are spacer nodes along the side rails of arc-gadgets in order for them to ``match'' the lengths of the side rails of vertex-gadgets, and there are also spacer nodes along rungs.
The length of side rails is chosen to ensure that the side rail of each sub-ladder has enough space to hold and receive tokens.
We set $\Sigma$ to be three times the sum of all integers in the input (that is, the sum of all demands and, for all tuples, all integers appearing inside the tuples), and to account for there being as many as four neighbors of a vertex, we set the length of a side rail of a sub-ladder to be $4\Sigma$.

In the following description of an instance $((H, \mathcal{C}), S, b, \phi)$ of \MSOD, we use specific colors both for greater clarity and for use in the \MSO-formula; moreover, for brevity, we often refer to nodes simply by color rather than both color and function.

\subparagraph*{Vertex-gadget.} Each vertex-gadget $H_u$ consists of two paths of length $4\Sigma(t+1)$ (the \emph{isolated side rail} and the \emph{connected side rail}), connected by $4\Sigma(t+1)$ paths of length four (the \emph{rungs}); the three intermediate nodes in each rung are all gray \emph{spacer nodes}. 
The $t+1$ subgraphs formed by splitting each side rail into subpaths of length~$4\Sigma$ are \emph{sub-ladders},  
numbered from $1$ at the top to~$t+1$ at the bottom.

The $\delta(u)$ nodes at the top of the connected side rail are yellow \emph{demand nodes} (with all other connected side rail nodes being purple) and the $\delta(u)$ nodes at the bottom of the isolated side rail are brown \emph{reservoir nodes} (with all other isolated side rail nodes being teal). 
The bottom-most sub-ladder of the vertex-gadget (sub-ladder~$t+1$) contains the brown reservoir nodes and has no connections outside the vertex-gadget.
The \MSO-formula implies that for each vertex-gadget, none of the nodes in sub-ladder~$t+1$ contain tokens.
We will show that tokens leave sub-ladder~$t+1$ of a vertex-gadget for vertex $u$ through the bottom-most isolated side rail node of sub-ladder $t$ of $H_u$, which we call the \emph{transit node for $u$}.

Since a vertex $u \in V(D)$ has degree at most four, a vertex-gadget is connected to at most four arc-gadgets.
To distinguish among the nodes involved in the four connections, we assign each a \emph{type} ($a$, $b$, $c$, or $d$).
We can view the $4\Sigma$ nodes in the side rail of a sub-ladder as being the interleaving of four paths of length $\Sigma$, one of each type. 
Moving from bottom to top, each side rail consists of nodes of types $a$, $b$, $c$, $d$, $a$, $b$, $c$, $d$, and so on.  
Initially, there is a token on each brown reservoir node, for a total of $\delta(u)$ tokens for $H_u$.
\Cref{fig:logic-solution-discovery-vertex-gadget-bandwidth-hardness} illustrates a vertex-gadget, depicting sub-ladders, demand nodes, and reservoir nodes.

\subparagraph*{Arc-gadget.} The arc-gadget $H_{uv}$ contains one side rail for $u$ and one for $v$.
To represent the supply pairs, an arc-gadget $H_{uv}$ contains one \emph{sub-ladder} for each supply tuple $(x,y)$ in $L_{uv}$ (for a total of $t$ sub-ladders since unlike in vertex-gadgets, here there are no reservoir nodes). 
The sub-ladders are formed by splitting each side rail into subpaths of length $4\Sigma$. 
As in the vertex-gadgets, we number the sub-ladders from~$1$ at the top to $t$ at the bottom.
Sub-ladder $i$ corresponds to the supply pair $i$ in the instance of \textsc{Planar Arc Supply}.
In each sub-ladder, the bottom-most node in each side rail is a black \emph{rung node}.
The \emph{rung} connecting two rung nodes is a path of length six, with five intermediate spacer nodes. 

To match the lengths of the paths between nodes of each type in the vertex-gadget, in a side rail of a sub-ladder, there are three gray spacer nodes between the rung node and the first orange \emph{token-bearing node}, and three gray spacer nodes between each consecutive pair of orange nodes.
The total number of orange token-bearing nodes in a side rail of a sub-ladder will equal the number in the corresponding supply pair for that endpoint.
All remaining nodes in the side rails of the sub-ladder are spacer nodes.
Tokens are initially placed on all token-bearing nodes.
\Cref{fig:logic-solution-discovery-arc-gadget-bandwidth-hardness} shows several sub-ladders in an arc-gadget.
\\

We complete our construction of $(H, \mathcal{C})$ by connecting vertex-gadgets and arc-gadgets.
For each vertex \(u\in V(D)\), we assign a distinct type ($a$, $b$, $c$, or $d$) to each neighbor of $u$. 
If neighbor $v$ of $u$ is assigned type $a$, then there are edges only between (some) nodes of type $a$ in $H_u$ and nodes in $H_{uv}$.
For each pair of sub-ladders with the same number, one in $H_u$ and one in $H_{uv}$, we add an edge between the black rung node on the side rail for $u$ in $H_{uv}$ and the bottom-most node of type $a$ in the sub-ladder on the connected side rail in $H_u$.
Moving bottom-up in both sub-ladders, we connect the orange nodes in $H_{uv}$ one by one to distinct nodes of type $a$ in the sub-ladder on the connected side rail in $H_u$. 
\Cref{fig:logic-solution-discovery-graph-bandwidth-hardness} sketches attachments among $H_{uv}$, $H_u$, and $H_v$, where $v$ is assigned type~$d$ in the neighborhood of $u$, and $u$ is assigned type $c$ in the neighborhood of $v$.

We let $\Delta$ be the sum of all demands and then set $b = \Delta + \sum_{u \in V(D)} \delta(u) \cdot (4\Sigma(t+1) - \delta(u))$, sufficient for sliding $\Delta$ tokens from orange nodes onto and then up the corresponding vertex-gadgets and then sliding tokens from the brown nodes to the teal mirrors of now empty orange nodes.
Each token that slides from an orange node~$h$ onto a purple node $g$ in a vertex-gadget $H_u$, requires $4\Sigma(t+1) - \delta(u)$ token slides to slide the token towards a yellow node and slide a token from the brown nodes to the teal mirror of $h$.

We complete $((H, \mathcal{C}), S, b)$ with an \MSO-formula $\phi(X)$ indicating that the set $X$ satisfies the following conditions (see Appendix~\ref{appsubsec:bandwidth} for the construction in \MSO):

B1. The gray, purple, and brown nodes do not contain tokens.
    
B2. Each yellow node contains a token. 

B3. A teal node contains a token if and only if an orange node is at distance five from it and does not contain a token. 
(That is, the rail neighbor condition is satisfied.)
    
B4. If an orange node contains a token, then the orange nodes at distance four from it also do.
Otherwise, they do not. 
(That is, for each sub-ladder, orange nodes on one side rail either all contain tokens or none do.)

B5. For each connected component formed of nodes of color gray, black, and orange, there exist exactly two orange nodes not containing tokens, each at distance four from a black node, such that the black nodes are at a distance six from each other. 
No other black nodes at distance four from an orange node not containing a token exist in the component.
(That is, for each arc-gadget, we empty the orange nodes each at distance four from the black rung nodes of exactly one sub-ladder. Combined with~(B4), this condition implies that for each arc-gadget, we empty exactly one sub-ladder.)

It is easy to see that this reduction can be performed in polynomial time.
We show that the bandwidth of $H$ is upper bounded by a computable function of $\kappa$, the number of vertices and edges in $D$, and then present both directions of the proof of \Cref{thm:logic-solution-discovery-bandwidth-hardness}.

\begin{lemma}\label{lem:logic-solution-discovery-bandwidth-hardness-paramter}
$bw(H) \le 66\kappa$. 
\end{lemma}

\begin{proof}
We describe a linear ordering of the vertices of $H$ that achieves a bandwidth that is at most~$66\kappa$. 
First, we use the term \emph{segment} to refer to the $\Sigma(t+1)$ subgraphs formed by splitting each side rail for each vertex-gadget and each arc-gadget into subpaths of length four.
We number the subpaths, with segment $1$ at the top and segment $\Sigma(t+1)$ ($\Sigma t$ for an arc-gadget) at the bottom.
In an arc-gadget, each segment containing black rung nodes also contains the rung connecting them.
Other segments of the arc-gadget do not contain rungs.

For $i \in [\Sigma t]$, we add to the ordering the nodes of segment $i$ for vertex-gadget~$H_u$ for every vertex $u \in V(D)$ and arc-gadget $H_{uv}$ for every arc $uv \in A(D)$.
Then, for  $i \in [\Sigma(t+1)] - [\Sigma t]$, we add to the ordering the nodes of segment $i$ for every vertex-gadget $H_u$.
Each node in $V(H)$ is inserted exactly once.  
Moreover, for each $i \in [\Sigma(t+1)]$ we add at most \(33\kappa\) nodes to the ordering (for each vertex-gadget, four rungs of three gray nodes each, four purple nodes, and four teal nodes and for each arc-gadget, at most eight side rail nodes and five nodes of a rung). 
Each node inserted at index~\(i\) is adjacent only to nodes inserted at indices \(i-1\) or \(i+1\), hence at distance at most $66\kappa$ away in the ordering;
by the definition of bandwidth, $bw(H) \le 66\kappa$.
\end{proof}

\begin{lemma}\label{lem:logic-solution-discovery-bandwidth-hardness-forward}
If $(D,\delta,\lambda)$ is a yes-instance of \textsc{Planar Arc Supply}, then $((H,\mathcal{C}), S, b, \phi)$ is a yes-instance of \MSOD. 
\end{lemma}

\begin{proof}
For an arc $uv \in A(D)$, let $i$ be the index of the supply tuple in $L_{uv}$ that sends its supply to the vertices $u$ and $v$ in $D$.
We perform the following token slides:

\begin{enumerate}
    \item We slide the tokens at the orange nodes of sub-ladder $i$ onto their purple neighboring nodes in vertex-gadgets $H_u$ and $H_v$, totalling $\Delta$ token slides for all arcs.
    \item For each vertex-gadget and for each purple node $g$ belonging to the gadget and now having a token, we slide the closest available token from the brown nodes of the gadget towards the teal rail neighbor of $g$ and slide the token on $g$ toward the farthest empty yellow node.
    Since in Step 1 a vertex-gadget for vertex $u \in V(D)$ receives $\delta(u)$ tokens, we perform this step $\delta(u)$ times and require $4\Sigma(t+1) - \delta(u)$ token slides each time, totaling $\sum_{u \in V(D)} \delta(u) \cdot (4\Sigma(t+1) - \delta(u))$ token slides for all vertices.
\end{enumerate}

After these token slides, we can show all conditions are satisfied. For each vertex $u \in V(D)$, the yellow nodes contain tokens and the brown nodes do not, as $H_u$ receives exactly $\delta(u)$ tokens from the incident arc-gadgets onto its purple nodes (Condition~B2).
No token is moved to a gray, purple, or brown node (Condition~B1).
Condition~B3 is satisfied by the token slides from Step 2 and since a token was placed on a teal node only if it has an token-free orange node at distance five.
For each arc-gadget, we slide the tokens on all orange nodes of exactly one sub-ladder of the gadget toward the vertex-gadgets (conditions (B4) and (B5)).
The number of token slides sums to $b = \Delta + \sum_{u \in V(D)} \delta(u) \cdot (4\Sigma(t+1) - \delta(u))$.
Thus, $((H,\mathcal{C}), S, b, \phi)$ is a yes-instance of \MSOD.
\end{proof}

To prove the backward direction, we first use the construction and the constraints on the formula to show that any solution can be transformed into a solution of a restricted form. We then make use of the restrictions on the solution to demonstrate that a solution to a yes-instance $((H,\mathcal{C}), S, b, \phi)$ of \MSOD\  can be used to form a solution to $(D, \delta, \lambda)$.

We first define a \emph{minimal solution} as one that has the minimum number of token slides and of all solutions that have the minimum number of token slides, one with the minimum number of tokens that slide. 
Since any solution must satisfy conditions (B1) through (B5), tokens are moved from selected orange nodes and all brown nodes to selected teal nodes and all yellow nodes. In the absence of unneeded token slides, we obtain the following:

\begin{claim}\label{obs:cor-73}
In a minimal solution, each token that slides does so from an orange or brown node to a teal or yellow node.
\end{claim}

To better describe the movement of tokens, we define the \emph{supernode} for a vertex $u \in V(D)$ to consist of the union of the following four sets of nodes:
\begin{itemize}
\item{} \emph{the orange nodes of the supernode}, that is, each orange node that is token-free in $T$ and adjacent to a node of $H_u$;  
\item{} \emph{the teal nodes of the supernode}, that is, each teal node that is the mirror of an orange node of the supernode;
\item{} \emph{the yellow nodes of the supernode}, that is, all yellow nodes of $H_u$; and
\item{} \emph{the brown nodes of the supernode}, that is, all brown nodes of $H_u$.
\end{itemize}  

We can view the movement of a token from one node to another as either an \emph{external arc} between supernodes (if the nodes are in different supernodes) or an \emph{internal arc} within a supernode (if the nodes are in the same supernode). We show that there cannot be any external arcs (\Cref{lem:no-external}), and then that the movements within supernodes are restricted to follow certain types of paths (\Cref{lem:logic-solution-discovery-bandwidth-hardness-backward}), 
which can be used form a solution to $(D, \delta, \lambda)$.

The conditions ensure that the numbers of tokens removed from the brown and orange nodes of the supernode are equal to the numbers of tokens added to the yellow and teal nodes of the same supernode.
Thus, if a supernode ``loses'' a token through an outgoing external arc, it must also ``gain'' a token through an incoming external arc:

\begin{claim}\label{obs:balanced}
In any supernode, the number of arcs leaving orange nodes is equal to the number of arcs arriving at teal nodes, and the number of arcs leaving brown nodes is equal to the number of arcs arriving at yellow nodes.
\end{claim}

We make use of \Cref{obs:balanced} to form sequences of arcs.
An orange-teal arc where the orange and teal nodes are of the same supernode but not mirrors is called a \emph{bad internal arc}.
The arc \emph{succeeding} a brown-teal or an orange-teal arc (where the orange and teal nodes are not mirrors) is the arc leaving from the teal node's mirror.
An arc \emph{succeeding} a brown-yellow or orange-yellow arc is any external arc leaving the supernode containing the yellow node.

Before proving \Cref{lem:no-external}, we further transform a minimal solution to obtain a few useful properties.
We use simple modifications that change neither the validity of a solution nor the total number of required token slides.
We call a path traversed by a token between an orange node and a node of another color in the same supernode \emph{friendly} if it first passes through the orange node's purple neighbor.
Similarly, we call a path traversed by a token to a teal node \emph{friendly} if it last passes through the teal node's purple neighbor and connecting rung. 

\begin{claim}\label{obs:logic-solution-discovery-bandwidth-hardness-backward-order}
  None of the following modifications changes the validity of a solution or the number of token slides:
  \begin{enumerate}
    \item{} reordering token slides;
    \item{} making a path traversed by a token between an orange node and a node of another color in the same supernode friendly;
    \item{} making a path traversed by a token to a teal node after passing through an arc-gadget friendly; or
    \item{} swapping the destinations of two tokens that pass through the same node.
  \end{enumerate}
\end{claim}

We define a \emph{canonical solution} to be a minimal solution $(T, \vec{T})$ for the instance $((H,\mathcal{C}), S,$ $b, \phi)$ of \MSOD\ such that:
  \begin{enumerate}
    \item{} for each vertex $u \in V(D)$, exactly $\delta(u)$ tokens slide from the brown nodes of the vertex-gadget for $u$ to the transit node for $u$, with these token slides appearing first in $\vec{T}$; 
    \item{} if there is an external arc to a teal node and an orange-teal arc leaving from its mirror, that arc is an external arc; 
    \item{} there is no external arc to a teal node and an internal orange-yellow arc leaving from its mirror; and 
    \item{} if we have an internal orange-teal arc, the teal mirror of the arc's orange node is part of an internal brown-teal arc.
  \end{enumerate}

\pagebreak
\begin{lemma}\label{lem:logic-solution-discovery-bandwidth-hardness-backward-firstslides}
If $((H,\mathcal{C}), S, b, \phi)$ is a yes-instance of \MSOD, then it has a canonical solution $(T, \vec{T})$.
\end{lemma}

\begin{proof}
Given conditions (B1) and (B3), we know that the tokens on the brown nodes of the vertex-gadget for $u$ for a vertex $u \in V(D)$ do not end up on any node of the sub-ladder $t+1$ for the vertex-gadget for $u$. 
Thus, these tokens will slide (or can be rerouted to slide without increasing the number of token slides) at least towards the transit node for $u$.
It also follows from \Cref{obs:logic-solution-discovery-bandwidth-hardness-backward-order}, Point 1, that the token slides that appear first in $\vec{T}$ are the token slides of the $\delta(u)$ tokens from the brown nodes of the vertex-gadget for $u$ to the transit node for $u$, for each $u \in V(D)$.
We have now satisfied Condition 1 of the canonical solution.

Suppose we have an external arc to a teal node and a bad orange-teal internal arc leaving from its mirror. 
By \Cref{obs:logic-solution-discovery-bandwidth-hardness-backward-order}, points (2) and (3), the paths traversed by the associated tokens can be made friendly. 
Since the paths traverse the same purple rail neighbor of the teal node, by \Cref{obs:logic-solution-discovery-bandwidth-hardness-backward-order}, Point 4, the destinations of the two tokens can be swapped. 
We repeat the argument until we find an external arc to a teal node with no bad orange-teal internal arc leaving from its mirror or an external arc to a yellow node (Condition 2).

If there is an external arc to a teal node and an internal orange-yellow arc leaving from its mirror, 
by \Cref{obs:logic-solution-discovery-bandwidth-hardness-backward-order}, points (2) and (3), the paths traversed by tokens can be made friendly. 
Since the paths traverse the same purple rail neighbor of the teal node, by \Cref{obs:logic-solution-discovery-bandwidth-hardness-backward-order}, Point 4, we can swap the destinations of the two tokens, satisfying Condition 3. 

Finally, suppose we have an internal orange-teal arc. 
Condition 2 is violated if the teal mirror~$tm$ of this arc's orange node is part of an external arc.
If $tm$ is part of an internal orange-teal arc, by \Cref{obs:logic-solution-discovery-bandwidth-hardness-backward-order}, points (2) and (3), we can make the paths traversed by the tokens of both internal orange-teal arcs friendly. 
Since both friendly paths traverse the same purple node (the rail neighbor of $tm$), by \Cref{obs:logic-solution-discovery-bandwidth-hardness-backward-order}, Point 4, we can swap the destinations of the two tokens. 
We repeat the argument until the teal node of the resulting internal orange-teal arc is part of an internal brown-teal arc, thus satisfying Condition 4.
\end{proof}

We show in \Cref{lem:no-external} that any cycle of external arcs can be removed, resulting in a smaller number of token slides, and hence contradicting the minimality of a canonical solution. We do so by starting with an external arc and using it to build a cycle.
The following claim is a direct consequence of \Cref{obs:balanced}.

\begin{claim}\label{obs:arctocycle}
If there is a single external arc in a canonical solution, there must be a cycle with one or more external arcs.
\end{claim}  

We capture the movement of tokens between supernodes using an auxiliary graph~$D^\pm$ containing supernodes and the external arcs between them.  
We show that no external arcs are possible by demonstrating that any cycle in $D^\pm$ can be transformed into a smaller solution.

We construct cycles in $D^\pm$ by considering the colors of the nodes (of supernodes) from and to which the token is moving; by \Cref{obs:cor-73} each arc is orange-teal, orange-yellow, brown-teal, or brown-yellow. 
Using a proof by contradiction, we show that if a canonical solution contains an external arc, we can find in $D^\pm$ a cycle consisting entirely of orange-teal external arcs or a cycle built out of smaller pieces. 
Specifically, we call a sequence of external arcs a \emph{brown-yellow segment} if the first arc in the sequence starts at (i) a brown node or (ii) an orange node whose teal mirror is part of an internal brown-teal arc, or (iii) an orange node whose teal mirror is part of an internal orange-teal arc whose teal node is part of an internal brown-teal arc, and the last arc in the sequence ends at a yellow node; in the sequence, each arc is followed by a succeeding arc.

\begin{claim}\label{obs:brownyellow}
If we can find a single brown-yellow segment in $D^\pm$, by \Cref{obs:arctocycle} we must find a cycle formed of one or more brown-yellow segments.
\end{claim}

\noindent For both types of cycles, we show that the token slides can be modified to produce a solution with a smaller number of slides, contradicting the assumption that the solution is canonical.

We are now ready to prove the following lemma:

\begin{lemma}\label{lem:no-external}
No canonical solution contains an external arc.
\end{lemma}

\begin{proof}
  
Suppose to the contrary that a canonical solution contains an external arc. We consider two possible cases:


\noindent{\bf Case 1:} There is a cycle in $D^\pm$ formed entirely of orange-teal arcs.

  Since the arc succeeding an external orange-teal arc starts at the teal node's orange mirror and is external, the cycle contains matched pairs of mirrors.
  We create a new solution in which, for each pair, the token slides from the orange node to its mirror.
  Since we no longer require tokens to move along rungs of an arc-gadget (as was required in the solution that forms the cycle in $D^\pm$), we have reduced the number of token slides while preserving the number of tokens that slide. 
  Doing so violates the assumption that the solution was canonical, as it was not minimal.

\noindent{\bf Case 2:}  There is no cycle in $D^\pm$ formed entirely of orange-teal arcs.

  It suffices to demonstrate that we can use any of the other types of external arcs (brown-yellow, brown-teal, or orange-yellow) to form a brown-yellow segment. 
  By \Cref{obs:brownyellow} we can then conclude that we have a cycle formed of brown-yellow segments (the cycle gives us matched pairs where each pair consists of a brown and a yellow node of a supernode in the cycle). 
  Since in the cycle, each external arc to a teal node is followed by an external arc leaving from its mirror, the cycle also only contains matched pairs of mirrors.
  We can then form a solution using a smaller number of token slides by, for each brown-yellow pair, sending the token on the brown node to the matched yellow node in the same supernode, and for each pair of matched mirrors, send the token on the orange node to its matched teal mirror. 
  Since we have decreased the number of token slides (by avoiding the traversal of rungs of at least one arc-gadget), the solution was not canonical.
  
  We show that in each remaining case we can obtain a brown-yellow segment.

\noindent{\bf Case 2a:} There is a brown-yellow arc in $D^\pm$.

  In this case, the single arc is a brown-yellow segment.

\noindent{\bf Case 2b:} There is a brown-teal arc in $D^\pm$.

  If there is a brown-teal arc in $D^\pm$, then the succeeding arc starts at the orange mirror of the teal node.
  If the succeeding arc is an orange-yellow arc, we either have a brown-yellow segment (if the brown and yellow nodes are in different supernodes) or have completed a cycle (if the brown and yellow nodes are in the same supernode).
  Otherwise, we add the orange-teal arc, and continue adding succeeding orange-teal arcs until we eventually reach an arc with a yellow node (and hence form either a brown-yellow segment or a cycle). 

\noindent{\bf Case 2c:}  There is an orange-yellow arc in $D^\pm$. 

    If there is an orange-yellow arc in $D^\pm$, it is the succeeding arc of either a brown-teal arc or an orange-teal arc (possible internal).
    As the case where we have a brown-teal external arc was covered in Case 2b, we consider the others. 
    Consider an orange-teal external arc, and observe that it too is the succeeding arc of either a brown-teal arc or an orange-teal arc, which can be internal. 
    We continue the process of moving backward through external arcs until we find an internal orange-teal arc or an internal brown-teal arc.
    If we have an internal brown-teal arc, we form either a brown-yellow segment or a cycle.
    If we have an internal orange-teal arc, in a canonical solution, the teal mirror of its orange node is part of an internal brown-teal arc. 
    Thus, we also form either a brown-yellow segment or a cycle.
\end{proof}  

We are now ready to complete our proof of the backward direction.

\begin{lemma}\label{lem:logic-solution-discovery-bandwidth-hardness-backward}
If $((H,\mathcal{C}), S, b, \phi)$ is a yes-instance of \MSOD, then \mbox{$(D,\\ \delta, \lambda)$} is a yes-instance of \textsc{Planar Arc Supply}.
\end{lemma}

\begin{proof}
Let $(T, \vec{T})$ be a canonical solution.
Each token on an orange or brown node of a supernode is destined for a teal or yellow node of the same supernode (\Cref{lem:no-external}), and no other token slides in~$\vec{T}$ (\Cref{obs:cor-73}). 
Thus, for each vertex $u \in V(D)$, the tokens on the orange nodes of supernode $u$ must first slide to a purple node (as all the nodes in $T$ are inside the vertex-gadgets). 
From \Cref{obs:logic-solution-discovery-bandwidth-hardness-backward-order}, we can assume that the tokens on these orange nodes slide first to their purple rail neighbors using one token slide each.

For each vertex $u \in V(D)$, we analyze the token slides within supernode $u$ needed to fill the yellow nodes.
Consider a single token on a brown node that is closest to the transit node for $u$.
If this token is destined for the token-free yellow node farthest from the bottom-most purple rail node of sub-ladder $1$ of $H_u$, it will require $4\Sigma(t+1) - \delta(u) + 4$ token slides.

Consider the following alternative token movement pattern.
A token on a brown node that is closest to the transit node is destined for a token-free teal node $y$.
The token on the purple node~$x$ (with $y$ as a rail neighbor), which came from the orange node adjacent to $x$, is destined for the token-free yellow node thatis farthest from the bottom-most purple node of sub-ladder $1$ of $H_u$.
This movement pattern requires $4\Sigma(t+1) - \delta(u) + 1$ token slides to get one token to a yellow node.
This is the minimum number of token slides possible to get a token to a yellow node from an orange or brown node (the only nodes from which a token can slide into a yellow node, as a consequence of \Cref{lem:no-external} and \Cref{obs:cor-73}), and the only manner in which this minimum can be achieved.

Given that $b = \Delta + \sum_{u \in V(D)} \delta(u) \cdot (4\Sigma(t+1) - \delta(u))$ and for each $u \in V(D)$, $\delta(u)$ tokens must slide into the yellow nodes of supernode $u$ incurring each at minimum $4\Sigma(t+1) - \delta(u) + 1$ token slides, we see that tokens that end up on the yellow nodes must slide using the described token movement pattern.
Consequently, for each vertex $u \in V(D)$, the number of orange nodes of supernode $u$ is $\delta(u)$.
Since $T$ satisfies conditions (B4) and (B5), for each arc-gadget $H_{uv}$ corresponding to an arc $uv \in A(D)$ incident to $u$ and $v$, the orange token-bearing nodes of exactly one sub-ladder are token-free in $T$. 
Thus, the orange nodes correspond to exactly one selected integer $i \in [t]$ for each arc-gadget representing an arc incident to $u$. 
Hence, for each arc, exactly one index $i \in [t]$ is selected and its tuple supplies exactly its first value to the tail vertex and its second value to the head vertex, and this supply equals the demand of the vertices. 
Thus, $(D,\delta,\lambda)$ is a yes-instance of \textsc{Planar Arc Supply}. 
\end{proof}

\appendix
\section{Constructing the formulas}
\subsection{Modulator to Stars and Modulator to Paths Numbers}\label{appsubsec:modulators}
Let us first write a formula $\phi(X)$ for conditions (S1), (S2), and (S3).
We assume the following predicate, $\text{Col}(x):= \lor_{i\in[5]} C_i(x)$, which holds if $x$ is a colored vertex in $(H,\mathcal{C})$ of~\Cref{sec:logic-solution-discovery-modulator-hardness} is available.
We now express each of the conditions:
\begin{flalign*}
\text{S1$(X)$:=} \quad & \forall t \in V(H) : \Bigg( C_3(t) \rightarrow \Bigg[ 
    \exists p \in V(H) : \bigg( 
        E(t,p) \wedge 
        \forall z \in V(H) : \Big( 
            \big(E(p,z)   
            \\ 
            & \qquad\qquad\qquad
            \wedge \neg \text{Col}(z)\big) \rightarrow \neg X(z) 
        \Big) 
    \bigg) \wedge 
    \forall q \in V(H): \bigg( 
        \big(E(t,q) \wedge 
            \\
    & \qquad\qquad\qquad
    (q \ne p)\big) \rightarrow  \Big( \forall z \in V(H): 
    \big(E(q,z)
            \wedge \neg \text{Col}(z) \big) 
            \\
    & \qquad\qquad\qquad
    \rightarrow X(z) 
        \Big) 
    \bigg) \Bigg] 
\Bigg), &
\end{flalign*}

\begin{flalign*}
\text{S2$(X)$:=} \quad & \forall t \in V(H): \Bigg( C_4(t) \rightarrow \Bigg[ 
    \exists p \in V(H): \bigg( 
        E(t,p) \wedge 
        \forall z \in V(H) : \Big( 
            \big( E(p,z) 
            \\
    & \qquad\qquad\qquad
           \wedge  \neg \text{Col}(z) \big) \rightarrow  X(z) 
        \Big) 
    \bigg) \wedge 
    \forall q \in V(H) : \bigg( 
        \big(E(t,q) \wedge 
        \\
    & \qquad\qquad\qquad
    (q \ne p)\big) \rightarrow  \Big( \forall z \in V(H) : 
        \big( 
            E(q,z) \wedge \neg \text{Col}(z) \big) \\
    & \qquad\qquad\qquad
    \rightarrow \neg X(z) 
        \Big) 
    \bigg) \Bigg] 
\Bigg), &
\end{flalign*}

\begin{flalign*}
\text{S3$(X)$:=} \quad & \forall x \in V(H): \Big( \text{Col}(x) \rightarrow \neg X(x) \Big), &
\end{flalign*}

\noindent and finally set $\phi(X):=$ S1$(X)$ $\land$ S2$(X)$ $\wedge$ S3$(X)$.

For a formula $\phi(X)$ with conditions (P1,) (P2), (P3), and (P4), we also assume that the predicate $\text{Dist2NoC345}(x,z)$ is available. $\text{Dist2NoC345}(x,z)$, which means that $z$ is a vertex in the distance-$2$ neighborhood of $x$ and is reachable from $x$ through a vertex not in $C_3$, $C_4$, or $C_5$, is written as: $\exists y: \text{Edge}(x,y) \wedge  \text{Edge}(y,z) \wedge \neg C_3(y) \wedge \neg C_4(y) \wedge \neg C_5(y)$. Now we write

\begin{flalign*}
\text{P1$(X)$:=} \quad & \forall t \in V(H): \Big( C_3(t) \rightarrow 
    \exists p \in V(H): \big( E(t,p) \wedge \neg \text{Col}(p) \wedge \neg X(p) \big) 
\Big), &
\end{flalign*}

\begin{flalign*}
\text{P2$(X)$:=} \quad & \forall t \in V(H): \Big( C_4(t) \rightarrow 
    \exists p \in V(H) : \big( E(t,p) \wedge \neg \text{Col}(p) \wedge X(p) \big) 
\Big), &
\end{flalign*}

\begin{flalign*}
\text{P3$(X)$:=} \quad & \forall x \in V(H) : \Bigg( \neg \text{Col}(x) \rightarrow \Big[ \Big(
    X(x) \rightarrow 
        \forall z \in V(H) : \big( 
            \text{Dist2NoC345}(x,z) 
    \\
    & \qquad\qquad\qquad
    \rightarrow  X(z) 
    \big) \Big)  \wedge 
    \Big(
    \neg X(x) \rightarrow 
        \forall z \in V(H) : 
            \\
    & \qquad\qquad\qquad
             \big( \text{Dist2NoC345}(x,z)  \rightarrow \neg X(z) 
        \big) \Big)
\Big] \Bigg), &
\end{flalign*}

\noindent and set P4$(X)$:= S3$(X)$ and $\phi(X):=$ P1$(X)$ $\land$ P2$(X)$ $\wedge$ P3$(X)$ $\wedge$ P4$(X)$.

\subsection{Twin Cover}\label{appsubsec:twincover}
Let us write a formula $\phi(X)$ for conditions (T1), (T2), (T3), (T4), and (T5).
We assume the following predicate, $\text{Col}(x) := \lor_{i\in[6]} C_i(x)$, which holds if $x$ is a colored vertex in $(H,\mathcal{C})$ of~\Cref{sec:twincover} is available.
We now express each of the conditions:
\begin{flalign*}
\text{T1$(X)$:=}\quad
& \forall x\in V(H):\big( C_4(x)\;\rightarrow\; X(x) \big), &
\end{flalign*}

\begin{flalign*}
\text{T2$(X)$:=}\quad
& \forall x \in V(H):\Bigg( C_5(x)\;\rightarrow\;\Big[ \Big(
        \big( \exists u\in V(H): C_2(u)\wedge E(x,u) \big)\;\wedge\;
        \\ &\quad \big( \exists v\in V(H): C_3(v)\wedge 
        E(x,v) \big) \Big)
        \;\;\rightarrow\;\;
        \neg X(x)
\Big] \Bigg), &
\end{flalign*}

\begin{flalign*}
\text{T3$(X)$:=}\quad
& \forall x\in V(H):\Bigg( C_2(x)\;\rightarrow\;\Bigg[ \exists p\in V(H):\bigg(
        C_6(p)\wedge E(x,p)\wedge  
\\
&\qquad
        \forall z \in V(H): \,\Big(\big[\big(C_6(z)\wedge E(p,z)\big)\rightarrow X(z)\big]\;\wedge\;
        \big[\big(C_5(z)\wedge E(p,z)\big)
        \\
        &\qquad 
        \rightarrow \neg X(z)\big]\Big)\bigg) \wedge
        \forall q\in V(H):\bigg(
            C_6(q)\wedge E(x,q) \wedge \neg E(p,q)  
            \\
        &\qquad  
      \wedge (q\neq p) \rightarrow \forall z \in V(H): \Big(\big[\big(C_6(z)\wedge E(q,z)\big)\rightarrow \neg X(z)\big]\;\wedge\;
        \\
        &\qquad  
        \big[\big(C_5(z)\wedge E(q,z)\big)
        \rightarrow X(z)\big]\Big)
        \bigg)
\Bigg] \Bigg), &
\end{flalign*}

\begin{flalign*}
\text{T4$(X)$:=}\quad
& \forall x\in V(H):\Bigg( C_3(x)\;\rightarrow\;\Bigg[ \exists p\in V(H):\bigg(
        C_6(p)\wedge E(x,p)\wedge 
\\
&\qquad
      \forall z \in V(H): \,\Big(\big[\big(C_6(z)\wedge E(p,z)\big)\rightarrow X(z)\big]\;\wedge\;
        \big[\big(C_5(z)\wedge E(p,z)\big)
        \\
        &\qquad
        \rightarrow \neg X(z)\big]\Big)\bigg)\wedge
        \forall q\in V(H):\bigg(
            C_6(q)\wedge E(x,q)\wedge \neg E(p,q)  
            \\
        &\qquad  
      \wedge (q\neq p) \rightarrow \forall z \in V(H): \Big(\big[\big(C_6(z)\wedge E(q,z)\big)\rightarrow \neg X(z)\big]\;\wedge\;
        \\
        &\qquad  
        \big[\big(C_5(z)\wedge E(q,z)\big)
        \rightarrow X(z)\big]\Big)
        \bigg)
\Bigg] \Bigg), &
\end{flalign*}

\begin{flalign*}
\text{T5$(X)$:=}\quad
& \forall x\in V(H):\Big( \big(C_1(x) \lor C_2(x) \lor C_3(x)\big) \;\rightarrow\; \neg X(x) \Big), &
\end{flalign*}

\noindent and set $\phi(X):=$ T1$(X)$ $\land$ T2$(X)$ $\wedge$ T3$(X)$ $\wedge$ T4$(X)$ $\wedge$ T5$(X)$.

\subsection{Bandwidth}\label{appsubsec:bandwidth}
Let us write a formula $\phi(X)$ for conditions (B1), (B2), (B3), (B4), and (B5).
We assume that the predicates $\mathrm{Col}_{3,6,7}(x)$, $\mathrm{Dist}_k(x,y)$, $\mathrm{Con}_{3,6,7}(x,y)$, and $\mathrm{CTFC6}(x)$ are available. 
\subparagraph*{Color test}
$x$ is a vertex of color gray ($C_3$) or black ($C_6$) or orange ($C_7$) if and only if this holds:
\[
  \mathrm{Col}_{3,6,7}(x):= C_3(x) \lor C_6(x) \lor C_7(x).
\]

\subparagraph*{Distance}
$x$ and $y$ are vertices at distance \(k\in[5]\) if and only if this holds:
\[
  \mathrm{Dist}_k(x,y):= \exists z_1,\ldots,z_{k-1}\;\Bigl( E(x,z_1)\wedge E(z_1,z_2)\wedge\cdots\wedge E(z_{k-2},z_{k-1})\wedge E(z_{k-1},y)\Bigr).
\]

\subparagraph*{Color-restricted connectivity}
The vertices $x$ and $y$ lie in the same connected component formed of vertices of colors $C_3$, $C_6$, and $C_7$ only if and only if this holds:
\[
\begin{aligned}
\mathrm{Con}_{3,6,7}(x,y) :=& (x=y)\;\lor\;
\bigg(
   \mathrm{Col}_{3,6,7}(x)\wedge\mathrm{Col}_{3,6,7}(y)\;\wedge \neg\exists S \subseteq V(H): \Bigl(
        S(x)\wedge\\
        &
        \neg S(y)\;\wedge\;\forall z \in V(H): \bigl(S(z)\rightarrow 
        \mathrm{Col}_{3,6,7}(z)\bigr)\;\wedge \forall u, v \in V(H):\bigl(
              S(u)\\\
        &
              wedge\mathrm{Col}_{3,6,7}(u)\wedge E(u,v)\wedge\mathrm{Col}_{3,6,7}(v)
              \;\rightarrow\; S(v)
        \bigr)
   \Bigr)
\bigg).
\end{aligned}
\]

\subparagraph*{Token-free \(C_7\) vertex four steps from a \(C_6\) vertex}
$x$ is a token-free orange ($C_7$) vertex that is at distance four to a black ($C_6$) vertex denoted $y$ if and only if this holds: 
\[
  \mathrm{CTFC6}(x,y):= C_3(x)\wedge\neg X(x)\wedge C_1(y)\wedge \mathrm{Dist}_4(x,y).
\]

Now we set $\phi(X):=$ B1$(X)$ $\land$ B2$(X)$ $\wedge$ B3$(X)$ $\wedge$ B4$(X)$ $\wedge$ B5$(X)$ where

\begin{flalign*}
\text{B1}(X) &:= 
  \forall x \in V(H): \Bigl(\bigl(C_3(x) \lor C_1(x) \lor  C_5(x)\bigr)\;\rightarrow\;\neg X(x)\Bigr), &
\end{flalign*}

\begin{flalign*}
\text{B2}(X) &:= 
  \forall x \in V(H): \bigl(C_2(x)\;\rightarrow\;X(x)\bigr), &
\end{flalign*}

\begin{flalign*}
\text{B3}(X) &:= 
\forall x \in V(H): \Bigl[
      C_4(x)\;\rightarrow\;
      \Bigl(
          X(x)\;\leftrightarrow\;
          \exists y \in V(H):
          \bigl[
              \mathrm{Dist}_5(x,y)\;\wedge\;C_7(y) \\
        &\qquad \;\wedge\;\neg X(y)
          \bigr]
      \Bigr)
  \Bigr], &
\end{flalign*}

\begin{flalign*}
\text{B4}(X) &:= 
  \forall x \in V(H): \Bigl[C_7(x)\;\rightarrow\;
      \forall y \in V(H): \Bigl(\bigl[\mathrm{Dist}_4(x,y)\wedge C_7(y)\bigr]\;\rightarrow\;\\ &\qquad
          \bigl[\,X(x)\leftrightarrow X(y)\,\bigr]\Bigr)\Bigr], &
\end{flalign*}

\begin{flalign*}
\text{B5}(X) &:= 
  \forall x \in V(H): \bigg[
      \Big(C_6(x) \lor C_7(x)\Big)
      \;\rightarrow \exists y_1 ,y_2, u_1, u_2 \in V(H): \bigg( 
      \\
  &\qquad
        \mathrm{CTFC6}(y_1,u_1)\;\wedge\;
        \mathrm{CTFC6}(y_2,u_2)\;\wedge\;
        y_1 \neq y_2\;\wedge 
        u_1 \neq u_2\;\wedge\;
        \\[4pt]
  &\qquad
            \mathrm{Dist}_6(u_1,u_2)\;\wedge 
            \mathrm{Con}_{3,6,7}(x,y_1)\;\wedge\;
            \mathrm{Con}_{3,6,7}(x,y_2)\;\wedge\;
            \mathrm{Con}_{3,6,7}(x,u_1)\; \\[4pt]
 &\qquad
            \wedge\; \mathrm{Con}_{3,6,7}(x,u_2) \;\wedge 
            \forall v \in V(H):\Bigl(
            \mathrm{Con}_{3,6,7}(x,v)\;\wedge\;
            C_6(v)\;\wedge\;\\[4pt]
 &\qquad
            \exists w \in V(H):\bigl(
                \mathrm{CTFC6}(w,v)\;
            \bigr)
            \;\rightarrow\;
            (v = u_1 \lor v = u_2)
        \Bigr)
    \bigg)
  \bigg]. &
\end{flalign*}

\bibliographystyle{plain}
\bibliography{ref}

\tikzset{cut mark/.style={
    densely dashed,    
    line width=.5pt,
    shorten <=2pt,    
    shorten >=2pt       
}}

\tikzset{
  vbrace/.style={decorate,
                 decoration={brace, amplitude=4pt}, 
                 thick}
}

 \begin{figure}[p]
  \centering
  \begin{tikzpicture}[>=stealth, x=2cm, y=0.1cm]
    \path[use as bounding box] (-0.6,10) rectangle (1.6,-130);
    \draw[thick] (0,20) -- (0,0);
    \draw[thick] (0,-60) -- (0,-65);
    \draw[thick] (0,-70) -- (0,-85);
    \draw[thick] (0,-90) -- (0,-105);
    \draw[thick] (0,-110) -- (0,-120);
    \draw[thick] (1,20) -- (1,0);
    \draw[thick] (1,-60) -- (1,-65);
    \draw[thick] (1,-70) -- (1,-85);
    \draw[thick] (1,-90) -- (1,-105);
    \draw[thick] (1,-110) -- (1,-120);
    \foreach \n in {-20,...,0} {
      \draw (0,-\n) -- (1,-\n);
      \foreach \f in {0.25,0.5,0.75}
        \fill[gray] (\f,-\n) circle[radius=0.04cm];
    }
    \foreach \n in {60,...,65} {
      \draw (0,-\n) -- (1,-\n);
      \foreach \f in {0.25,0.5,0.75}
        \fill[gray] (\f,-\n) circle[radius=0.04cm];
    }
    \foreach \n in {70,...,85} {
      \draw (0,-\n) -- (1,-\n);
      \foreach \f in {0.25,0.5,0.75}
        \fill[gray] (\f,-\n) circle[radius=0.04cm];
    }
    \foreach \n in {90,...,105} {
      \draw (0,-\n) -- (1,-\n);
      \foreach \f in {0.25,0.5,0.75}
        \fill[gray] (\f,-\n) circle[radius=0.04cm];
    }
    \foreach \n in {110,...,120} {
      \draw (0,-\n) -- (1,-\n);
      \foreach \f in {0.25,0.5,0.75}
        \fill[gray] (\f,-\n) circle[radius=0.04cm];
    }
    \draw[dotted,violet,thick] (0,-85) -- (0,-90);
    \draw[dotted,violet,thick] (0,-106) -- (0,-109);
    \draw[dotted,violet,thick] (0,-65) -- (0,-70);
    \draw[dotted,violet,thick] (0,0) -- (0,-60);
    \foreach \n in {-19,...,0}    \fill[violet]      (0,-\n) circle[radius=0.04cm];
    \foreach \n in {60,...,65}    \fill[violet]      (0,-\n) circle[radius=0.04cm];
    \foreach \n in {70,...,85}    \fill[violet]      (0,-\n) circle[radius=0.04cm];
    \foreach \n in {110,...,120}   \fill[violet]      (0,-\n) circle[radius=0.04cm];
    \foreach \n in {90,...,105}   \fill[violet]      (0,-\n) circle[radius=0.04cm];
   
    \draw[dotted,teal,thick] (1,-85) -- (1,-90);
    \draw[dotted,teal,thick] (1,-106) -- (1,-109);
    \draw[dotted,teal,thick] (1,-65) -- (1,-70);
    \draw[dotted,teal,thick] (1,0) -- (1,-60);
    \foreach \n in {-20,...,0}    \fill[teal]  (1,-\n) circle[radius=0.04cm];
    \foreach \n in {60,...,65}    \fill[teal]      (1,-\n) circle[radius=0.04cm];
    \foreach \n in {70,...,85}    \fill[teal]  (1,-\n) circle[radius=0.04cm];
    \foreach \n in {110,...,114}   \fill[teal]  (1,-\n) circle[radius=0.04cm];
    \foreach \n in {90,...,105}   \fill[teal]  (1,-\n) circle[radius=0.04cm];
    \foreach \n in {-20,...,-15}   \fill[yellow] (0,-\n) circle[radius=0.04cm];
    \foreach \n in {115,...,120}   \fill[brown] (1,-\n) circle[radius=0.04cm];
    \node[anchor=north west,font=\scriptsize,color=green!60!black] at (-1.3,15) {node of type $d$};
    \draw[->,color=green!60!black] (-0.5,15) -- (0,20);
    \draw[decorate,decoration={brace,amplitude=3pt,mirror},color=green!60!black] 
        (-0.1,20) -- (-0.1,15);
    \draw[->,color=green!60!black] (-1.3,18) -- (-0.2,18);
    \node[anchor=east,font=\scriptsize,color=green!60!black] at (-1.35,18) {demand nodes};
    \node[anchor=north west,font=\scriptsize,color=green!60!black] at (-1.3,-105)  {transit node for $u$};
    \draw[->,color=green!60!black] (-0.5,-105) -- (0,-100);
    \node[anchor=north west,font=\scriptsize,color=green!60!black] at (-1.3,-125)  {node of type $a$};
    \draw[->,color=green!60!black] (-0.5,-125) -- (0,-120);
    \node[anchor=north west,font=\scriptsize,color=green!60!black] at (-1.3,-110)  {node of type $b$};
    \draw[->,color=green!60!black] (-0.5,-114) -- (0,-119);

    \draw[decorate,decoration={brace,amplitude=3pt},color=green!60!black] 
        (1.1,-115) -- (1.1,-120);
    \draw[->,color=green!60!black] (2,-117.5) -- (1.25,-117.5);
    \node[anchor=west,font=\scriptsize,color=green!60!black] at (2.05,-117.5) {reservoir nodes};
    \foreach \y in {20,0}{
    \draw[cut mark] (-1.15,\y) -- (2.15,\y);
    }
    \foreach \y in {-60,-80,-100,-120}{
        \draw[cut mark] (-1.15,\y) -- (2.15,\y);
    }
    \foreach \top/\pt in {19/1, -61/{t-1}, -81/{t}, -101/{t+1}}{%
        \draw[vbrace,color=green!60!black] (1.35,\top) -- ++(0,-19)
        node[midway,right=6pt,font=\scriptsize,color=green!60!black] {sub-ladder $\pt$};
    }
  \end{tikzpicture}
  \caption{A sketch of the isolated and connected side rails for some vertex $u \in V(D)$ as discussed in \Cref{thm:logic-solution-discovery-bandwidth-hardness}. 
  Demand nodes at the top of the connected side rail are colored yellow; reservoir nodes at the bottom of the isolated side rail are brown. 
  Reservoir nodes contain tokens.
  The thin dotted lines indicate additional nodes that are omitted for clarity.}
  \label{fig:logic-solution-discovery-vertex-gadget-bandwidth-hardness}
\end{figure}
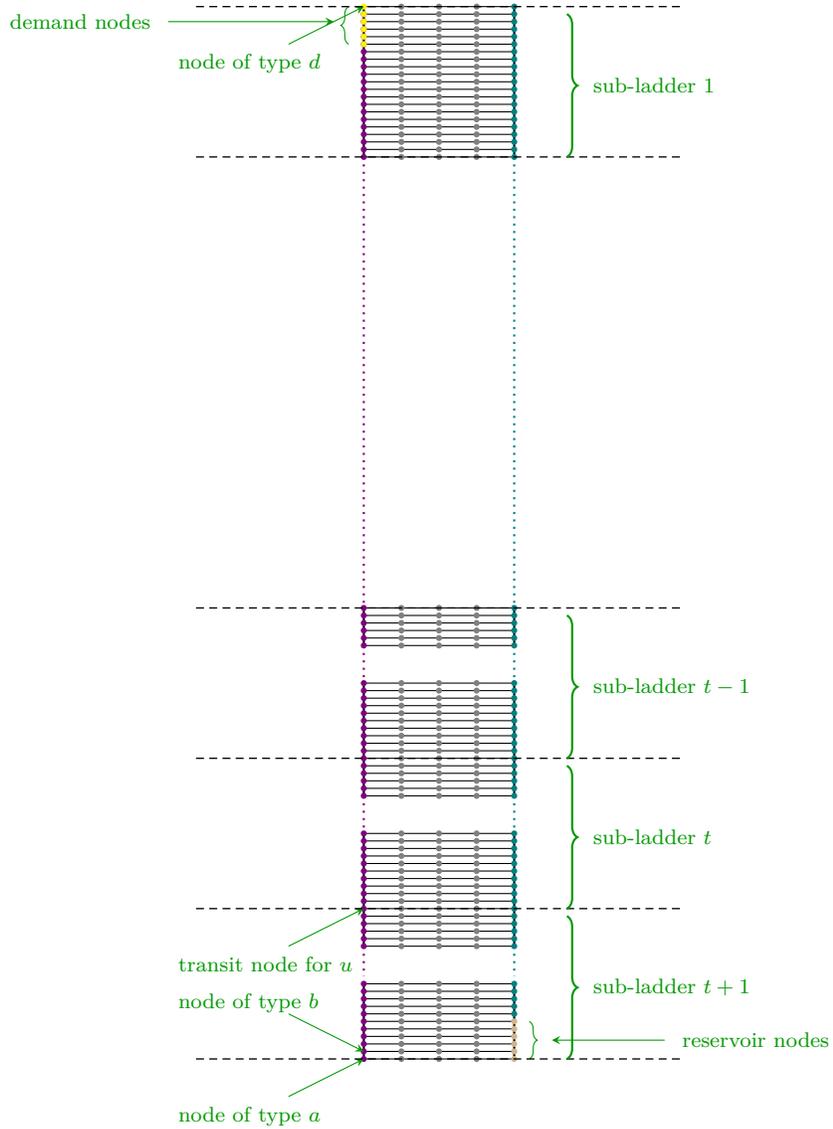   

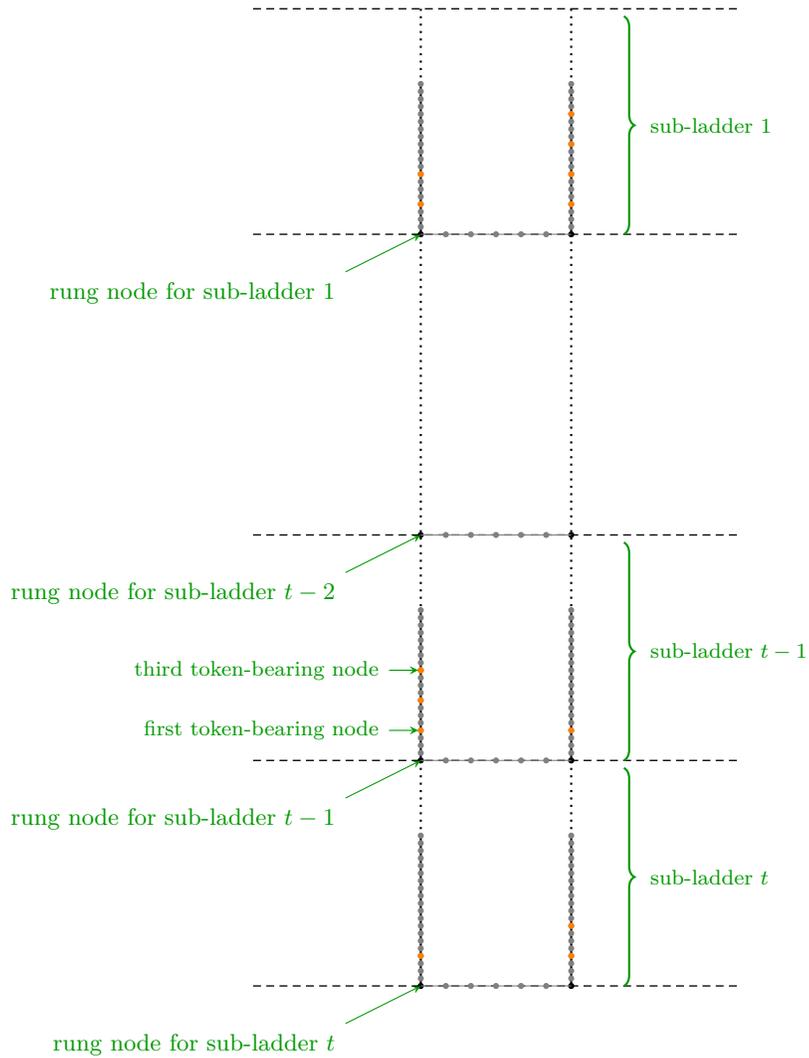
\begin{figure}[h]
  \centering
\begin{tikzpicture}[>=stealth, x=2cm, y=0.1cm]
  \path[use as bounding box] (-2.1,10) rectangle (2.1,-130);
  \draw[thick] (-0.5, 0) -- (-0.5, -20);
  \draw[thick] (-0.5, -70) -- (-0.5, -90);
  \draw[thick] (-0.5,-100) -- (-0.5,-120);
  \draw[dotted,thick] (-0.5,10)  -- (-0.5,-10);
  \draw[dotted,thick] (-0.5,-20)  -- (-0.5,-70);
  \draw[dotted,thick] (-0.5,-90)  -- (-0.5,-100);
  \foreach \y in {0,-1,...,-20}    \fill[gray] (-0.5,\y)   circle[radius=0.04cm];
  \foreach \y in {-70,-71,...,-90}    \fill[gray] (-0.5,\y)   circle[radius=0.04cm];
  \foreach \y in {-100,-101,...,-120} \fill[gray] (-0.5,\y)   circle[radius=0.04cm];
  \foreach \y/\lbl in {-20/{1}, -60/{t-2}, -90/{t-1}, -120/{t}} {
  \fill[black] (-0.5,\y) circle[radius=0.04cm];
  \node[anchor=north east, font=\footnotesize, color=green!60!black] at (-1,\y-5) {rung node for sub-ladder $\lbl$};
  \draw[->,color=green!60!black] (-1,\y-5) -- (-0.5,\y); 
  }
  \foreach \y in {-116, -82, -78, -16, -12}        
    \fill[orange] (-0.5,\y) circle[radius=0.04cm];
  \fill[orange] (-0.5,-86) circle[radius=0.04cm];
  \node (olab-86) [left=12pt,font=\scriptsize,color=green!60!black]
        at (-0.5,-86) {first token-bearing node};
  \draw[->,shorten >=1pt,color=green!60!black] (olab-86.east) -- (-0.5,-86);
  \fill[orange] (-0.5,-78) circle[radius=0.04cm];
  \node (olab-78) [left=12pt,font=\scriptsize,color=green!60!black]
        at (-0.5,-78) {third token-bearing node};
  \draw[->,shorten >=1pt,color=green!60!black] (olab-78.east) -- (-0.5,-78);
  \draw[thick] (0.5, 0) -- (0.5, -20);
  \draw[thick] (0.5, -70) -- (0.5, -90);
  \draw[thick] (0.5,-100) -- (0.5,-120);
  \draw[dotted,thick] (0.5,10)  -- (0.5,-10);
  \draw[dotted,thick] (0.5,-20)  -- (0.5,-70);
  \draw[dotted,thick] (0.5,-90)  -- (0.5,-100);
  \foreach \y in {0,-1,...,-20}    \fill[gray] (0.5,\y)   circle[radius=0.04cm];
  \foreach \y in {-70,-71,...,-90}    \fill[gray] (0.5,\y)   circle[radius=0.04cm];
  \foreach \y in {-100,-101,...,-120} \fill[gray] (0.5,\y)   circle[radius=0.04cm];
  \foreach \y/\lbl in {-20/{-2}, -60/{-1}, -90/{}, -120/{}} {
    \fill[black] (0.5,\y) circle[radius=0.04cm];
    \node (ylab\y) [right=12pt,font=\scriptsize,color=black]
          at (0.5,\y) {};
  }
  \foreach \y in {-116, -16, -12, -8, -4}        
    \fill[orange] (0.5,\y) circle[radius=0.04cm];
  \foreach \y/\lbl/\llbl in {-112/{}/{2}, -86/{-1}/{3}} {
    \fill[orange] (0.5,\y) circle[radius=0.04cm];
    \node (olab\y) [right=12pt,font=\scriptsize,color=orange]
          at (0.5,\y) {};
  }
  \foreach \y in {10,-20,-60,-90,-120}{
    \draw[cut mark] (-1.65,\y) -- (1.65,\y);
    }
    \foreach \top/\pt in {9/1, -61/{t-1}, -91/{t}}{%
        \draw[vbrace,color=green!60!black] (0.85,\top) -- ++(0,-29)
        node[midway,right=6pt,font=\scriptsize,color=green!60!black] {sub-ladder $\pt$};
    }
    \foreach \y in {-20,-60,-90,-120} {
    \draw[gray] (-0.5,\y) -- (0.5,\y);
    \foreach \i in {1,2,3,4,5} {
      \fill[gray] ({-0.5 + \i*1/6},\y) circle[radius=0.04cm];
    }
  }
    
\end{tikzpicture}
  \caption{A sketch of the arc-gadget for some edge $uv \in A(D)$. 
  Only the orange token-bearing nodes contain tokens.
  Here, the $x$ and $y$ values for the first supply tuple for $uv$ are $2$ and $4$, respectively. 
  The values for the supply tuple $t-1$ for $uv$ are $3$ and $1$, respectively, and those for the supply tuple $t$ for $uv$ are $1$ and $2$, respectively.
  }
  \label{fig:logic-solution-discovery-arc-gadget-bandwidth-hardness}
\end{figure}
 
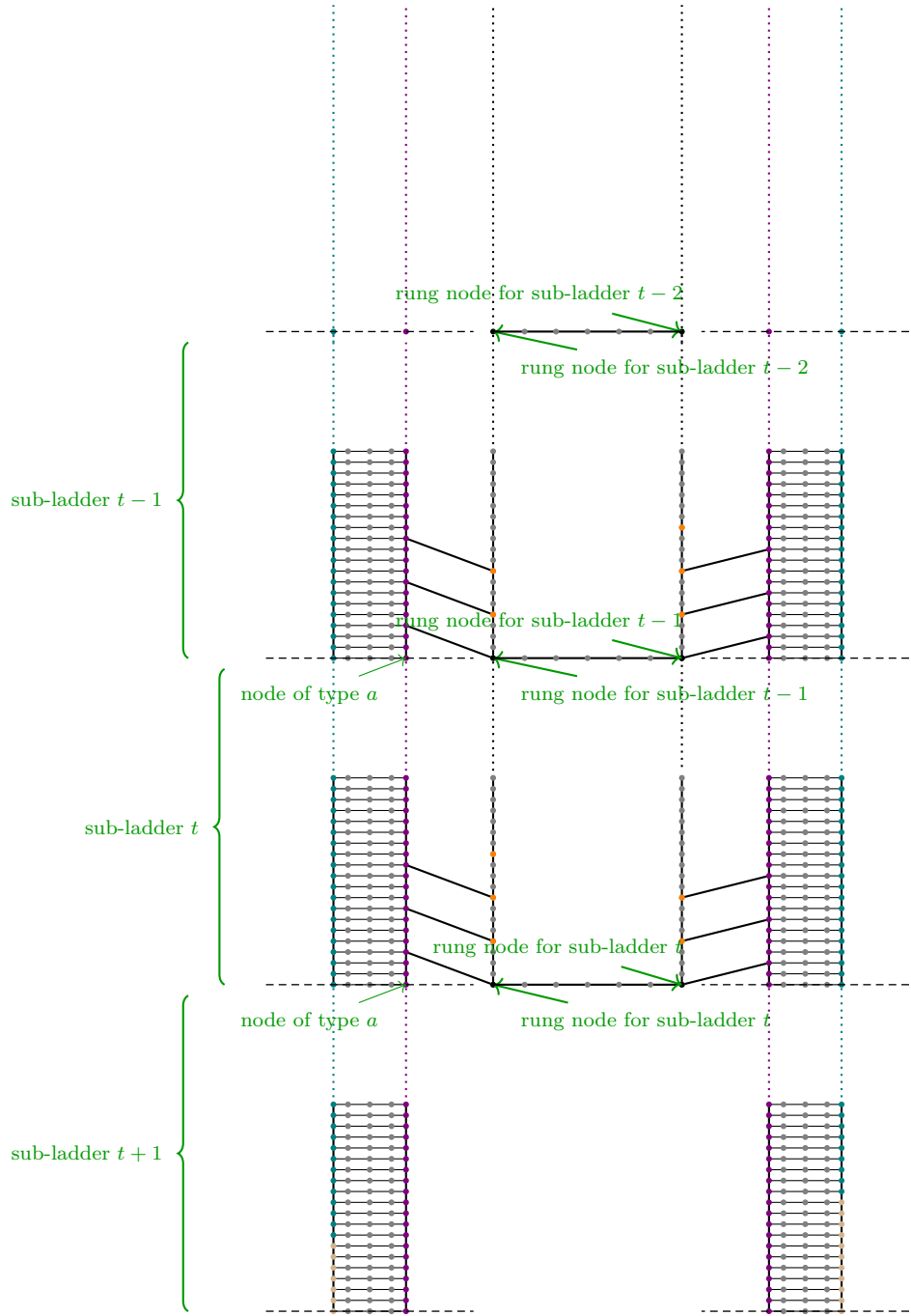
\begin{figure}[h]
  \centering
  \begin{tikzpicture}[
      x=0.15cm,
      y=2cm,
      rotate=90,
      baseline=(current bounding box.center)
    ]
    \draw[thick] (60,0.6) -- (62,0);
    \draw[thick] (64,0.6) -- (66,0);
    \draw[thick] (68,0.6) -- (70,0);
    \draw[thick] (90,0.6) -- (92,0);
    \draw[thick] (94,0.6) -- (96,0);
    \draw[thick] (98,0.6) -- (100,0);
    \foreach \x in {30,...,49}{
      \draw (\x,0) -- (\x,-0.25) -- (\x,-0.5);
      \foreach \y in {-0.10,-0.25,-0.40}
        \fill[gray] (\x,\y) circle[radius=0.04cm];
    }
    \foreach \x in {60,...,79}{
      \draw (\x,0) -- (\x,-0.25) -- (\x,-0.5);
      \foreach \y in {-0.10,-0.25,-0.40}
        \fill[gray] (\x,\y) circle[radius=0.04cm];
    }
    \foreach \x in {90,...,109}{
      \draw (\x,0) -- (\x,-0.25) -- (\x,-0.5);
      \foreach \y in {-0.10,-0.25,-0.40}
        \fill[gray] (\x,\y) circle[radius=0.04cm];
    }
    \draw[thick] (30,0) -- (49,0);
    \draw[thick] (60,0) -- (79,0);
    \draw[dotted,violet,thick] (49,0) -- (60,0);
    \draw[dotted,violet,thick] (79,0) -- (90,0);
    \draw[thick] (90,0) -- (109,0);
    \draw[dotted,violet,thick] (109,0) -- (120,0);
    \draw[dotted,violet,thick] (120,0) -- (150,0);
    \foreach \x in {30,...,49}{
      \fill[violet] (\x,0) circle[radius=0.04cm];
    }
    \foreach \x in {60,...,79}{
      \fill[violet] (\x,0) circle[radius=0.04cm];
    }
    \foreach \x in {90,...,109}{
      \fill[violet] (\x,0) circle[radius=0.04cm];
    }
    \fill[violet] (120,0) circle[radius=0.04cm];
    \foreach \x in {60,...,78}   \draw[thick] (\x,0.6) -- ++(1,0);
    \foreach \x in {90,...,108}  \draw[thick] (\x,0.6) -- ++(1,0);
    \draw[dotted,thick] (80,0.6) -- (90,0.6);
    \draw[dotted,thick] (109,0.6) -- (120,0.6);
    \draw[dotted,thick] (119,0.6) -- (150,0.6);
    \foreach \x in {60,...,79}{
      \fill[gray] (\x,0.6) circle[radius=0.04cm];
    }
    \foreach \x in {90,...,109}{
      \fill[gray] (\x,0.6) circle[radius=0.04cm];
    }
    \draw[thick] (30,-0.5) -- (49,-0.5);
    \draw[thick] (60,-0.5) -- (79,-0.5);
    \draw[dotted,teal,thick] (49,-0.5) -- (60,-0.5);
    \draw[dotted,teal,thick] (79,-0.5) -- (90,-0.5);
    \draw[thick] (90,-0.5) -- (109,-0.5);
    \draw[dotted,teal,thick] (109,-0.5) -- (120,-0.5);
    \draw[dotted,teal,thick] (120,-0.5) -- (150,-0.5);
    \foreach \x in {30,...,40}{
        \fill[brown] (\x,-0.5) circle[radius=0.04cm];
    }
    \foreach \x in {41,...,49}{
        \fill[teal] (\x,-0.5) circle[radius=0.04cm];
    }
    \foreach \x in {60,...,79}{
        \fill[teal] (\x,-0.5) circle[radius=0.04cm];
    }
    \foreach \x in {90,...,109}{
        \fill[teal] (\x,-0.5) circle[radius=0.04cm];
    }
    \fill[teal] (120,-0.5) circle[radius=0.04cm];
    \node[left=12pt,font=\scriptsize]  (lab2a) at (30,0) {};
    \node[left=12pt,font=\scriptsize] (lab1a) at (33,0) {};
    \node[left=12pt,font=\scriptsize] (lab1aa) at (39,0) {};
    \node[below left=12pt,font=\scriptsize] (lab1b) at (60,0) {};
    \node[above right=50pt,font=\scriptsize] (lab1c) at (62,0) {};
    \node[below left=12pt,font=\scriptsize] (labGray90) at (90,0)
      {};
    \foreach \x/\txt in {90/{}}
      {
    \fill[black] (\x,0.6) circle[radius=0.04cm];
        \node[above left=10pt,font=\scriptsize,color=black] (ylab\x)
          at (\x,0.6) {\txt};
      }
    \fill[black] (120,0.6) circle[radius=0.04cm];
    \node[above left=10pt,font=\scriptsize,color=green!60!black] (ylab120)
      at (120,0.4) {rung node for sub-ladder $t-2$};
    \draw[->,green!60!black,thick] (ylab120) -- (120,0.6);
    \foreach \x in {64,68,94,98,102}
      \fill[orange] (\x,0.6) circle[radius=0.04cm];
    \node[left=10pt,font=\scriptsize,color=orange] (olab68)
      at (68,0.6) {};
    \node[left=10pt,font=\scriptsize,color=orange] (olab102)
      at (102,0.6) {};
    \draw[thick] (60,1.9) -- (63,2.5);
    \draw[thick] (64,1.9) -- (67,2.5);
    \draw[thick] (68,1.9) -- (71,2.5);
    \draw[thick] (90,1.9) -- (93,2.5);
    \draw[thick] (94,1.9) -- (97,2.5);
    \draw[thick] (98,1.9) -- (101,2.5);
    \foreach \x in {30,...,49}{
      \draw (\x,2.5) -- (\x,2.75) -- (\x,3);
      \foreach \y in {2.6,2.75,2.9}
        \fill[gray] (\x,\y) circle[radius=0.04cm];
    }
    \foreach \x in {60,...,79}{
      \draw (\x,2.5) -- (\x,2.75) -- (\x,3);
      \foreach \y in {2.6,2.75,2.9}
        \fill[gray] (\x,\y) circle[radius=0.04cm];
    }
    \foreach \x in {90,...,109}{
      \draw (\x,2.5) -- (\x,2.75) -- (\x,3);
      \foreach \y in {2.6,2.75,2.9}
        \fill[gray] (\x,\y) circle[radius=0.04cm];
    }
    \draw[thick] (30,2.5) -- (49,2.5);
    \draw[thick] (60,2.5) -- (79,2.5);
    \draw[dotted,violet,thick] (49,2.5) -- (60,2.5);
    \draw[dotted,violet,thick] (79,2.5) -- (90,2.5);
    \draw[thick] (90,2.5) -- (109,2.5);
    \draw[dotted,violet,thick] (109,2.5) -- (119,2.5);
    \draw[dotted,violet,thick] (120,2.5) -- (150,2.5);
    \foreach \x in {30,...,49}{
        \fill[violet] (\x,2.5) circle[radius=0.04cm];
    }
    \foreach \x in {60,...,79}{
        \fill[violet] (\x,2.5) circle[radius=0.04cm];
    }
    \foreach \x in {90,...,109}{
        \fill[violet] (\x,2.5) circle[radius=0.04cm];
    }
    \fill[violet] (120,2.5) circle[radius=0.04cm];
    \foreach \x in {60,...,78}   \draw[thick] (\x,1.9) -- ++(1,0);
    \foreach \x in {90,...,108}  \draw[thick] (\x,1.9) -- ++(1,0);
    \draw[dotted,thick] (79,1.9) -- (90,1.9);
    \draw[dotted,thick] (109,1.9) -- (120,1.9);
    \draw[dotted,thick] (120,1.9) -- (150,1.9);
    \foreach \x in {60,...,79}{
      \fill[gray] (\x,1.9) circle[radius=0.04cm];
    }
    \foreach \x in {90,...,109}{
      \fill[gray] (\x,1.9) circle[radius=0.04cm];
    }
    \draw[thick] (30,3) -- (49,3);
    \draw[thick] (60,3) -- (79,3);
    \draw[dotted,teal,thick] (49,3) -- (60,3);
    \draw[dotted,teal,thick] (79,3) -- (90,3);
    \draw[thick] (90,3) -- (109,3);
    \draw[dotted,teal,thick] (109,3) -- (120,3);
    \draw[dotted,teal,thick] (119,3) -- (150,3);
    \foreach \x in {30,...,36}{
        \fill[brown] (\x,3) circle[radius=0.04cm];
    }
    \foreach \x in {37,...,49}{
      \fill[teal] (\x,3) circle[radius=0.04cm];
    }
    \foreach \x in {60,...,79}{
      \fill[teal] (\x,3) circle[radius=0.04cm];
    }
    \foreach \x in {90,...,109}{
      \fill[teal] (\x,3) circle[radius=0.04cm];
    }
    \fill[teal] (120,3) circle[radius=0.04cm];
    \node[right=12pt,font=\scriptsize] (mLab1a) at (30,2.5)
      {};
    \node[below right=12pt,font=\scriptsize] (mLabGray90) at (90,2.5)
      {};
    \node[below right=12pt,font=\scriptsize] (mLabGray91) at (60,2.5)
      {};
    \foreach \x/\txt in {60/{rung node for sub-ladder $t$}, 90/{rung node for sub-ladder $t-1$}}
      {
        \fill[black] (\x,1.9) circle[radius=0.04cm];
        \node[below right=10pt,font=\scriptsize,color=green!60!black] (yMirror\x)
          at (\x,1.9) {\txt};
          \draw[->,green!60!black,thick] (yMirror\x) -- (\x,1.9);
      }
    \node[below left=10pt,font=\scriptsize,color=green!60!black] (a60)
          at (60,2.5) {node of type $a$};
          \draw[->,green!60!black] (a60) -- (60,2.5);
    \node[below left=10pt,font=\scriptsize,color=green!60!black] (a90)
          at (90,2.5) {node of type $a$};
          \draw[->,green!60!black] (a90) -- (90,2.5);
    \fill[black] (120,1.9) circle[radius=0.04cm];
    \node[below right=10pt,font=\scriptsize,color=green!60!black] (yMirror120)
      at (120,1.9) {rung node for sub-ladder $t-2$};
    \draw[->,green!60!black,thick] (yMirror120) -- (120,1.9);
    \foreach \x in {64,68,72,94,98}
      \fill[orange] (\x,1.9) circle[radius=0.04cm];
    \node[above right=10pt,font=\scriptsize,color=orange] (oMirror68)
      at (68,1.9) {};
    \foreach \x in {60,90,120}{  
    \draw[thick] (\x,0.6) -- (\x,1.9);
     \foreach \k in {1,...,5}{
      \pgfmathsetmacro{\yy}{0.6 + (1.9-0.6)*\k/6}
      \fill[gray] (\x,\yy) circle[radius=0.04cm];
     }
    }
    \foreach \x/\txt in {60/{rung node for sub-ladder $t$}, 90/{rung node for sub-ladder $t-1$}}
      {
    \fill[black] (\x,0.6) circle[radius=0.04cm];
        \node[above left=10pt,font=\scriptsize,color=green!60!black] (ylab\x)
          at (\x,0.4) {\txt};
          \draw[->,green!60!black,thick] (ylab\x) -- (\x,0.6);
      }
    \foreach \y in {30,60,90,120}{
    \draw[cut mark] (\y,-1) -- (\y,0.5);
    }
    \foreach \y in {30,60,90,120}{
    \draw[cut mark] (\y,3.5) -- (\y,2);
    }
    \draw[vbrace,color=green!60!black] (30,4) -- (59,4)
    node[midway,left=6pt,font=\scriptsize,color=green!60!black] {sub-ladder $t+1$};
    \draw[vbrace,color=green!60!black] (60,3.75) -- (89,3.75)
    node[midway,left=6pt,font=\scriptsize,color=green!60!black] {sub-ladder $t$};
    \draw[vbrace,color=green!60!black] (90,4) -- (119,4)
    node[midway,left=6pt,font=\scriptsize,color=green!60!black] {sub-ladder $t-1$};   
  \end{tikzpicture}
  \caption{A sketch of how the connected side rail for some vertex $u \in V(D)$ and the connected side rail for some vertex $v \in V(D)$ connect to an arc-gadget for arc $uv \in A(D)$.
  Here $v$ is assigned type $d$ in the neighborhood of $u$, and $u$ is assigned type $c$ in the neighborhood of $v$.
  The thin dotted lines indicate additional nodes that are omitted for clarity.}
  \label{fig:logic-solution-discovery-graph-bandwidth-hardness}
\end{figure}

\end{document}